\documentclass[acmsmall,10pt]{acmart}
\usepackage{amsmath}

\usepackage{latexsym}
\usepackage{proof}

\usepackage{amsfonts}
\usepackage{graphicx}
\usepackage{booktabs}
\usepackage{tabularx}
\PassOptionsToPackage{colorlinks=true,citecolor=blue,linkcolor=blue,urlcolor=blue}{hyperref}

\usepackage{color}
\usepackage{multirow}
\usepackage{wrapfig,lipsum,booktabs}
\usepackage{semantic}
\usepackage{isabelle,isabellesym}
\usepackage{subcaption}
\usepackage{listings}
\usepackage{xcolor}
\usepackage{adjustbox}
\usepackage{lineno}

\newcommand{\specrg}[1]{{\color{blue}{#1}}}
\newcommand{\speccomment}[1]{{\color{gray}{#1}}}
\definecolor{proofcolor}{rgb}{0,0.45,0}

\newcommand{\sectprefix}{{Section}}

\newcommand{\figprefix}{{Fig.}}

\newcommand{\tableprefix}{{Table}}

\newcommand{\equationprefix}{{Equation}}
\newcommand{\lemmaprefix}{Lemma}
\newcommand{\theoremprefix}{Theorem}
\newcommand{\defprefix}{Definition}

\usepackage{array}
\newcommand{\PreserveBackslash}[1]{\let\temp=\\#1\let\\=\temp}
\newcolumntype{C}[1]{>{\PreserveBackslash\centering}p{#1}}
\newcolumntype{R}[1]{>{\PreserveBackslash\raggedleft}p{#1}}
\newcolumntype{L}[1]{>{\PreserveBackslash\raggedright}p{#1}}

\isabellestyle{sl} %tt,it,sl,literal

\let\endisabellecode=\endisabellecode

\newcommand{\llbrace}{\lbrace\mkern-4.5mu\mid}
\newcommand{\rrbrace}{\mid\mkern-4.5mu\rbrace}

\newcommand{\slang}{PiCore} %{$\texttt{SPECURITY}$}

\newcommand{\cmdfont}[1]{\textbf{#1}} %{\mathsf{#1}}

\newcommand{\defi}{\triangleq}

\newcommand{\cmdwhile}[2]{\cmdfont{While} \ {#1} \ {#2}}
\newcommand{\cmdawait}[2]{\cmdfont{Await} \ {#1} \ {#2}}

\newcommand{\event}[1]{\cmdfont{Event} \ {#1}}
\newcommand{\anonevt}[1]{\lfloor{#1}\rfloor}

\newcommand{\evtsystwo}[2]{\{{#1}, \ ... ,\ {#2}\}}
\newcommand{\evtsysdef}{\evtsystwo{\symbEvt_0}{\symbEvt_n}}
\newcommand{\evtseq}[2]{{#1}\triangleright{#2}}

\newcommand{\parsysc}{\symbCore \rightarrow \symbevtsys}
\newcommand{\stmtirq}[2]{{#1} $\blacktriangleright$ {#2}}
\newcommand{\stmtatom}[1]{\textbf{ATOM}\ {#1}\ \textbf{END}}
\newcommand{\stmtawait}[2]{\textbf{AWAIT}\ {#1}\ \textbf{THEN}\ {#2}\ \textbf{END}}
\newcommand{\stmtevent}[5]{\textbf{EVENT}\ {#1}\ [ {#2} ] @ {#3}\ \textbf{WHEN}\ {#4}\ \textbf{THEN}\ {#5}\ \textbf{END}}

\newcommand{\eventsystem}[1]{\textbf{ESYS}\ {#1} \equiv \evtsystwo{\symbEvt_0}{\symbEvt_n}}

\newcommand{\symbprog}{P}
\newcommand{\symbbexp}{b}
\newcommand{\symbEvt}{\mathcal{E}}
\newcommand{\symbevt}{ev}
\newcommand{\symbevtbd}{\alpha}
\newcommand{\symbevtsys}{\mathcal{S}}
\newcommand{\symbpes}{\mathcal{PS}}
\newcommand{\symbState}{S}
\newcommand{\symbstate}{s}

\newcommand{\symbevtctx}{x}
\newcommand{\symbConf}{\Delta}
\newcommand{\symbconf}{\mathcal{C}}
\newcommand{\symbpcomp}{c}

\newcommand{\symbact}{t}
\newcommand{\symbCore}{\mathcal{K}}
\newcommand{\symbcore}{\kappa}
\newcommand{\symbactk}{\delta}

\newcommand{\actk}[2]{{#1}@{#2}}

\newcommand{\symbDomain}{\mathcal{D}}
\newcommand{\symbdomain}{d}

\newcommand{\symbSM}{\mathcal{M}}
\newcommand{\symbAction}{A}
\newcommand{\symbaction}{a}
\newcommand{\symbactions}{as}
\newcommand{\symbspec}{\sharp}
\newcommand{\transenv}{\Sigma \vdash}
\newcommand{\progenv}{\Sigma}
 
\newcommand{\trant}[2]{\stackrel{{#1}}{\longrightarrow}_{#2}}
\newcommand{\trans}[3]{\stackrel{\actk{#1}{#2}}{\longrightarrow}_{#3}}

\newcommand{\tranb}[1]{\stackrel{{#1}}{\longrightarrow}_{\square}}

\newcommand{\tranes}[2]{\trans{#1}{#2}{es}}

\newcommand{\tranenv}[1]{\stackrel{env}{\longrightarrow}_{#1}}

\newcommand{\tranenvpes}{\tranenv{pes}}

\newcommand{\myinfer}[2]{
\begin{tabular}{l}
  \textsc{[#1]} \\
  {#2}
\end{tabular}}

\newcommand{\dsim}[1]{\stackrel{{#1}}{\sim}}

\newcommand{\interf}{\leadsto}

\newcommand{\reachablef}{\mathcal{R}}
\newcommand{\compfun}{\Psi}
\newcommand{\symbcomp}{\varpi}

\newcommand{{\compstps}}{\Psi_\symbpes}
\newcommand{{\compstes}}{\Psi_\symbevtsys}
\newcommand{{\compste}}{\Psi_\symbEvt}
\newcommand{{\compstp}}{\Psi_\symbprog}

\newcommand{\rgcond}[4]{\langle #1, #2, #3, #4 \rangle}
\newcommand{\rgconddefault}{\rgcond{pre}{R}{G}{pst}}
\newcommand{\RGSAT}[2]{\progenv \models {#1} \ \mathbf{sat} \ {#2}}
\newcommand{\rgsat}[2]{\progenv \vdash {#1} \ \mathbf{sat} \ {#2}}

\newcommand{\stset}[1]{\llbrace {#1} \rrbrace}

\newcommand{\isactrlenum}{$\blacktriangleright$}

\makeatletter
\newcommand{\superimpose}[2]{%
  {\ooalign{$#1\@firstoftwo#2$\cr\hfil$#1\@secondoftwo#2$\hfil\cr}}}
\makeatother
\newcommand{\ninterf}{\mathrel{\mathpalette\superimpose{{\slash}{\leadsto}}}}

\def\BibTeX{{\rm B\kern-.05em{\sc i\kern-.025em b}\kern-.08emT\kern-.1667em\lower.7ex\hbox{E}\kern-.125emX}}

\setcopyright{acmcopyright}
\acmJournal{TOSEM}
\acmYear{2020}\acmVolume{1}\acmNumber{1}\acmArticle{1}\acmMonth{1}
\acmDOI{10.1145/1122445.1122456}

\begin{document}

%
% The "title" command has an optional parameter, allowing the author to define a "short title" to be used in page headers.
%\title{A Parametric Rely-guarantee Reasoning Framework for Event-based Systems and its Applications}
%\title{PiCore: A Parametric Rely-guarantee Reasoning Framework for Concurrent Systems}

%\title[Rely-guarantee Reasoning about Concurrent Memory Management]
%{Rely-guarantee Reasoning about Safety and Security of Concurrent Memory Management}
\title[Rely-guarantee Reasoning about Concurrent Memory Management]
{Rely-guarantee Reasoning about Concurrent Memory Management: Correctness, Safety and Security}

\titlenote{This article is an extended version of \cite{zhao19cav}. %\\
%This work has been supported by the National Natural Science Foundation of China (NSFC) under the Grant No. 61872016. This work has been partially supported by the National Satellite of Excellence in Trustworthy Software Systems (Award No. NRF2018NCR-NSOE003), and award NRF Investigatorship NRFI06-2020-0022,funded by NRF Singapore under National Cyber-security R\&D (NCR) programme. This work has been partially supported by the Ministry of Education, Singapore, under its Academic Tier-2 Research Fund (MOE2018-T2-1-068).
}. 

%
% The "author" command and its associated commands are used to define the authors and their affiliations.
% Of note is the shared affiliation of the first two authors, and the "authornote" and "authornotemark" commands
% used to denote shared contribution to the research.
\author{Yongwang Zhao}
%\email{zhaoyw@buaa.edu.cn}
\affiliation{%
  \institution{Zhejiang University}
  \department{School of Cyber Science and Technology, College of Computer Science and Technology}
  \streetaddress{38 Zheda Road}
  \city{Hangzhou}
  \state{Zhejiang}
  \country{China}
  \postcode{310007}
}

\author{David San\'{a}n}
\affiliation{%
  \institution{Nanyang Technological University}
  \department{School of Computer Science and Engineering}
  \streetaddress{50 Nanyang Avenue}
  \city{Singapore}
  \country{Singapore}
  \postcode{639798}
}
%
% By default, the full list of authors will be used in the page headers. Often, this list is too long, and will overlap
% other information printed in the page headers. This command allows the author to define a more concise list
% of authors' names for this purpose.
\renewcommand{\shortauthors}{Yongwang Zhao and David San\'{a}n}

%
% The abstract is a short summary of the work to be presented in the article.
\begin{abstract}
Formal verification of concurrent operating systems (OSs) is challenging, in particular the verification of the dynamic memory management due to its complex data structures and allocation algorithm. An incorrect specification and implementation of the memory management may lead to system crashes or exploitable attacks. This article presents the first formal specification and mechanized proof of a concurrent memory management for a real-world OS concerning a comprehensive set of properties, including functional correctness, safety and security. 
To achieve the highest assurance evaluation level, we develop a fine-grained formal specification of the Zephyr RTOS buddy memory management, which closely follows the C code easing validation of the specification and the source code. 
The rely-guarantee-based compositional verification technique has been enforced over the formal model. To support formal verification of the security property, we extend our rely-guarantee framework {\slang} by a compositional reasoning approach for integrity. 
Whilst the security verification of the design shows that it preserves the integrity property, the verification of the functional properties shows several problems. These verification issues are translated into finding three bugs in the C implementation of Zephyr, after inspecting the source code corresponding to the design lines breaking the properties.

\end{abstract}




%
% Keywords. The author(s) should pick words that accurately describe the work being
% presented. Separate the keywords with commas.

%
% A "teaser" image appears between the author and affiliation information and the body
% of the document, and typically spans the page.
%%  \includegraphics[width=\textwidth]{sampleteaser}
%%  \caption{Seattle Mariners at Spring Training, 2010.}
%%  \Description{Enjoying the baseball game from the third-base seats. Ichiro Suzuki preparing to bat.}

%
% This command processes the author and affiliation and title information and builds
% the first part of the formatted document.
\maketitle

\section{Introduction} % to the position of 1.5 page
\label{sect:intro}

%background

\subsection{Context and Motivation}
Operating systems, and in particular Real Time Operating Systems (RTOS), are a fundamental component of critical systems. Correctness, safety and security of systems highly depend on the system's underlying OS. As a key component of OSs, the memory management provides ways to dynamically allocate portions of memory to programs at their request, and to free them for reuse when no longer needed. The buddy memory allocation technique \cite{Knowlton65} is a memory allocation algorithm that splits memory into halves or quarters to try to satisfy a memory request in a best-fit manner. Buddy memory allocation has been widely applied in OS kernels (e.g. Linux kernel and Zephyr RTOS \footnote{https://www.zephyrproject.org/}). 
Since program variables and data are stored in the allocated memory, correct, safe and secure memory management is extremely critical for the whole system. An incorrect specification and implementation of the memory management may lead to system crashes and exploitable attacks.

Formal verification has been intensively conducted on OS kernels in recent years \cite{Klein09b,zhao17}. 
Most of these efforts focus on sequential OS kernels and assume that there is no in-kernel concurrency (e.g. seL4 \cite{Klein14}). 
Concurrent kernels allow interleaved execution of kernel/user modules due to user thread preemption, I/O interrupts and execution in multicore architectures. 
Some related work studying the building and verification of concurrent kernels has been covered in \cite{Gu16,Chen16,Xu16}. However, formal verification of concurrent OS kernels still present several open challenges. For instance, formal verification in \cite{Chen16} concerns kernels with device drivers using a verification framework that does not support preemptive and multicore concurrency. As a consequence it is only possible to verify interrupt handlers for device drivers not sharing data with and non-handler kernel code. 
%Formal verification of preemptive OS kernels in \cite{Xu16} permits explicit invocation of schedulers in system services, which makes the kernel interleaving at certain program locations. 

Formal verification of OS memory management has been studied in CertiKOS \cite{Vaynberg12,Gu16}, seL4 \cite{Klein04,Klein09a}, Verisoft \cite{Alkassar08}, and in the hypervisors from \cite{Blan15,Bolig16}. Algorithms and implementations of dynamic memory allocation have been formally specified and verified in an extensive number of works \cite{Yu03,Fang17a,Marti06,Su16,Fang17b,Fang18}. Concurrency is only studied in \cite{Blan15,Gu16} considering much simpler data structures and algorithms than our work. Moreover, only \cite{Fang18} studies the buddy memory allocation considering very abstract data structures. 
%Finally, security of concurrent OS kernels is another open challenge \cite{Murray12,Costanzo16}. 
Finally, formal verification of security properties, e.g. integrity and confidentiality, is still challenging for concurrent OS kernels \cite{Murray12,Costanzo16} and has not been studied for concurrent memory management before. Confidentiality refers to protecting information from being accessed by unauthorized parties. In general, it is not preserved by memory allocation at OS level. For instance, in Zephyr the memory release by a thread does not clear the allocated memory and the information may be accessed by other threads. Thus, we consider the integrity of concurrent memory management in this article, which means the allocated memory of a thread cannot be altered by other threads. 

%As a typical security property, the integrity of concurrent memory management has not been studied before. 

%%\qmark{
%%formal verification of implementation of concurrent OSs: difficult. we dont have parser. So we do formal verification at low-level design, very close to the code. 
%%
%%CC evaluation. 
%%}

This article concentrates on the formal specification and formal verification of functional correctness, safety and security properties on the concurrent buddy memory management of Zephyr. \emph{Zephyr} is an open-source state of the art RTOS managed by the Linux Foundation for connected, resource-constrained devices, and built with security and safety design in mind. It has been deployed in IoT gateways, safety shoes, and smart watches, etc. It uses a buddy memory allocation algorithm optimized for RTOS and is completely concurrent in the kernel that allows multiple threads to concurrently manipulate shared memory pools with fine-grained locking. We apply the {\slang} rely-guarantee framework \cite{zhao19fm} to the verification of Zephyr. The compositionality of rely-guarantee allows makes possible to handle the complexity of the memory allocation algorithm used in Zephyr, and of its data structures.

\subsection{Challenges}

Formal verification of concurrent memory management, in particular the buddy memory allocation in Zephyr, is a challenging work.
\begin{enumerate}
\item \emph{Fine-grained concurrency of the execution of the memory services and of the shared memory structure}: On one hand, thread preemption and interruption make the kernel execution of memory services to be concurrent. 
Memory allocation usually uses fine-grained locking for threads. When manipulating a shared memory pool, memory services in a thread lock the pool by disabling interruptions inside critical sections as small as possible.  
On the  other hand, memory pools are shared by threads in a fine-grained parts of its structure. When a thread is splitting a memory block into smaller ones to get a suitable block size, another thread may be coalescing partner blocks in the same pool into a larger one. That is, a thread is manipulating a block sub-tree of a memory pool, meanwhile another thread is manipulating another block sub-tree of the same pool.

\item \emph{Complex data structure and algorithm of buddy memory management}: to achieve high performance, data structures and algorithms in Zephyr are laid out in a complex manner. 
First, the buddy memory allocation can split large blocks into smaller ones, allowing blocks of different sizes to be allocated and released efficiently while limiting memory fragmentation concerns. Seeking performance, Zephyr uses a multi-level structure where each level has a bitmap and a linked list of free memory blocks. The levels of bitmaps actually form a forest of quad trees of bits. Memory addresses are used as a reference to memory blocks, so the algorithm has to deal with address alignment and computation concerning the block size at each level, increasing the complexity of its verification. 
Second, the allocation algorithm supports various temporal modes. If a block of the desired size is unavailable, a thread can optionally wait for one to become available. There are three different modes for waiting threads: waiting forever, waiting for a time out, and no wait. Before a thread changes its state to waiting, it invokes rescheduling in the allocation service and thus is preempted by ready threads. The system must guarantee that each service eventually returns from each of the waiting modes. 

\item \emph{Verification of safety and security of concurrent memory management is difficult}: first, a complex algorithm and data structures implies as well complex invariants over them, that the formal model must preserve as functional safety properties. These invariants have to guarantee the multi-level well-shaped bitmaps and their consistency to multi-level free lists. To prevent memory leaks and block overlapping, a precise reasoning shall keep track of both numerical and shape properties. Second, as a security property we verify the integrity of allocated memory blocks among threads, i.e. one thread can not modify blocks allocated by other threads. In this context, formal verification of integrity on OS kernels needs to consider the system events (e.g. kernel services, interrupt handlers), and hence integrity is seen as a property of state-event based information-flow security (IFS) \cite{Murray12,Costanzo16}. Works on state-event IFS \cite{rushby92,Oheimb04,Murray12} tackle sequential systems and can not be applied either on the verification concurrent OS kernels. Although there are some work related to the verification of concurrent IFS~\cite{Mantel11,Murray16,Murray18}, these focus on language-based IFS which can not be used on the verification of integrity for a concurrent operating system. Therefore the verification of state-event IFS is still an open challenge.

\end{enumerate}

\subsection{Approach and Contributions}
The safety and security properties in this article concern the small steps inside memory services that must be preserved by any internal step of the services. For instance, in the case of Zephyr RTOS, a safety property is that memory blocks do not overlap each other even during internal steps of the allocation and release services. It is therefore necessary to find a verification approach that allows to reason at such fine-grained detail.

In this article, we apply the rely-guarantee reasoning technique to verify the Zephyr concurrent memory management. 
This work uses {\slang} \cite{zhao19fm}, a two-level event-based rely-guarantee framework in Isabelle/HOL for the specification and verification of concurrent reactive systems (CRS). 
{\slang} has support for concurrent OSs features like modelling shared-variable concurrency of multiple threads, interruptable execution of handlers, self-suspending threads, and rescheduling. 

{\slang} separates the specification and verification at two levels. 
The top level introduces the notion of ``events'' into the rely-guarantee method for system reactions. This level defines the events composing a system, and how and when they are triggered. It supports reactive semantics of interrupt handlers (e.g. kernel services, scheduler) in OSs, which makes formal specification of OSs much simpler than those represented by pure programs (e.g. in \cite{Andronick15}). 
The second level focuses on the specification and reasoning of the behaviour of the events composing the first level. {\slang} parametrizes the second level using a rely-guarantee interface, allowing to easily reuse existing rely-guarantee frameworks of imperative programs. 
{\slang} concurrent constructs allow the specification of Zephyr multi-thread interleaving, fine-grained locking, and thread preemption. Compositionality of rely-guarantee makes feasible to prove the functional correctness of Zephyr and invariants over its data structures.

%We provides a formal definition of memory pools in the buddy memory allocation. 
In this article, we first formalize the data structures of Zephyr memory pools in Isabelle/HOL, and we analyze its structural properties. The properties clarify the constraints and consistency of quad trees, free block lists, memory pool configuration, and waiting threads. These properties conforms the safety of the memory management. They are defined as invariants for which its preservation under the execution of services is formally verified. The set of properties is comprehensive for buddy memory allocation since we can derive the memory separation property at the memory-block level as discussed below from them.

Second, we consider memory separation as the security of Zephyr at two levels: memory-block level and thread level. At the memory-block level, memory separation ensures that the memory blocks of a memory pool cover the whole memory address of the pool, but do not overlap each other. This property is necessary to prevent memory leaks and can be derived from well-shaped properties of quad trees defined in the invariants. For memory separation at the thread level, we consider the aforementioned memory integrity among threads. To tackle this, we extend {\slang} with a compositional reasoning approach for integrity on event-based concurrent systems. This approach redefines the concept of integrity in terms of fine-grain semantics and it uses rely-guarantee, as in the core of \slang, for reasoning on the integrity of the events by means of observable equivalence among threads.

Third, together with the formal verification of Zephyr, we aim at the highest evaluation assurance level (EAL 7) of Common Criteria (CC) \cite{cc}, which was declared in the last year as the candidate standard for security certification by the Zephyr project. Therefore, we develop a fine-grained low level formal specification of a buddy memory management. The specification closely follows the Zephyr C code, and thus is able to do the \emph{code-to-spec} review required by the EAL 7 evaluation, covering all the data structures and imperative statements present in the implementation. 
The functional correctness of the memory management is specified by pre and post conditions of each service and compositionally proved by the rely-guarantee proof system of {\slang}. 
%For simplicity, we use abstract data types in the specification, such as \emph{list} for free block lists and bitmaps. It is enough to check bound violations.

Finally, we enforce the formal verification of functional correctness, invariant preservation, and memory separation by using the extended rely-guarantee proof system of {\slang}. It supports total correctness for loops where fairness does not need to be considered. %Non-violations of array bound are also verified in small steps of the proof. 
The formal verification shows the preservation of memory integrity, however, revealed three functional and safety bugs in the C code: an incorrect block split, an incorrect return from the kernel services, and non-termination of a loop. Two of them are critical and have been repaired in the latest release of Zephyr. The third bug causes nontermination of the allocation service when trying to allocate a block of a larger size than the maximum allowed. %\qmark{The formal verification also found security flaws. ...}

%The security flaw involves not checking tampering of the block information structure, leading to memory fragmentation. 

To the best of our knowledge, this article presents the first formal specification and mechanized proof of correctness, safety and security for concurrent memory allocation of a realistic operating system. 
The formal specification and proofs in this article are completely developed in Isabelle/HOL. All the Isabelle/HOL sources are available at {\color{ACMBlue}\url{https://lvpgroup.github.io/tosem2021/}}. We summarize the main contributions of this article as follows.

\begin{enumerate}
\item A comprehensive set of critical properties for concurrent buddy memory management, including functional correctness, safety by invariants, and security by memory separation. In particular, we clarify the constraints and consistency of the complicated structure of buddy memory pools. 
%\item The safety and security properties and the first verified formal specification for concurrent buddy memory allocations. 

\item The first compositional reasoning approach for state-event based integrity, its application on a concurrent OS kernel, and its formal proof for the Zephyr concurrent memory management. 

\item The first verified formal specification for concurrent buddy memory allocations which corresponds to the low-level design specification in CC EAL 7 evaluation. 

\item Critical bugs founds on the functional correctness and safety in Zephyr C code, which have been repaired in the latest release of Zephyr. 

\end{enumerate}

\subsection{Roadmap}

\begin{figure}%[t]
\begin{center}
\includegraphics[width=3.6in]{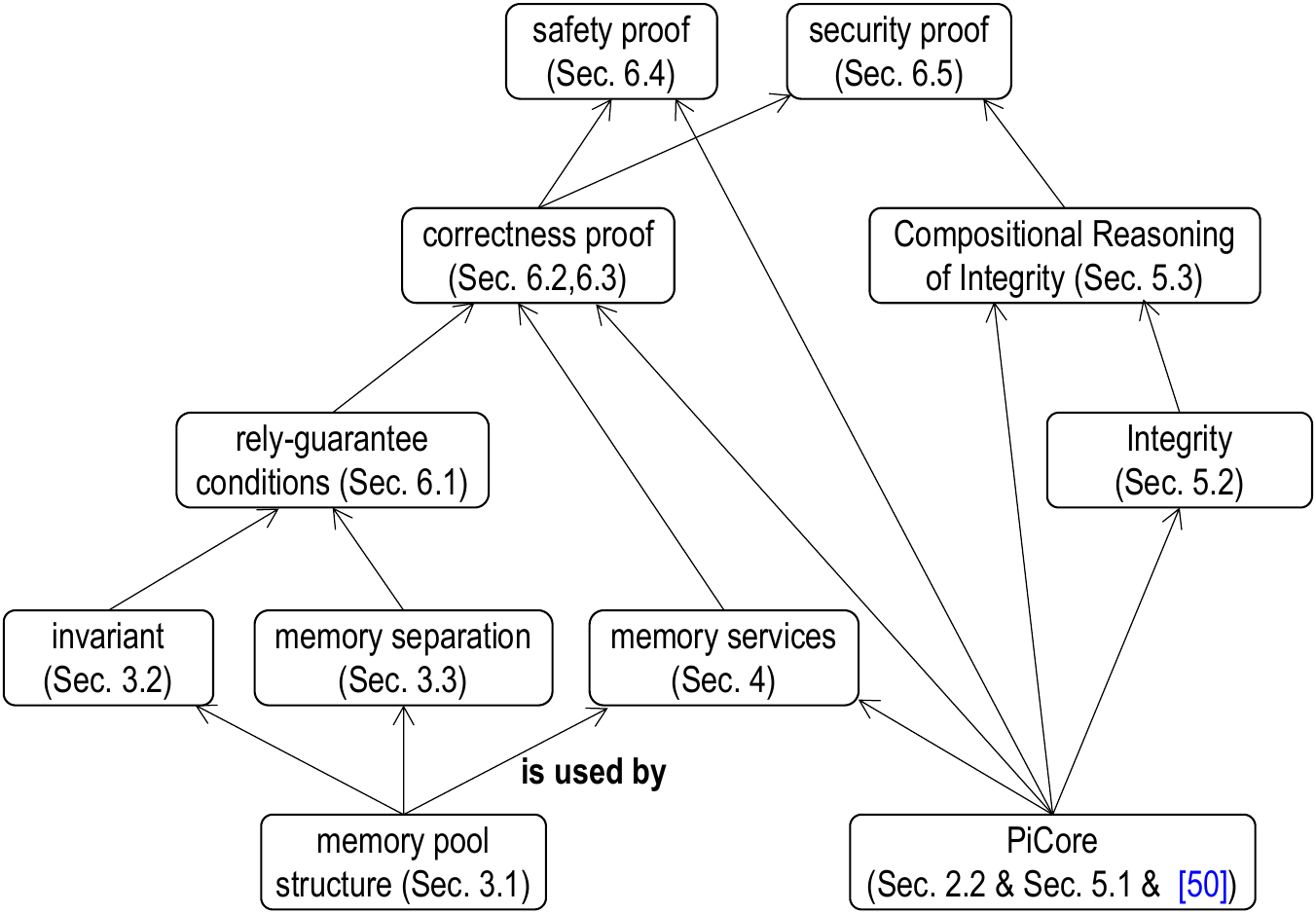}
\end{center}
\caption{Outline of Main Results}
\label{fig:roadmap}
\end{figure}
{\figprefix} \ref{fig:roadmap} summarizes the main results presented in this article. 
First, {\sectprefix} \ref{sect:preliminary} presents the preliminaries of this article including the buddy memory management in Zephyr ({\sectprefix} \ref{subsect:zephyr_mem}) and our previous {\slang} framework ({\sectprefix} \ref{subsect:picore}). 
In {\sectprefix} \ref{sect:mem_struct}, we formalize the memory data structures, and the safety and security properties of buddy memory pools. We define the memory structures in {\sectprefix} \ref{subsect:struct}, the invariant properties in {\sectprefix} \ref{subsect:inv}, and the memory separation properties in {\sectprefix} \ref{subsect:memsep}. The formal specification of memory allocation and release services of Zephyr is presented in {\sectprefix} \ref{sect:mem_spec}. For the compositional verification of security for Zephyr, we propose a security property \emph{integrity} for {\slang} specifications in {\sectprefix} \ref{subsect:ifs} and we discuss the compositional verification approach in {\sectprefix} \ref{subsect:compreason}. In {\sectprefix} \ref{sect:proof}, we show the rely-guarantee proofs of Zephyr. We first give the correctness specification by rely-guarantee conditions in {\sectprefix} \ref{subsect:corspec}, which is embedded with the invariant and memory separation properties. We then present the proof of partial correctness ({\sectprefix} \ref{subsect:corproof}), termination ({\sectprefix} \ref{subsect:termi}), safety ({\sectprefix} \ref{subsect:safetyproof}) and security ({\sectprefix} \ref{subsect:securityproof}) of the memory services. 

This article is an extension of our previous paper \cite{zhao19cav}. Compared to \cite{zhao19cav}, (1) we add security properties about memory separation at memory-block level and thread level in this article; (2) we present the invariants in a more comprehensive and formal way; (3) for memory separation in concurrent settings, we extend our {\slang} framework \cite{zhao19fm} by a compositional verification approach of integrity, and then we apply the new {\slang} framework to the rely-guarantee reasoning of Zephyr RTOS; (4) we add security proof and present more comprehensive proofs of correctness and safety in this article; (5) finally, we add comparison to related work and present the limitation and discussion of our work. 

\section{Preliminaries}
\label{sect:preliminary}

\subsection{Concurrent Memory Management in Zephyr RTOS}
\label{subsect:zephyr_mem}

In Zephyr, a memory pool is a kernel object that allows memory blocks to be dynamically allocated, from a designated memory region, and released back into the pool. Its C code implementation is shown in the left part of {\figprefix} \ref{fig:mem_datastruct}. The right part of this figure shows the formalization of the memory pool, which will be discussed in next section. 
A memory pool's buffer ($*buf$) is an $n\_max$-size array of blocks of $max\_sz$ bytes at level $0$, with no wasted space between them. The size of the buffer is thus $n\_max \times max\_sz$ bytes long. Zephyr tries to accomplish a memory request by splitting available blocks into smaller ones fitting as best as possible the requested size. 
Each ``level 0'' block is a quad-block that can be split into four smaller ``level 1'' blocks of equal size. Likewise, each level 1 block is itself a quad-block that can be split again. At each level, the four smaller blocks become \emph{buddies} or \emph{partners} to each other. The block size at level $l$ is thus $max\_sz / 4 ^ l$.

\lstdefinestyle{customc}{
  belowcaptionskip=1\baselineskip,
  breaklines=true,
  frame=none, %L %single %none
  xleftmargin=0pt, %\parindent,
  language=C,
  numbers=none,
  stepnumber=1,
  showstringspaces=false,
  basicstyle=\scriptsize\ttfamily, %\footnotesize\ttfamily,
  keywordstyle=\bfseries\color{green!40!black},
  commentstyle=\itshape\color{purple!40!black},
  identifierstyle=\color{blue},
  stringstyle=\color{orange},
}
\lstset{escapechar=@,style=customc}

%%\begin{figure}%[t]
%%\begin{minipage}[t]{1.0\textwidth}
%%\begin{minipage}[t]{0.48\textwidth}
%%\begin{lstlisting}
%%struct k_mem_block_id {
%%  u32_t pool : 8;
%%  u32_t level : 4;
%%  u32_t block : 20;
%%};
%%struct k_mem_pool_lvl {
%%  union {
%%    u32_t *bits_p;
%%    u32_t bits;
%%  };
%%  sys_dlist_t free_list;
%%};
%%\end{lstlisting}
%%\end{minipage}
%%~
%%\begin{minipage}[t]{0.48\textwidth}
%%\begin{lstlisting}
%%struct k_mem_block {
%%  void *data;
%%  struct k_mem_block_id id;
%%};
%%struct k_mem_pool {
%%  void *buf;
%%  size_t max_sz;
%%  u16_t n_max;
%%  u8_t n_levels;
%%  u8_t max_inline_level;
%%  struct k_mem_pool_lvl *levels;
%%  _wait_q_t wait_q;
%%};
%%\end{lstlisting}
%%\end{minipage}
%%\end{minipage}
%%\caption{The Data Structure of Memory Pool in Zephyr v1.8.0}
%%\label{fig:mem_datastruct}
%%\end{figure}
\begin{figure}[t]
\centering
\begin{minipage}[t]{1.0\textwidth}
\begin{minipage}[t]{0.40\textwidth}
\vspace{-4mm}
\begin{lstlisting}
struct k_mem_block_id {
  u32_t pool : 8;
  u32_t level : 4;
  u32_t block : 20;
};
struct k_mem_block {
  void *data;
  struct k_mem_block_id id;
};
struct k_mem_pool_lvl {
  union {
    u32_t *bits_p;
    u32_t bits;
  };
  sys_dlist_t free_list;
};
struct k_mem_pool {
  void *buf;
  size_t max_sz;
  u16_t n_max;
  u8_t n_levels;
  u8_t max_inline_level;
  struct k_mem_pool_lvl *levels;
  _wait_q_t wait_q;
};
\end{lstlisting}
\end{minipage}
~
\begin{minipage}[t]{0.60\textwidth}
\begin{isabellec} \fontsize{7pt}{0cm} %\footnotesize %\scriptsize
\isacommand{typedef}\isamarkupfalse%
\ mempool{\isacharunderscore}ref \ {\isacharequal}\ ref 

\isacommand{type{\isacharunderscore}synonym}\isamarkupfalse%
\ mem{\isacharunderscore}ref\ {\isacharequal}\ nat 

\isacommand{record}\isamarkupfalse%
\ Mem{\isacharunderscore}block\ {\isacharequal}\ pool\ {\isacharcolon}{\isacharcolon}\ mempool{\isacharunderscore}ref\isanewline
\ \ \ \ \ \ \ \ \ \ \ \ \ \ \ \ \ \ \ level\ {\isacharcolon}{\isacharcolon}\ nat\isanewline
\ \ \ \ \ \ \ \ \ \ \ \ \ \ \ \ \ \ \ block\ {\isacharcolon}{\isacharcolon}\ nat\ \isanewline
\ \ \ \ \ \ \ \ \ \ \ \ \ \ \ \ \ \ \ data\ {\isacharcolon}{\isacharcolon}\ mem{\isacharunderscore}ref 
\isanewline

\isacommand{datatype}\isamarkupfalse%
\ BlockState\ {\isacharequal}\ ALLOCATED\ {\isacharbar}\ FREE\ {\isacharbar}\ DIVIDED\ {\isacharbar}\isanewline
\quad \quad \quad \quad NOEXIST\ {\isacharbar}\ FREEING\ {\isacharbar}\ ALLOCATING 

\isacommand{record}\isamarkupfalse%
\ Mem{\isacharunderscore}pool{\isacharunderscore}lvl\ {\isacharequal}\ \ \isanewline
\ \ \ \ \ \ \ \ \ \ \ \ \ \ \ \ \ \ \ \ \ \ bits\ {\isacharcolon}{\isacharcolon}\ {\isachardoublequoteopen}BlockState\ list{\isachardoublequoteclose}\isanewline
\ \ \ \ \ \ \ \ \ \ \ \ \ \ \ \ \ \ \ \ \ \ free{\isacharunderscore}list\ {\isacharcolon}{\isacharcolon}\ {\isachardoublequoteopen}mem{\isacharunderscore}ref\ list{\isachardoublequoteclose}
\isanewline 

\isacommand{record}\isamarkupfalse%
\ Mem{\isacharunderscore}pool\ {\isacharequal}\ buf\ {\isacharcolon}{\isacharcolon}\ mem{\isacharunderscore}ref\ \isanewline
\ \ \ \ \ \ \ \ \ \ \ \ \ \ \ \ \ \ max{\isacharunderscore}sz\ {\isacharcolon}{\isacharcolon}\ nat\ \isanewline
\ \ \ \ \ \ \ \ \ \ \ \ \ \ \ \ \ \ n{\isacharunderscore}max\ {\isacharcolon}{\isacharcolon}\ nat\ \isanewline
\ \ \ \ \ \ \ \ \ \ \ \ \ \ \ \ \ \ n{\isacharunderscore}levels\ {\isacharcolon}{\isacharcolon}\ nat\ \isanewline
\ \ \ \ \ \ \ \ \ \ \ \ \ \ \ \ \ \ levels\ {\isacharcolon}{\isacharcolon}\ {\isachardoublequoteopen}Mem{\isacharunderscore}pool{\isacharunderscore}lvl\ list{\isachardoublequoteclose}\ \isanewline
\ \ \ \ \ \ \ \ \ \ \ \ \ \ \ \ \ \ wait{\isacharunderscore}q\ {\isacharcolon}{\isacharcolon}\ {\isachardoublequoteopen}Thread\ list{\isachardoublequoteclose} 

%%\isacommand{record}\isamarkupfalse%
%%\ State\ {\isacharequal}\ \isanewline
%%\ \ mem{\isacharunderscore}pools\ {\isacharcolon}{\isacharcolon}\ {\isachardoublequoteopen}mempool{\isacharunderscore}ref\ set{\isachardoublequoteclose}\isanewline
%%\isanewline
%%\ \ mem{\isacharunderscore}pool{\isacharunderscore}info\ {\isacharcolon}{\isacharcolon}\ {\isachardoublequoteopen}mempool{\isacharunderscore}ref\ {\isasymRightarrow}\ Mem{\isacharunderscore}pool{\isachardoublequoteclose}
\end{isabellec}
\end{minipage}
\end{minipage}
\caption{Data Structure of Memory Pool in Zephyr v1.8.0 and Its Formalization}
\label{fig:mem_datastruct}
\end{figure}

The pool is initially configured with the parameters  $n\_max$ and $max\_sz$, together with a third parameter $min\_sz$. $min\_sz$ defines the minimum size for an allocated block and must be a multiple of four, i.e., there exists an $X > 0$ such that $\min\_sz = 4 \times X$. Memory pool blocks are recursively split into quarters until blocks of the minimum size are obtained, at which point no further split can occur. 
The depth at which $min\_sz$ blocks are allocated is $n\_levels$ and satisfies that $n\_max = min\_sz \times 4 ^ {n\_levels}$.

Every memory block is composed of a $level$; a $block$ index within its level, ranging from $0$ to $(n\_max \times 4 ^ {level}) - 1$; and the $data$ as a pointer to the block start address, which is equal to $buf + (max\_sz / 4 ^ {level}) \times block$. We use the tuple $(level, block)$ to uniquely represent a block within a pool $p$.

A memory pool keeps track of how its buffer space has been split using a linked list \emph{free\_list} with the start address of the free blocks in each level. To improve the performance of coalescing partner blocks, memory pools maintain a bitmap at each level to indicate the allocation status of each block in the level.  This structure is represented by a C union of an integer \emph{bits} and an array \emph{bits\_p}. The implementation can allocate bitmaps at levels smaller than $max\_inlinle\_levels$ using only an integer \emph{bits}. However, the number of blocks in levels higher than $max\_inlinle\_levels$ make necessary to allocate the bitmap information using the array \emph{bits\_map}. 
In such a design, the levels of bitmaps actually form a forest of complete quad trees. 
The bit $i$ in the bitmap of level $j$ is set to $1$ for the block $(i,j)$ iff it is a free block, i.e., it is in the free list at level $i$. Otherwise, the bitmap for such block is set to $0$.

\lstdefinestyle{customc}{
  belowcaptionskip=1\baselineskip,
  breaklines=true,
  frame=single, %L
  xleftmargin=0pt, %\parindent,
  language=C,
  numbers=left,
  stepnumber=1,
  showstringspaces=false,
  basicstyle=\scriptsize\ttfamily, %\footnotesize\ttfamily,
  keywordstyle=\bfseries\color{green!40!black},
  commentstyle=\itshape\color{purple!40!black},
  identifierstyle=\color{blue},
  stringstyle=\color{orange},
}
\lstset{escapechar=@,style=customc}
%%
%%static void *break_block(struct k_mem_pool *p, void *block,
%%			 int l, size_t *lsizes)
%%{
%%  int i, bn, key;
%%  key = irq_lock();
%%  bn = block_num(p, block, lsizes[l]);
%%
%%  for (i = 1; i < 4; i++) {
%%    int lbn = 4*bn + i;
%%    int lsz = lsizes[l + 1];
%%    void *block2 = (lsz * i) + (char *)block;
%%    set_free_bit(p, l + 1, lbn);
%%    if (block_fits(p, block2, lsz)) {
%%      sys_dlist_append(&p->levels[l + 1].free_list, block2);
%%    }
%%  }
%%  irq_unlock(key);
%%  return block;
%%}

\begin{figure}
\begin{lstlisting}
static int pool_alloc(struct k_mem_pool *p,struct k_mem_block *block,size_t size)
{
  size_t lsizes[p->n_levels];
  int i, alloc_l = -1, free_l = -1, from_l;
  void *blk = NULL;
  lsizes[0] = _ALIGN4(p->max_sz);
  for (i = 0; i < p->n_levels; i++) {
    if (i > 0) { lsizes[i] = _ALIGN4(lsizes[i-1] / 4); }
    if (lsizes[i] < size) { break; }
    alloc_l = i;
    if (!level_empty(p, i)) { free_l = i; }
  }
  if (alloc_l < 0 || free_l < 0) {
    block->data = NULL;
    return -ENOMEM;
  }
  blk = alloc_block(p, free_l, lsizes[free_l]);
  if (!blk) {  return -EAGAIN; }
  /* Iteratively break the smallest enclosing block... */
  for (from_l = free_l; level_empty(p, alloc_l) && from_l < alloc_l; from_l++) {
    blk = break_block(p, blk, from_l, lsizes);
  }
  block->data = blk; block->id.pool = pool_id(p); block->id.level = alloc_l;
  block->id.block = block_num(p, block->data, lsizes[alloc_l]);
  return 0;
}

int k_mem_pool_alloc(struct k_mem_pool *p, struct k_mem_block *block, size_t size, s32_t timeout)
{
  int ret, key;
  s64_t end = 0;
  
  if (timeout > 0) { end = _tick_get() + _ms_to_ticks(timeout); }
  while (1) {
    ret = pool_alloc(p, block, size);
    if (ret == 0 || timeout == K_NO_WAIT || ret == -EAGAIN || (ret && ret != -ENOMEM)) {
      return ret;
    }
    key = irq_lock();
    _pend_current_thread(&p->wait_q, timeout);
    _Swap(key);

    if (timeout != K_FOREVER) {
      timeout = end - _tick_get();
      if (timeout < 0) { break; }
    }
  }
  return -EAGAIN;
}

\end{lstlisting}
\caption{The C Source Code of Memory Allocation in Zephyr v1.8.0}
\label{fig:mem_alloc_code}
\end{figure}

Zephyr provides two kernel services \emph{k\_mem\_pool\_alloc} and \emph{k\_mem\_pool\_free}, for memory allocation and release respectively. 
The main part of the C code of \emph{k\_mem\_pool\_alloc} is shown in {\figprefix} \ref{fig:mem_alloc_code} in a compact manner. 
When an application requests for a memory block, Zephyr first computes $alloc\_l$ and $free\_l$ (Lines 7 - 16). $alloc\_l$ is the level with the size of the smallest block that will satisfy the request, and $free\_l$, with $free\_l \leqslant alloc\_l$, is the lowest level where there are free memory blocks. Since the services are concurrent, when the service tries to allocate a free block \emph{blk} from level $free\_l$ (Line 17), blocks at that level may be allocated or merged into a bigger block by other concurrent threads. In such case the service will back out (Line 18) and tell the main function \emph{k\_mem\_pool\_alloc} to retry. If $blk$ is successfully locked for allocation, then it is broken down to level $alloc\_l$ (Lines 20 - 22). 
The allocation service \emph{k\_mem\_pool\_alloc} supports a \emph{timeout} parameter to allow threads waiting for that pool for a period of time when the call does not succeed. If the allocation fails (Line 36) and the timeout is not \emph{K\_NO\_WAIT}, the thread is suspended (Line 40) in a linked list \emph{wait\_q} and the context is switched to another thread (Line 41). 

Interruptions are always enabled in both services with the exception of the code for the functions \emph{alloc\_block} and \emph{break\_block}, which invoke \emph{irq\_lock} and \emph{irq\_unlock} to respectively enable and disable interruptions. 
Similar to \emph{k\_mem\_pool\_alloc}, the execution of \emph{k\_mem\_pool\_free} is interruptable as well. 

\subsection{The {\slang} Rely-guarantee Framework}
\label{subsect:picore}

The abstract syntax of the {\slang} language \cite{zhao19fm} is shown in {\figprefix} \ref{fig:lang}. 
The syntax for events distinguishes basic events pending to be triggered from already triggered events that are under execution.  
A basic event is defined as $\event{(l,g,P)}$, where $l$ is the event name, $g$ the guard condition, and $\symbprog$ the body of the event. When $\event{(l,g,P)}$ is triggered, its body begins to be executed and it becomes a triggered event $\anonevt{\symbprog}$. The execution of $\anonevt{\symbprog}$ just simulates the program $\symbprog$.
Events are parametrized in the meta-logic as ``$\lambda (plist, \symbcore). \  \event{(l,g,P)}$'', where $plist$ is the list of input parameters, and $\symbcore$ is the event system identifier that the event belongs to. These parameters are not part of the syntax of events to make the guard $g$ and the event body $P$, as well as the rely and guarantee relations, more flexible, allowing to define different instances of the relations for different values of $plist$ and $\symbcore$. 
{\figprefix} \ref{fig:event_examp} illustrates an \emph{event} in the concrete syntax of {\slang}. 
Instead of defining a language for programs, {\slang} reuses existing languages and their rely-guarantee proof systems. 

%{\figprefix} \ref{fig:event_examp} illustrates an \emph{event}, which has an event name, a list of input parameters, a guard condition to determine the conditions triggering the event, and an imperative program as its body.

\begin{figure}[t]
\begin{minipage}[t]{0.45\textwidth}
\begin{figure}[H]
{ \footnotesize
\begin{tabular}{l}
\textbf{Event}: \\
$
\begin{aligned}
\symbEvt \ ::= & \ \event{(l,g,P)} & (Basic \ Event)  \\ 
 | & \ \anonevt{\symbprog} & (Triggered\ Event)
\end{aligned}
$ \\ \\
\textbf{Event System}: \\
$
\begin{aligned}
\symbevtsys \ ::= & \ \evtsysdef & (Event \ Set) \\
| & \ \evtseq{\symbEvt}{\symbevtsys} & (Event \ Sequence)
\end{aligned}
$ \\ \\
\textbf{Parallel Event System}: \\
$
\begin{aligned}
\symbpes \ ::= \ \parsysc  %\parsys{\mathcal{S}_1}{\mathcal{S}_2}{\mathcal{S}_n} 
\end{aligned}
$
\end{tabular}
}
\caption{Abstract Syntax of {\slang} Language}
\label{fig:lang}
\end{figure}
\end{minipage}
\hspace{0.3cm}
\begin{minipage}[t]{0.45\textwidth}
\begin{figure}[H]
\begin{flushleft}
\begin{isabellec}
%\isacodeftsz
%\scriptsize
\footnotesize
\isacommand{EVENT} alloc [Ref p, Nat size, Int timeout] $@$ $\symbcore$ \isanewline 
\isacommand{WHEN} \isanewline
\quad p\ {\isasymin}\ {\isasymacute}mem{\isacharunderscore}pools\ 
{\isasymand}\ timeout\ {\isasymge}\ {\isacharminus}{\isadigit{1}} \isanewline
\isacommand{THEN} \isanewline
\quad ...... \isanewline
\quad \textbf{IF}\ timeout\ {\isachargreater}\ {\isadigit{0}}\ \textbf{THEN}\ \isanewline
\quad \quad  {\isasymacute}endt\ {\isacharcolon}{\isacharequal}\ {\isasymacute}endt{\isacharparenleft}t\ {\isacharcolon}{\isacharequal}\ {\isasymacute}tick\ {\isacharplus}\ timeout{\isacharparenright}\isanewline
\quad \textbf{FI};; \isanewline
\quad ...... \isanewline
\isacommand{END}
\end{isabellec}
\end{flushleft}
%\vspace{5pt}
\caption{An Example of Event in Concrete Syntax}
\label{fig:event_examp}
\end{figure}
\end{minipage}
\end{figure}

At the system reaction level, {\slang} considers a reactive system as a set of event handlers called \emph{event systems} responding to stimulus from the environment. 
The execution of an event system concerns the continuous evaluation of guards of the events with their input arguments. From the set of events for which their associated guard condition holds in the current state, one event is non-deterministically selected to be triggered, and then its body is executed. After the event finishes, the evaluation of guards starts again looking for the next event to be triggered. We call the semantics of event systems \emph{reactive semantics}, where the event context shows the event currently being executed.
A CRS is modeled as the \emph{parallel composition} of event systems that are concurrently executed.

{\slang} supports the verification of two different kinds of properties in the rely-guarantee proof system for reactive systems: pre and post conditions of events and invariants in the fine-grained execution of events.
A rely-guarantee specification for a system is a quadruple $RGCond = \rgconddefault$, where $pre$ is the pre-condition, $R$ is the rely condition, $G$ is the guarantee condition, and $pst$ is the post-condition. The intuitive meaning of a valid rely-guarantee specification for a parallel component $\symbprog$, denoted by  $\RGSAT{\symbprog}{\rgconddefault}$, is that if $\symbprog$ is executed from an initial state $s \in pre$ and any environment transition belongs to the rely relation $R$, then the state transitions carried out by $\symbprog$ belong to the guarantee relation $G$ and the final states belong to $pst$. $\progenv$ is used to represent static configuration of programs like environments for procedure declarations. 

We have defined a rely-guarantee axiomatic proof system for the {\slang} specification language to prove validity of rely-guarantee specifications. Soundness of the proof system with regards to the definition of validity has been proven in Isabelle/HOL. Some of the rules composing the axiomatic reasoning system are shown in {\figprefix} \ref{fig:proofrule}. %All proof rules are discussed in our technical report \cite{Zhao17}. 

\begin{figure}[t]
\centering
%\scriptsize
%\footnotesize
\fontsize{8pt}{0cm}

\begin{tabular}{cc}
\myinfer{Await}{\infer{\rgsat{(\cmdawait{\symbbexp}{\symbprog})}{\rgconddefault}}
{
\begin{tabular}{l}
$\rgsat{P}{\rgcond{pre \cap \symbbexp \cap \{V\}}{Id}{UNIV}{\{\symbstate.\ (V, \symbstate) \in G\} \cap pst}}$ \\
$stable(pre, R) \quad stable(pst, R)$
\end{tabular}
}}
&
\myinfer{BasicEvt}{\infer{\rgsat{\event{\symbevtbd}}{\rgconddefault}}
{
\begin{tabular}{l}
$\rgsat{body(\symbevtbd)}{\rgcond{pre \cap guard(\symbevtbd)}{R}{G}{pst}}$\\
$stable(pre, R) \quad \forall \symbstate. \ (\symbstate, \symbstate) \in G$
\end{tabular}
}} 
\end{tabular}

\begin{tabular}{cc}
\myinfer{While}{\infer{\rgsat{(\cmdwhile{\symbbexp}{P})}{\rgcond{loopinv}{R}{G}{pst}}}
{
\begin{tabular}{l}
$\rgsat{P}{\rgcond{loopinv \cap \symbbexp}{R}{G}{loopinv}}$ \\
$loopinv \cap - \symbbexp \subseteq pst \quad \forall \symbstate. \ (\symbstate, \symbstate) \in G $\\
$stable(loopinv, R) \quad stable(pst, R) $
\end{tabular}
}
}
&
\myinfer{Par}{\infer{\rgsat{\symbpes}{\rgconddefault}}
{
\begin{tabular}{l}
$(1)\forall \symbcore. \ \rgsat{\symbpes(\symbcore)}{\rgcond{pres_\symbcore}{Rs_\symbcore}{Gs_\symbcore}{psts_\symbcore}}$ \\
$(2)\forall \symbcore. \ pre \subseteq pres_\symbcore \quad (3)\forall \symbcore. \ psts_\symbcore \subseteq pst \quad (4)\forall \symbcore. \ Gs_\symbcore \subseteq G $ \\
$(5)\forall \symbcore. \ R \subseteq Rs_\symbcore \quad (6)\forall \symbcore, \symbcore'. \ \symbcore \neq \symbcore' \longrightarrow Gs_\symbcore \subseteq Rs_{\symbcore'}$
\end{tabular}
}}
\end{tabular}
%\vspace{0.1cm}

\caption{Subset of Rely-guarantee Proof Rules in {\slang}}
\label{fig:proofrule}
\end{figure}

A predicate $P$ is stable w.r.t. a relation $R$, represented as $stable(P,R)$, when for any pair of states $(s,t)$ such that $s \in P$ and $(s,t) \in R$ then $t \in P$. The intuitive meaning is that an environment represented by $R$ does not affect the satisfiability of $P$.
The parallel rule in  {\figprefix} \ref{fig:proofrule} establishes compositionality of the proof system, where  verification of the parallel specification can be reduced to the verification of individual event systems first and then to the verification of individual events. 
It is necessary that each event system $\symbpes(\symbcore)$ satisfies its specification $\rgcond{pres_\symbcore}{Rs_\symbcore}{Gs_\symbcore}{psts_\symbcore}$ (Premise 1); the pre-condition for the parallel composition implies all the event system's pre-conditions (Premise 2); the overall post-condition must
be a logical consequence of all post-conditions of event systems (Premise 3);
since an action transition of the concurrent system is performed by one of its event system, the guarantee condition $Gs_\symbcore$ of each event system must be a subset of the overall guarantee condition $G$ (Premise 4);
an environment transition $Rs_\symbcore$ for the event system $\symbcore$ corresponds to a transition from the overall environment $R$ (Premise 5); and an action transition of an event system $\symbcore$ should be defined in the rely condition of another event system $\symbcore'$, where $\symbcore \neq \symbcore'$ (Premise 6). 

{\slang} considers invariants of CRSs in safety verification.
To show that $inv$ is preserved by a system $\symbpes$, it suffices to show the invariant verification theorem as follows. This theorem indicates that (1) the system satisfies its rely-guarantee specification $\rgcond{init}{R}{G}{post}$, (2) $inv$ initially holds in the set of initial states, and (3) each action transition as well as each environment transition preserve $inv$. Later, invariant verification is decomposed to the verification of individual events by the proof system of {\slang}.

\begin{theorem}[Invariant Verification]
\label{thm:invariant}
For formal specification $\symbpes$ and $\progenv$, a state set $init$, a rely condition $R$, and $inv$, if 
\begin{itemize}
\item $\rgsat{\symbpes}{\rgcond{init}{R}{G}{post}}$.
\item $init \subseteq \{s.\ inv(s)\}$.
\item $stable(\{s.\ inv(s)\}, R)$ and $stable(\{s.\ inv(s)\}, G)$ are satisfied.
\end{itemize}
then $inv$ is preserved by $\symbpes$ w.r.t. $init$ and $R$. 
\end{theorem}

%%\subsection{Isabelle/HOL Notations}
%%\label{subsect:isabelle}

\section{Defining Structures and Properties of Buddy Memory Pools}
\label{sect:mem_struct}

This section formalizes the whole data structure of memory pools in Zephyr. Based on that formalization, we define safety properties as a comprehensive set of invariants and the security property as a two-level memory separation. The memory separation at memory-block level can be derived from the invariants on the memory. 

\subsection{Structure of Memory Pools}
\label{subsect:struct}

As a specification of low-level design, we use abstract data types to represent the complete structure of memory pools. The formalization of the memory pool in Zephyr is shown in the right part of {\figprefix} \ref{fig:mem_datastruct}. 
We use an abstract reference \emph{ref} in Isabelle to define pointers to memory pools. Starting addresses of memory blocks, memory pools, and unsigned integers in the implementation are defined as \emph{natural} numbers (\emph{nat}). Linked lists used in the implementation for the elements \emph{levels} and \emph{free\_list}, together with the bitmaps used in \emph{bits} and \emph{bits\_p}, are defined as a \emph{list} type. 
C \emph{structs} are modelled in Isabelle as \emph{records} of the same name as the implementation and comprising the same data. There are two exceptions to this: (1) $k\_mem\_block\_id$ and $k\_mem\_block$ are merged in one single record, (2) the union in the struct $k\_mem\_pool\_lvl$ is replaced by a single list representing the bitmap, and thus \emph{max\_inline\_level} is removed. %We use the same of fields in records as in the \emph{structs} of the C code. 

\begin{figure}[t]
\begin{center}
\includegraphics[width=4.8in]{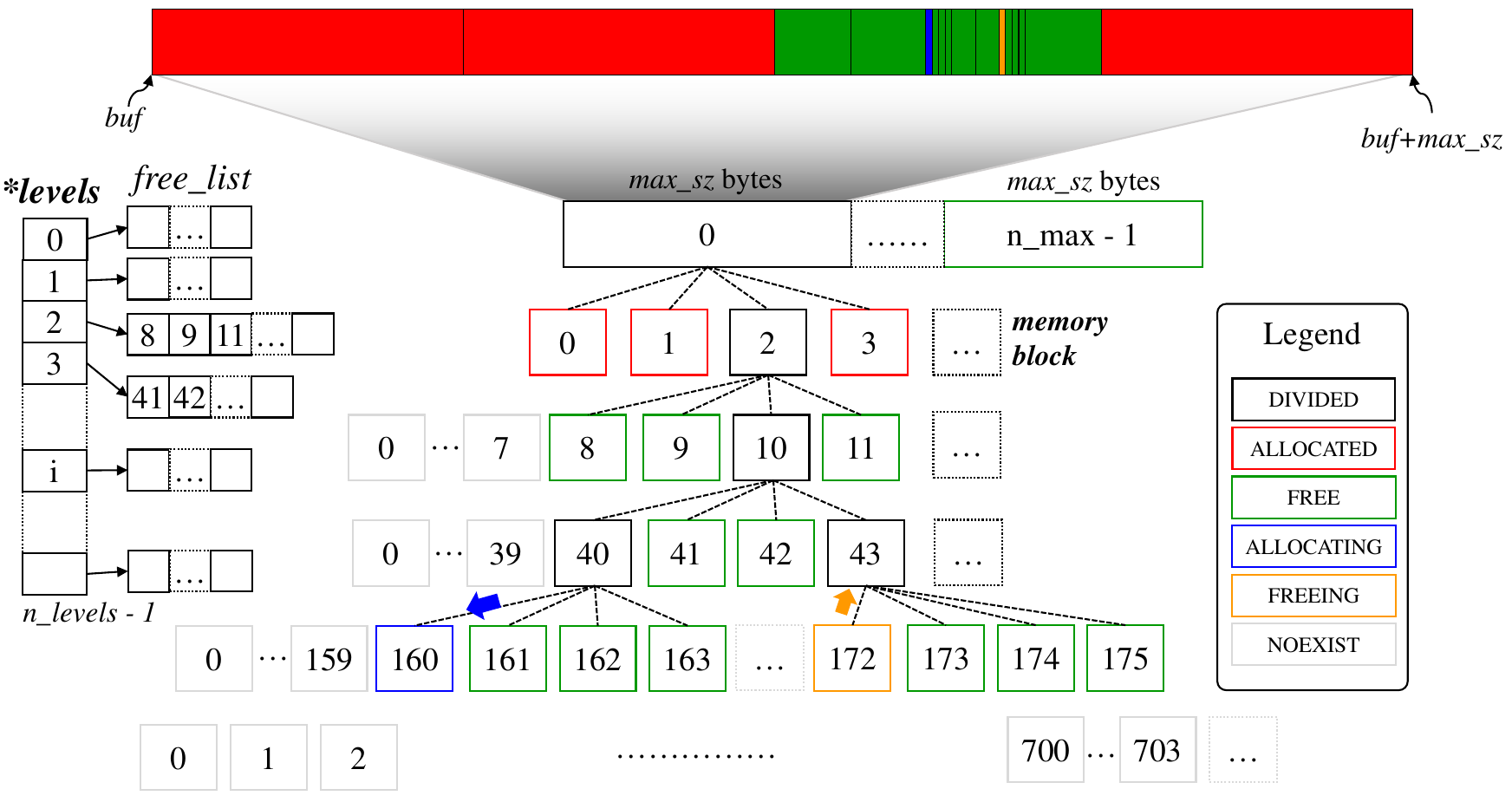}
\end{center}
\caption{Structure of Memory Pools}
\label{fig:mempool}
\end{figure}

Threads may concurrently split and coalesce memory blocks during the execution of the allocation and realease services. The Zephyr implementation makes use of a bitmap to represent the state of a memory block. The bit $j$ of the bitmap for level a $i$ is set to $1$ iff the memory address of the memory block $(i,j)$ is in the free list at level $i$. A bit $j$ at a level $i$ is set to $0$ under the following conditions: (1) its corresponding memory block is allocated (\emph{ALLOCATED}), (2) the memory block has been split (\emph{DIVIDED}), (3) the memory block is being split in the allocation service (\emph{ALLOCATING}) (Line 21 in {\figprefix} \ref{fig:mem_alloc_code}), (4) the memory block is being coalesced in the release service (\emph{FREEING}), and (5) the memory block does not exist (\emph{NOEXIST}). Instead of only using a binary representation, our formal specification models the bitmap using a datatype \emph{BlockState} that is composed of these cases together with \emph{FREE}. The reason of this decision is to simplify proving that the bitmap shape is well-formed. In particular, this representation makes less complex to verify the case in which the descendant of a free block is a non-free block. This is the case where the last free block has not been split and therefore lower levels do not exist. 
We illustrate a structure of a memory pool in {\figprefix} \ref{fig:mempool}. The top of the figure shows the real memory of the first block at level $0$.

\subsection{Invariant}
\label{subsect:inv}

The structural properties clarify the constraints on and the consistency of quad trees, free block lists, the memory pool configuration, and waiting threads. All of them are thought of as invariants on the kernel state and have been formally verified on the formal specification in Isabelle/HOL.

\subsubsection{Well-shaped bitmaps}
We say that the logical memory block $j$ at a level $i$ physically exists iff the value of the bitmap $j$ at the level $i$ is \emph{ALLOCATED}, \emph{FREE}, \emph{ALLOCATING}, or \emph{FREEING}, represented by the predicate $is\_memblock$. We do not consider blocks marked as \emph{DIVIDED} as physical blocks since it is only a logical block containing other blocks. 
 A valid forest is defined by the following rules: (1) the parent bit of an existing memory block is \emph{DIVIDED} and its child bits are \emph{NOEXIST}, denoted by the predicate $noexist\_bits$ that checks for a given bitmap $b$ and a position $j$ that nodes $b!j$ to $b!(j+3)$ are set as \emph{NOEXIST}; (2) the parent bit of a \emph{DIVIDED} block is also \emph{DIVIDED}; and (3) the child bits of a \emph{NOEXIST} bit are also \emph{NOEXIST} and its parent can not be a \emph{DIVIDED} block. The property is defined as the predicate \isacommand{inv{\isacharunderscore}bitmap}($s$) as follows, where $s$ is the system state of Zephyr memory management. 
 
\vspace{0.2cm}
\begin{isabellec}
\isacommand{inv{\isacharunderscore}bitmap}\ s\ {\isasymequiv} {\isasymforall}p{\isasymin}mem{\isacharunderscore}pools\ s{\isachardot}\ \textbf{let}\ mp\ {\isacharequal}\ mem{\isacharunderscore}pool{\isacharunderscore}info\ s\ p\ \textbf{in}\ \isanewline
\quad {\isasymforall}i\ {\isacharless}\ length\ {\isacharparenleft}levels\ mp{\isacharparenright}{\isachardot} \
 \textbf{let}\ bts\ {\isacharequal}\ bits\ {\isacharparenleft}levels\ mp\ {\isacharbang}\ i{\isacharparenright}\ \textbf{in}\isanewline
\quad \ \ {\isacharparenleft}{\isasymforall}j\ {\isacharless}\ length\ bts{\isachardot}\ 
 {\isacharparenleft}is\_memblock(bts\ {\isacharbang}\ j) {\isasymlongrightarrow}\ {\isacharparenleft}i\ {\isachargreater}\ {\isadigit{0}}\ {\isasymlongrightarrow}\ {\isacharparenleft}bits\ {\isacharparenleft}levels\ mp\ {\isacharbang}\ {\isacharparenleft}i\ {\isacharminus}\ {\isadigit{1}}{\isacharparenright}{\isacharparenright}{\isacharparenright}\ {\isacharbang}\ {\isacharparenleft}j\ div\ {\isadigit{4}}{\isacharparenright}\ {\isacharequal}\ DIVIDED{\isacharparenright}
 
\quad \quad \quad \quad \quad \quad \quad \quad \quad \quad \quad \quad \quad \quad \quad \quad \quad \quad 
{\isasymand}\ {\isacharparenleft}i\ {\isacharless}\ length\ {\isacharparenleft}levels\ mp{\isacharparenright}\ {\isacharminus}\ {\isadigit{1}}\ {\isasymlongrightarrow}\ noexist{\isacharunderscore}bits\ mp\ {\isacharparenleft}i{\isacharplus}{\isadigit{1}}{\isacharparenright}\ {\isacharparenleft}j{\isacharasterisk}{\isadigit{4}}{\isacharparenright}\ {\isacharparenright}{\isacharparenright}

\quad \ \ \ {\isasymand}\ {\isacharparenleft}bts\ {\isacharbang}\ j\ {\isacharequal}\ DIVIDED\ {\isasymlongrightarrow}\ i\ {\isachargreater}\ {\isadigit{0}}\ {\isasymlongrightarrow}\ {\isacharparenleft}bits\ {\isacharparenleft}levels\ mp\ {\isacharbang}\ {\isacharparenleft}i\ {\isacharminus}\ {\isadigit{1}}{\isacharparenright}{\isacharparenright}{\isacharparenright}\ {\isacharbang}\ {\isacharparenleft}j\ div\ {\isadigit{4}}{\isacharparenright}\ {\isacharequal}\ DIVIDED{\isacharparenright}\isanewline
\quad \ \ \ {\isasymand}\ {\isacharparenleft}bts\ {\isacharbang}\ j\ {\isacharequal}\ NOEXIST\ {\isasymlongrightarrow}\ i\ {\isacharless}\ length\ {\isacharparenleft}levels\ mp{\isacharparenright}\ {\isacharminus}\ {\isadigit{1}} {\isasymlongrightarrow}\ noexist{\isacharunderscore}bits\ mp\ {\isacharparenleft}i{\isacharplus}{\isadigit{1}}{\isacharparenright}\ {\isacharparenleft}j{\isacharasterisk}{\isadigit{4}}{\isacharparenright}{\isacharparenright}\isanewline
\quad \ \ \ {\isasymand}\ {\isacharparenleft}bts\ {\isacharbang}\ j\ {\isacharequal}\ NOEXIST\ {\isasymand}\ i\ {\isachargreater}\ {\isadigit{0}}\ {\isasymlongrightarrow}\ {\isacharparenleft}bits\ {\isacharparenleft}levels\ mp\ {\isacharbang}\ {\isacharparenleft}i\ {\isacharminus}\ {\isadigit{1}}{\isacharparenright}{\isacharparenright}{\isacharparenright}\ {\isacharbang}\ {\isacharparenleft}j\ div\ {\isadigit{4}}{\isacharparenright}\ {\isasymnoteq}\ DIVIDED{\isacharparenright}\ {\isacharparenright}
\end{isabellec}
\vspace{0.2cm}

In Isabelle, \emph{mem\_pools\ s} captures the set of pools in state $s$, and \emph{mem\_pool\_info\ s\ p} gets the memory pool referred by $p$. For a list $l$, \emph{l ! i} gets the $i$th element. 

There are two additional properties on bitmaps. First, the address space of any memory pool cannot be empty, i.e., the bits at level 0 have to be different to \emph{NOEXIST}. Second, the allocation algorithm may split a memory block into smaller ones, but not the those blocks at the lowest level (i.e. level $n\_levels - 1$), therefore the bits at the lowest level have to be different than \emph{DIVIDED}, being invalid if it is divided. The first property is defined as \isacommand{inv{\isacharunderscore}bitmap{\isadigit{0}}}($s$) and the second as \isacommand{inv{\isacharunderscore}bitmapn}($s$).

\vspace{0.2cm}
\begin{isabellec}
\isacommand{inv{\isacharunderscore}bitmap{\isadigit{0}}}\ s\ {\isasymequiv} {\isasymforall}p{\isasymin}mem{\isacharunderscore}pools\ s{\isachardot}

\quad \quad 
\textbf{let}\ bits{\isadigit{0}}\ {\isacharequal}\ bits\ {\isacharparenleft}levels\ {\isacharparenleft}mem{\isacharunderscore}pool{\isacharunderscore}info\ s\ p{\isacharparenright}\ {\isacharbang}\ {\isadigit{0}}{\isacharparenright}\ \textbf{in}\
{\isasymforall}i{\isacharless}length\ bits{\isadigit{0}}{\isachardot}\ bits{\isadigit{0}}\ {\isacharbang}\ i\ {\isasymnoteq}\ NOEXIST

\isacommand{inv{\isacharunderscore}bitmapn}\ s\ {\isasymequiv} {\isasymforall}p{\isasymin}mem{\isacharunderscore}pools\ s{\isachardot}

\quad \quad 
\textbf{let}\ bitsn\ {\isacharequal}\ bits\ {\isacharparenleft}{\isacharparenleft}levels\ {\isacharparenleft}mem{\isacharunderscore}pool{\isacharunderscore}info\ s\ p{\isacharparenright}\ {\isacharbang}\ {\isacharparenleft}length\ {\isacharparenleft}levels\ {\isacharparenleft}mem{\isacharunderscore}pool{\isacharunderscore}info\ s\ p{\isacharparenright}{\isacharparenright}\ {\isacharminus}\ {\isadigit{1}}{\isacharparenright}{\isacharparenright}{\isacharparenright}\isanewline
\ \ \ \ \ \ \ \ \ \ \ \ \ \ \ \ \ \ \textbf{in}\ {\isasymforall}i{\isacharless}length\ bitsn{\isachardot}\ bitsn\ {\isacharbang}\ i\ {\isasymnoteq}\ DIVIDED
\end{isabellec}
\vspace{0.2cm}

\subsubsection{Consistency of the memory configuration}
The configuration of a memory pool is set when it is initialized. Since the minimum block size is aligned to 4 bytes, there must exist an $n > 0$ such that the maximum size of a pool is equal to $4 \times n \times 4 ^ {n\_levels}$, relating the number of levels of a level 0 block with its maximum size. Moreover, the number of blocks at level 0 and the number of levels have to be greater than zero, since the memory pool cannot be empty. The number of levels is equal to the length of the  pool $levels$ list. Finally, the length of the bitmap at level $i$ has to be $n\_max \times 4 ^ i$. This property is defined as \isacommand{inv{\isacharunderscore}mempool{\isacharunderscore}info}($s$).

\vspace{0.2cm}
\begin{isabellec} 
\isacommand{inv{\isacharunderscore}mempool{\isacharunderscore}info}\ s\ {\isasymequiv}\ {\isasymforall}p{\isasymin}mem{\isacharunderscore}pools\ s{\isachardot}\ \textbf{let}\ mp\ {\isacharequal}\ mem{\isacharunderscore}pool{\isacharunderscore}info\ s\ p\ \textbf{in}\ \isanewline
\ \ \ \ \ \ \ {\isacharparenleft}{\isasymexists}n{\isachargreater}{\isadigit{0}}{\isachardot}\ max{\isacharunderscore}sz\ mp\ {\isacharequal}\ {\isacharparenleft}{\isadigit{4}}\ {\isacharasterisk}\ n{\isacharparenright}\ {\isacharasterisk}\ {\isacharparenleft}{\isadigit{4}}\ {\isacharcircum}\ n{\isacharunderscore}levels\ mp{\isacharparenright}{\isacharparenright}\isanewline
\ \ \ \ \ \ {\isasymand}\ n{\isacharunderscore}max\ mp\ {\isachargreater}\ {\isadigit{0}}\ {\isasymand}\ n{\isacharunderscore}levels\ mp\ {\isachargreater}\ {\isadigit{0}} {\isasymand}\ n{\isacharunderscore}levels\ mp\ {\isacharequal}\ length\ {\isacharparenleft}levels\ mp{\isacharparenright}\isanewline
\ \ \ \ \ \ {\isasymand}\ {\isacharparenleft}{\isasymforall}i{\isacharless}length\ {\isacharparenleft}levels\ mp{\isacharparenright}{\isachardot}\ length\ {\isacharparenleft}bits\ {\isacharparenleft}levels\ mp\ {\isacharbang}\ i{\isacharparenright}{\isacharparenright}\ {\isacharequal}\ {\isacharparenleft}n{\isacharunderscore}max\ mp{\isacharparenright}\ {\isacharasterisk}\ {\isadigit{4}}\ {\isacharcircum}\ i{\isacharparenright}
\end{isabellec}
\vspace{0.2cm}

\subsubsection{No partner fragmentation}
The memory release algorithm in Zephyr coalesces free partner memory blocks into blocks as large as possible for all the descendants from the root level, without including it. Thus, a memory pool does not contain four \emph{FREE} partner bits. This is checked by the $partner\_bits$ function. Note that the blocks of a pool at level 0 should not be coalesced. This property is defined as the {\isacommand{inv{\isacharunderscore}bitmap{\isacharunderscore}not{\isadigit{4}}free}(s)} predicate as follows.

\vspace{0.2cm}
\begin{isabellec}
\isacommand{inv{\isacharunderscore}bitmap{\isacharunderscore}not{\isadigit{4}}free}\ s\ {\isasymequiv} {\isasymforall}p{\isasymin}mem{\isacharunderscore}pools\ s{\isachardot}\ \textbf{let}\ mp\ {\isacharequal}\ mem{\isacharunderscore}pool{\isacharunderscore}info\ s\ p\ \textbf{in}\ \isanewline
\ \ \ \ \ \ \ \ \ \ \ \ {\isasymforall}i\ {\isacharless}\ length\ {\isacharparenleft}levels\ mp{\isacharparenright}{\isachardot}
\textbf{let}\ bts\ {\isacharequal}\ bits\ {\isacharparenleft}levels\ mp\ {\isacharbang}\ i{\isacharparenright}\ \textbf{in}\isanewline
\ \ \ \ \ \ \ \ \ \ \ \ \ \ {\isacharparenleft}{\isasymforall}j\ {\isacharless}\ length\ bts{\isachardot}\ i\ {\isachargreater}\ {\isadigit{0}}\ {\isasymlongrightarrow}\ {\isasymnot}\ partner{\isacharunderscore}bits\ mp\ i\ j{\isacharparenright}
\end{isabellec}
\vspace{0.2cm}

\subsubsection{Validity of free block lists}
The free list at one level keeps the starting address of free memory blocks. The memory management ensures that the addresses in the list are valid, i.e., they are different from each other and aligned to the \emph{block size}, which at a level $i$ is given by ($max\_sz / 4 ^ i$). 
Moreover, a memory block is in the free list iff the corresponding bit of the bitmap is \emph{FREE}. This property is defined as the {\isacommand{inv{\isacharunderscore}bitmap{\isacharunderscore}freelist}(s)} predicate as follows. 

\vspace{0.2cm}
\begin{isabellec}
\isacommand{inv{\isacharunderscore}bitmap{\isacharunderscore}freelist}\ s\ {\isasymequiv} {\isasymforall}p{\isasymin}mem{\isacharunderscore}pools\ s{\isachardot}\ \textbf{let}\ mp\ {\isacharequal}\ mem{\isacharunderscore}pool{\isacharunderscore}info\ s\ p\ \textbf{in}\ \isanewline
\quad {\isasymforall}i\ {\isacharless}\ length\ {\isacharparenleft}levels\ mp{\isacharparenright}{\isachardot}
\textbf{let}\ bts\ {\isacharequal}\ bits\ {\isacharparenleft}levels\ mp\ {\isacharbang}\ i{\isacharparenright}{\isacharsemicolon}\
fl\ {\isacharequal}\ free{\isacharunderscore}list\ {\isacharparenleft}levels\ mp\ {\isacharbang}\ i{\isacharparenright}\ \textbf{in}\isanewline
\quad \quad {\isacharparenleft}{\isasymforall}j\ {\isacharless}\ length\ bts{\isachardot}\ bts\ {\isacharbang}\ j\ {\isacharequal}\ FREE\ {\isasymlongleftrightarrow}\ buf\ mp\ {\isacharplus}\ j\ {\isacharasterisk}\ {\isacharparenleft}max{\isacharunderscore}sz\ mp\ div\ {\isacharparenleft}{\isadigit{4}}\ {\isacharcircum}\ i{\isacharparenright}{\isacharparenright}\ {\isasymin}\ set\ fl{\isacharparenright}\ {\isasymand} \isanewline
\quad \quad {\isacharparenleft}{\isasymforall}j\ {\isacharless}\ length\ fl{\isachardot}\ {\isacharparenleft}{\isasymexists}n{\isachardot}\ n\ {\isacharless}\ n{\isacharunderscore}max\ mp\ {\isacharasterisk}\ {\isacharparenleft}{\isadigit{4}}\ {\isacharcircum}\ i{\isacharparenright}\ {\isasymand}\ fl\ {\isacharbang}\ j\ {\isacharequal}\ buf\ mp\ {\isacharplus}\ n\ {\isacharasterisk}\ {\isacharparenleft}max{\isacharunderscore}sz\ mp\ div\ {\isacharparenleft}{\isadigit{4}}\ {\isacharcircum}\ i{\isacharparenright}{\isacharparenright}{\isacharparenright}{\isacharparenright}
{\isasymand}\ distinct\ fl\ 
\end{isabellec}
\vspace{0.2cm}

\subsubsection{Non-overlapping of memory pools}
The memory spaces of the set of pools defined in a system must be disjoint, so the set of memory addresses of a pool does not belong to the memory space of any other pool. This property is defined as the {\isacommand{inv{\isacharunderscore}pools{\isacharunderscore}notoverlap}(s)} predicate as follows. 

\vspace{0.2cm}
\begin{isabellec}
\isacommand{inv{\isacharunderscore}pools{\isacharunderscore}notoverlap}\ s\ {\isasymequiv} {\isacharparenleft} {\isasymforall}p1 \ p2{\isachardot}\ 
p1{\isasymin}mem{\isacharunderscore}pools\ s {\isasymand}\ p1{\isasymin}mem{\isacharunderscore}pools\ s {\isasymand} p1 {\isasymnoteq} p2 {\isasymlongrightarrow}

\quad \quad \quad 
{\isacharparenleft}{\isasymnexists}addr{\isachardot}\ addr \ {\isasymin} {mempool{\isacharunderscore}addrspace} \ s \ p1 {\isasymand} addr \ {\isasymin} {mempool{\isacharunderscore}addrspace} \ s \ p2 {\isacharparenright}{\isacharparenright}
\end{isabellec}
\vspace{0.2cm}

\subsubsection{Consistency of waiting threads}
The state of a suspended thread in \emph{wait\_q} has to be consistent with the threads waiting for a memory pool. Threads can only be blocked once, and those threads waiting for available memory blocks have to be in a \emph{BLOCKED} state. This property is defined as the {\isacommand{inv{\isacharunderscore}thd{\isacharunderscore}waitq}(s)} predicate as follows. 

\vspace{0.2cm}
\begin{isabellec}
%\isacommand{inv{\isacharunderscore}cur}\ s\ {\isasymequiv}\ {\isasymforall}t{\isachardot}\ cur\ s\ {\isacharequal}\ Some\ t\ {\isasymlongleftrightarrow}\ thd{\isacharunderscore}state\ s\ t\ {\isacharequal}\ RUNNING
\isacommand{inv{\isacharunderscore}thd{\isacharunderscore}waitq}\ s\ {\isasymequiv}\
{\isacharparenleft}{\isasymforall}p{\isasymin}mem{\isacharunderscore}pools\ s{\isachardot}\ {\isasymforall}t{\isasymin}\ set\ {\isacharparenleft}wait{\isacharunderscore}q\ {\isacharparenleft}mem{\isacharunderscore}pool{\isacharunderscore}info\ s\ p{\isacharparenright}{\isacharparenright}{\isachardot}\ thd{\isacharunderscore}state\ s\ t\ {\isacharequal}\ BLOCKED{\isacharparenright}

\quad \quad \quad {\isasymand}\ {\isacharparenleft}{\isasymforall}t{\isachardot}\ thd{\isacharunderscore}state\ s\ t\ {\isacharequal}\ BLOCKED\ {\isasymlongrightarrow}\ {\isacharparenleft}{\isasymexists}p{\isasymin}mem{\isacharunderscore}pools\ s{\isachardot}\ t\ {\isasymin}\ set\ {\isacharparenleft}wait{\isacharunderscore}q\ {\isacharparenleft}mem{\isacharunderscore}pool{\isacharunderscore}info\ s\ p{\isacharparenright}{\isacharparenright}{\isacharparenright}{\isacharparenright}

\quad \quad \quad {\isasymand}\ {\isacharparenleft}{\isasymforall}p{\isasymin}mem{\isacharunderscore}pools\ s{\isachardot}\ dist{\isacharunderscore}list\ {\isacharparenleft}wait{\isacharunderscore}q\ {\isacharparenleft}mem{\isacharunderscore}pool{\isacharunderscore}info\ s\ p{\isacharparenright}{\isacharparenright}{\isacharparenright}

\quad \quad \quad {\isasymand}\ {\isacharparenleft}{\isasymforall}p\ q{\isachardot}\ p{\isasymin}mem{\isacharunderscore}pools\ s\ {\isasymand}\ q{\isasymin}mem{\isacharunderscore}pools\ s\ {\isasymand}\ p\ {\isasymnoteq}\ q {\isasymlongrightarrow}

\quad \quad \quad \quad \quad \quad  {\isacharparenleft}{\isasymnexists}t{\isachardot}\ t\ {\isasymin}\ set\ {\isacharparenleft}wait{\isacharunderscore}q\ {\isacharparenleft}mem{\isacharunderscore}pool{\isacharunderscore}info\ s\ p{\isacharparenright}{\isacharparenright}\ {\isasymand}\ t{\isasymin}\ set\ {\isacharparenleft}wait{\isacharunderscore}q\ {\isacharparenleft}mem{\isacharunderscore}pool{\isacharunderscore}info\ s\ q{\isacharparenright}{\isacharparenright}{\isacharparenright}{\isacharparenright}
\end{isabellec}
\vspace{0.2cm}

\subsubsection{Consistency of freeing and allocating blocks}
During allocation and release of a memory block, blocks of the tree may temporally be manipulated during the coalesce and division process. A block can be only manipulated by a thread at a time, and the state bit of a block being temporally manipulate has to be \emph{FREEING} or \emph{ALLOCATING}. Moreover, one of these memory blocks is being manipulated at most by one thread. This property is defined as the {\isacommand{inv{\isacharunderscore}aux{\isacharunderscore}vars}(s)} predicate as follows. 

\vspace{0.2cm}
\begin{isabellec}
\isacommand{inv{\isacharunderscore}aux{\isacharunderscore}vars}\ s\ {\isasymequiv}\

\quad {\isacharparenleft}{\isasymforall}t\ n{\isachardot}\ freeing{\isacharunderscore}node\ s\ t\ {\isacharequal}\ Some\ n {\isasymlongrightarrow} get{\isacharunderscore}bit\ {\isacharparenleft}mem{\isacharunderscore}pool{\isacharunderscore}info\ s{\isacharparenright}\ {\isacharparenleft}pool\ n{\isacharparenright}\ {\isacharparenleft}level\ n{\isacharparenright}\ {\isacharparenleft}block\ n{\isacharparenright}\ {\isacharequal}\ FREEING{\isacharparenright} {\isasymand}

\quad {\isacharparenleft}{\isasymforall}n{\isachardot}\ get{\isacharunderscore}bit\ {\isacharparenleft}mem{\isacharunderscore}pool{\isacharunderscore}info\ s{\isacharparenright}\ {\isacharparenleft}pool\ n{\isacharparenright}\ {\isacharparenleft}level\ n{\isacharparenright}\ {\isacharparenleft}block\ n{\isacharparenright}\ {\isacharequal}\ FREEING {\isasymand}\ mem{\isacharunderscore}block{\isacharunderscore}addr{\isacharunderscore}valid\ s\ n\ {\isasymlongrightarrow}\ \isanewline
\quad \quad \quad \quad \quad \quad \quad   {\isacharparenleft}{\isasymexists}t{\isachardot}\ freeing{\isacharunderscore}node\ s\ t\ {\isacharequal}\ Some\ n{\isacharparenright}{\isacharparenright} {\isasymand}

\quad {\isacharparenleft}{\isasymforall}t\ n{\isachardot}\ allocating{\isacharunderscore}node\ s\ t\ {\isacharequal}\ Some\ n {\isasymlongrightarrow}\ \isanewline
\quad \quad \quad \quad \quad \quad get{\isacharunderscore}bit\ {\isacharparenleft}mem{\isacharunderscore}pool{\isacharunderscore}info\ s{\isacharparenright}\ {\isacharparenleft}pool\ n{\isacharparenright}\ {\isacharparenleft}level\ n{\isacharparenright}\ {\isacharparenleft}block\ n{\isacharparenright}\ {\isacharequal}\ ALLOCATING{\isacharparenright} {\isasymand} 

\quad {\isacharparenleft}{\isasymforall}n{\isachardot}\ get{\isacharunderscore}bit\ {\isacharparenleft}mem{\isacharunderscore}pool{\isacharunderscore}info\ s{\isacharparenright}\ {\isacharparenleft}pool\ n{\isacharparenright}\ {\isacharparenleft}level\ n{\isacharparenright}\ {\isacharparenleft}block\ n{\isacharparenright}\ {\isacharequal}\ ALLOCATING {\isasymand}\ mem{\isacharunderscore}block{\isacharunderscore}addr{\isacharunderscore}valid\ s\ n\ {\isasymlongrightarrow}\ \isanewline
\quad \quad \quad \quad \quad \quad \quad  {\isacharparenleft}{\isasymexists}t{\isachardot}\ allocating{\isacharunderscore}node\ s\ t\ {\isacharequal}\ Some\ n{\isacharparenright}{\isacharparenright} {\isasymand} 

\quad {\isacharparenleft}{\isasymforall}t{\isadigit{1}}\ t{\isadigit{2}}\ n{\isadigit{1}}\ n{\isadigit{2}}{\isachardot}\ t{\isadigit{1}}\ {\isasymnoteq}\ t{\isadigit{2}}\ {\isasymand}\ freeing{\isacharunderscore}node\ s\ t{\isadigit{1}}\ {\isacharequal}\ Some\ n{\isadigit{1}}\ {\isasymand}\ freeing{\isacharunderscore}node\ s\ t{\isadigit{2}}\ {\isacharequal}\ Some\ n{\isadigit{2}}\ {\isasymlongrightarrow}\ \isanewline
\quad \quad \quad \quad \quad \quad \quad {\isasymnot}{\isacharparenleft}pool\ n{\isadigit{1}}\ {\isacharequal}\ pool\ n{\isadigit{2}}\ {\isasymand}\ level\ n{\isadigit{1}}\ {\isacharequal}\ level\ n{\isadigit{2}}\ {\isasymand}\ block\ n{\isadigit{1}}\ {\isacharequal}\ block\ n{\isadigit{2}}{\isacharparenright}{\isacharparenright} {\isasymand}

\quad {\isacharparenleft}{\isasymforall}t{\isadigit{1}}\ t{\isadigit{2}}\ n{\isadigit{1}}\ n{\isadigit{2}}{\isachardot}\ t{\isadigit{1}}\ {\isasymnoteq}\ t{\isadigit{2}}\ {\isasymand}\ allocating{\isacharunderscore}node\ s\ t{\isadigit{1}}\ {\isacharequal}\ Some\ n{\isadigit{1}}\ {\isasymand}\ allocating{\isacharunderscore}node\ s\ t{\isadigit{2}}\ {\isacharequal}\ Some\ n{\isadigit{2}}\ {\isasymlongrightarrow}\ \isanewline
\quad \quad \quad \ {\isasymnot}{\isacharparenleft}pool\ n{\isadigit{1}}\ {\isacharequal}\ pool\ n{\isadigit{2}}\ {\isasymand}\ level\ n{\isadigit{1}}\ {\isacharequal}\ level\ n{\isadigit{2}}\ {\isasymand}\ block\ n{\isadigit{1}}\ {\isacharequal}\ block\ n{\isadigit{2}}{\isacharparenright}{\isacharparenright} {\isasymand} 

\quad {\isacharparenleft}{\isasymforall}t{\isadigit{1}}\ t{\isadigit{2}}\ n{\isadigit{1}}\ n{\isadigit{2}}{\isachardot}\ allocating{\isacharunderscore}node\ s\ t{\isadigit{1}}\ {\isacharequal}\ Some\ n{\isadigit{1}}\ {\isasymand}\ freeing{\isacharunderscore}node\ s\ t{\isadigit{2}}\ {\isacharequal}\ Some\ n{\isadigit{2}}\ {\isasymlongrightarrow}\  \isanewline
\quad \quad \quad \quad \quad \quad \quad {\isasymnot}{\isacharparenleft}pool\ n{\isadigit{1}}\ {\isacharequal}\ pool\ n{\isadigit{2}}\ {\isasymand}\ level\ n{\isadigit{1}}\ {\isacharequal}\ level\ n{\isadigit{2}}\ {\isasymand}\ block\ n{\isadigit{1}}\ {\isacharequal}\ block\ n{\isadigit{2}}{\isacharparenright}{\isacharparenright}
\end{isabellec}
\vspace{0.2cm}

\subsection{Memory Separation}
\label{subsect:memsep}
To operate securely, the memory management must enforce memory separation. That is, it must prevent memory leaks and insecure data flows. Informally, this means that it must prevent memory leaks among blocks due to incorrect memory split and coalescing for the Zephir buddy memory allocation service. It also must avoid that an allocated memory areas is influenced by any thread different than the one carrying out the allocation of that area. 
We call the two security properties memory separation at memory-block level and at thread level respectively. 

\subsubsection{Memory Separation at Memory-block Level}
When allocating and releasing memory for threads, the memory allocator must correctly manage the allocated and freed memory blocks. On one hand, it should prevent allocating one memory address to multiple memory blocks to avoid possible data leaks among threads. On the other hand, it should be able to control all the memory addresses in a memory pool to avoid interspaces between memory blocks. 
Therefore, memory blocks in a memory pool partition its address space, i.e. blocks are not overlapping each other and the addresses of all blocks cover the address space of the pool. 

For a memory block with index $j$ at level $i$, its address space is the interval $[j \times (max\_sz / 4 ^ i), (j + 1) \times (max\_sz / 4 ^ i))$. For any relative memory address $addr$ in the memory domain of a memory pool, and hence $addr < n\_max * max\_sz$, there is one and only one memory block whose address space contains $addr$. Here, we use relative address for $addr$.  The property is defined as the {\isacommand{mem{\isacharunderscore}part}(s)} predicate as follows.

\vspace{0.2cm}
\begin{isabellec}
\isacommand{addr{\isacharunderscore}in{\isacharunderscore}block}\ mp\ addr\ i\ j\ {\isasymequiv}
i\ {\isacharless}\ length\ {\isacharparenleft}levels\ mp{\isacharparenright}\ {\isasymand}\ j\ {\isacharless}\ length\ {\isacharparenleft}bits\ {\isacharparenleft}levels\ mp\ {\isacharbang}\ i{\isacharparenright}{\isacharparenright}  {\isasymand}\ 

\ \ \ {\isacharparenleft}bits\ {\isacharparenleft}levels\ mp\ {\isacharbang}\ i{\isacharparenright}\ {\isacharbang}\ j\ {\isacharequal}\ FREE\ {\isasymor}\ bits\ {\isacharparenleft}levels\ mp\ {\isacharbang}\ i{\isacharparenright}\ {\isacharbang}\ j\ {\isacharequal}\ FREEING\ {\isasymor} \isanewline
\ \ \ \ \ \ bits\ {\isacharparenleft}levels\ mp\ {\isacharbang}\ i{\isacharparenright}\ {\isacharbang}\ j\ {\isacharequal}\ ALLOCATED\ {\isasymor}\ bits\ {\isacharparenleft}levels\ mp\ {\isacharbang}\ i{\isacharparenright}\ {\isacharbang}\ j\ {\isacharequal}\ ALLOCATING{\isacharparenright} {\isasymand}\ 

\ \ \  addr\ {\isasymin}\ {\isacharbraceleft}j\ {\isacharasterisk}\ {\isacharparenleft}max{\isacharunderscore}sz\ mp\ div\ {\isacharparenleft}{\isadigit{4}}\ {\isacharcircum}\ i{\isacharparenright}{\isacharparenright}\ {\isachardot}{\isachardot}{\isacharless}\ (j+1)\ {\isacharasterisk}\ {\isacharparenleft}max{\isacharunderscore}sz\ mp\ div\ {\isacharparenleft}{\isadigit{4}}\ {\isacharcircum}\ i{\isacharparenright}{\isacharparenright}{\isacharbraceright}

\isacommand{mem{\isacharunderscore}part}\ s\ {\isasymequiv}\ {\isasymforall}p{\isasymin}mem{\isacharunderscore}pools\ s{\isachardot}\ \textbf{let}\ mp\ {\isacharequal}\ mem{\isacharunderscore}pool{\isacharunderscore}info\ s\ p\ \textbf{in}

\quad \quad \quad \quad \quad \quad \quad \quad {\isacharparenleft}{\isasymforall}addr\ {\isacharless}\ n{\isacharunderscore}max\ mp\ {\isacharasterisk}\ max{\isacharunderscore}sz\ mp{\isachardot}\ {\isacharparenleft}{\isasymexists}{\isacharbang}{\isacharparenleft}i{\isacharcomma}j{\isacharparenright}{\isachardot}\ \textbf{addr{\isacharunderscore}in{\isacharunderscore}block}\ mp\ addr\ i\ j{\isacharparenright}\ {\isacharparenright}
\end{isabellec}
\vspace{0.2cm}

From the invariants of bitmaps and memory configuration in the previous subsection, we can derive the general property for the memory partition as the following theorem. 

%\vspace{-3mm}
\begin{theorem}[Memory Partition]
For any kernel state $s$, If the memory pools in $s$ are consistent in their configuration, and their bitmaps are well-shaped, the memory pools satisfy the partition property in $s$, i.e. 
%\vspace{-3mm}
\[
\textbf{inv\_mempool\_info}(s) \wedge \textbf{inv\_bitmap}(s) \wedge \textbf{inv\_bitmap0}(s) \wedge \textbf{inv\_bitmapn}(s) \Longrightarrow \textbf{mem\_part}(s)
\]
\end{theorem}

That is, the proof for the memory separation preservation at the memory-block level is discharged by the proof of the invariants preservation. 

\subsubsection{Memory Separation at Thread Level}
When a thread requests for memory blocks, the memory management allocates some block to that thread by adding it to the local memory of the thread. As a consequence the thread becomes the exclusive owner of the block, which can only be freed by the same thread that the block belongs to. %We define a strong separation property for threads in this article: allocated memory blocks are local to threads and they are not shared among them. If the memory separation property is preserved by the memory management in such an extreme case, it will be preserved in the case that memory blocks may be shared among threads. 
For this purpose, we define a mapping of allocated memory blocks of threads in the system state of the memory management as 
\[mblocks :: Thread \Rightarrow Mem\_block\ set\]

The execution of memory services by a thread $t$ cannot affect the allocated memory blocks of other threads, i.e. for arbitrary states $s$ and $r$, which are the previous state and the following state of an execution step of $t$ respectively, the memory separation requires that 
\[\forall t'. t' \neq t \longrightarrow mblocks\ s\ t' = mblocks\ r\ t' \]

In {\sectprefix} \ref{subsect:securityproof}, we will show that the separation property implies integrity of the memory management services. 

\section{Formalizing Zephyr Memory Management}
\label{sect:mem_spec}

In this section, we first introduce an execution model of Zephyr using the {\slang} language. Then we discuss in detail the low-level design specification for the kernel services that the memory management provides. Since this work focuses on the memory management, we only provide very abstract models for other kernel functionalities such as the kernel scheduling and thread control.

\subsection{Event-based Execution Model of Zephyr}

\subsubsection{{\slang} statements for memory management}
Interrupt handlers in OSs are considered as reaction services of {\slang} representing as \emph{events}: 
\[\stmtevent{\symbEvt}{p_1,...,p_n}{\symbcore}{g}{\symbprog}\]

In addition to the input parameters, an event has a special parameter $\symbcore$ indicating the execution context, e.g. the scheduler and the thread invoking the event. 
The imperative commands of an event body $P$ in {\slang} are standard sequential constructs such as conditional execution, loop, and sequential composition of programs. It also includes a synchronization construct for concurrent processes represented by $\stmtawait{b}{\symbprog}$. The body $\symbprog$ is executed atomically if and only if the boolean condition $\symbbexp$ holds, not progressing otherwise. $\stmtatom{\symbprog}$ denotes an \emph{Await} statement for which its guard is $True$. 

Threads and kernel processes have their own execution context and local states. Each of them is modelled in {\slang} as a set of events called \emph{event systems} and denoted as $\eventsystem{\symbevtsys}$.
The operational semantics of an event system is the \emph{sequential composition} of the execution of the events composing it. 
%It consists in the continuous evaluation of the guards of the system events. From the set of events for which the associated guard $g$ holds in the current state, one event $\symbEvt$ is non-deterministically selected to be triggered, and its body $P$ executed. After $P$ finishes, the evaluation of the guards starts again looking for the next event to be executed. 
Finally, {\slang} has a construct for parallel composition of event systems $esys_0 \parallel ... \parallel esys_n$ which interleaves the execution of the events composing each event system  $esys_i$ for $0 \leq i\leq n$.

\subsubsection{Execution model of Zephyr}

\begin{figure}[t]
\begin{center}
\includegraphics[width=4.8in]{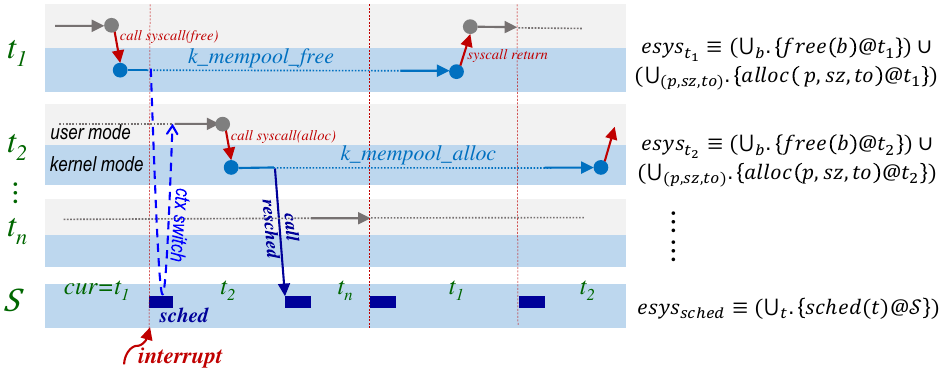}
\end{center}
\caption{An Execution Model of Zephyr Memory Management}
\label{fig:kernel_model}
\end{figure}

After being initialized, an OS kernel can be considered as a reactive system that is in an \emph{idle} loop until it receives an interruption which is handled by an interruption handler. 
Whilst an interrupt handler execution is atomic in sequential kernels, it can be interrupted in concurrent kernels~\cite{Chen16,Xu16} allowing services invoked by threads to be interrupted and resumed later. 
In the execution model of Zephyr, we consider a scheduler $\mathcal{S}$, a timer and a set of threads $t_1, ..., t_n$. In this model, the execution of the scheduler is atomic since kernel services can not interrupt it. But kernel services can be interrupted via the scheduler, i.e., the execution of a memory service invoked by a thread $t_i$ may be interrupted by the kernel scheduler to execute a thread $t_j$. {\figprefix}\ref{fig:kernel_model} illustrates Zephyr execution model, where solid lines represent execution steps of the threads/kernel services and dotted lines mean the suspension of the thread/code. For instance, the execution of \emph{k\_mempool\_free} in thread $t_1$ is interrupted by the scheduler, and the context is switched to thread $t_2$, which invokes \emph{k\_mempool\_alloc}. During the execution of $t_2$, the kernel service may suspend the thread and switch to another thread $t_n$ by calling \emph{rescheduling}. Later, the execution is switched back to $t_1$ and continues the execution of \emph{k\_mempool\_free} in a different state from when it was interrupted.

The event systems of Zephyr are illustrated in the right part of {\figprefix} \ref{fig:kernel_model}. A user thread $t_i$ invoke allocation/release services, thus the event system for $t_i$ is $esys_{t_i}$, a set composed of the events \emph{alloc} and \emph{free}. The input parameters for these events correspond with the arguments of the service implementation, that are constrained by the guard for each service. Together with system users we model the event service for the scheduler $esys_{sched}$ consisting on a unique event \emph{sched} whose argument is a thread $t$ to be scheduled when $t$ is in the \emph{READY} state. 
The formal specification of the memory management is the parallel composition of the event system for the threads, the scheduler and the timer. % $esys_{t_1} \parallel ... \parallel esys_{t_n} \parallel esys_{sched}$

\subsubsection{Thread context and preemption}

Events are parametrized by a thread identifier used to access to the execution context of the thread invoking it. As shown in Figure~\ref{fig:kernel_model}, the execution of an event executed by a thread can be stopped by the scheduler to be resumed later. This behaviour is modelled using a global variable $cur$ indicating that the thread being currently has been scheduled and it is being executed. The model conditions the execution of parametrized events in $t$ only when $t$ is scheduled. This is achieved by using the expression \stmtirq{t}{p} $\equiv$ $\stmtawait{cur = t}{p}$, so an event invoked by a thread $t$ only progresses when $t$ is scheduled. This scheme allows to use rely-guarantee for concurrent execution of threads on mono-core architectures, where only the scheduled thread is able to modify the memory.

\subsection{Formal Specification of Memory Management Services}

This section discusses the formal specification of the memory management services. These services deal with the initialization of pools, and memory allocation and release. 

\subsubsection{System state}
The system state includes the memory model introduced in Section~\ref{sect:mem_struct}, together with the thread under execution represented by the variable $cur$ and the local variables to the memory services. The local variables are used to keep temporal changes to the structure, guards in conditional and loop statements, and index accesses. 
The memory model is represented as a set \emph{mem\_pools} storing the references of all memory pools and a mapping \emph{mem\_pool\_info} to query a pool by a pool reference. 
Local variables are modelled as total functions from threads to variable values, representing that the event is accessing the thread context. In the formal model of the events we represent the access to a state component $c$ with ${\isasymacute} c$ and the value of a local component $c$ for the thread $t$ is represented as ${\isasymacute}c\ t$.
Local variables \emph{allocating\_node} and \emph{freeing\_node} are relevant for the memory services, storing the temporal blocks being split/coalesced in alloc/release services respectively. The memory blocks allocated in a thread are stored in the local variable \emph{mblocks} as discussed in the previous section.

\subsubsection{Memory pool initialization}
Zephyr defines and initializes memory pools at compile time by constructing a static variable of type \emph{\textbf{struct} k\_mem\_pool}. The implementation initializes each pool with \emph{n\_max} level 0 blocks with size \emph{max\_sz} bytes. Bitmaps of level 0 are set to 1 and free list contains all level 0 blocks. Bitmaps and free lists of other level are initialized to 0 and to the empty list respectively. In the formal model, we specify a state corresponding to the implementation initial state and we show that it belongs to the set of states satisfying the invariant.

\subsubsection{Memory allocation/release services}
The C code of Zephyr uses the recursive function \emph{free\_block} to coalesce free partner blocks and the \emph{break} statement to stop the execution of a loop statements, which are not supported by the imperative language in {\slang}. The formal specification overcomes this by transforming the recursion into a loop controlled by the recursion condition, and using a control variable to exit loops with breaks when the condition to execute the loop break is satisfied. 
Additionally, the memory management services use the atomic body \emph{irq\_lock(); P; irq\_unlock();} to keep interruption handlers \emph{reentrant} by disabling interruptions. 
We simplify this behaviour in the specification using an \textbf{ATOM} statement, avoiding the service to be interrupted at that point. The rest of the formal specification closely follows the implementation, where variables are modified using higher order functions changing the state as the code does it. The reason of using Isabelle/HOL functions is that {\slang} does not provide a semantic for expressions, using instead state transformer relying on high order functions to change the state.

\begin{figure}[t]
\begin{flushleft}
\begin{isabellec} \fontsize{7pt}{0cm} %\footnotesize %\scriptsize
\ \ 1 \ \isacommand{WHILE}\ {\isasymacute}free{\isacharunderscore}block{\isacharunderscore}r\ t\ \isacommand{DO}\isanewline
\ \ 2 \ \quad t\ {\isactrlenum} \ {\isasymacute}lsz\ {\isacharcolon}{\isacharequal}\ {\isasymacute}lsz\ {\isacharparenleft}t\ {\isacharcolon}{\isacharequal}\ {\isasymacute}lsizes\ t\ {\isacharbang}\ {\isacharparenleft}{\isasymacute}lvl\ t{\isacharparenright}{\isacharparenright}{\isacharsemicolon}{\isacharsemicolon}\isanewline
\ \ 3 \ \quad t\ {\isactrlenum} \ {\isasymacute}blk\ {\isacharcolon}{\isacharequal}\ {\isasymacute}blk\ {\isacharparenleft}t\ {\isacharcolon}{\isacharequal}\ block{\isacharunderscore}ptr\ {\isacharparenleft}{\isasymacute}mem{\isacharunderscore}pool{\isacharunderscore}info\ {\isacharparenleft}pool\ b{\isacharparenright}{\isacharparenright}\ {\isacharparenleft}{\isasymacute}lsz\ t{\isacharparenright}\ {\isacharparenleft}{\isasymacute}bn\ t{\isacharparenright}{\isacharparenright}{\isacharsemicolon}{\isacharsemicolon}\isanewline
\ \ 4 \ \quad t\ {\isactrlenum} \ \isacommand{ATOM}\isanewline
\ \ 5 \ \quad \quad {\isasymacute}mem{\isacharunderscore}pool{\isacharunderscore}info\ {\isacharcolon}{\isacharequal}\ set{\isacharunderscore}bit{\isacharunderscore}free\ {\isasymacute}mem{\isacharunderscore}pool{\isacharunderscore}info\ {\isacharparenleft}pool\ b{\isacharparenright}\ {\isacharparenleft}{\isasymacute}lvl\ t{\isacharparenright}\ {\isacharparenleft}{\isasymacute}bn\ t{\isacharparenright}{\isacharsemicolon}{\isacharsemicolon}\isanewline
\ \ 6 \ \quad \quad {\isasymacute}freeing{\isacharunderscore}node\ {\isacharcolon}{\isacharequal}\ {\isasymacute}freeing{\isacharunderscore}node\ {\isacharparenleft}t\ {\isacharcolon}{\isacharequal}\ None{\isacharparenright}{\isacharsemicolon}{\isacharsemicolon}\isanewline
\ \ 7 \ \quad \quad \isacommand{IF}\ {\isasymacute}lvl\ t\ {\isachargreater}\ {\isadigit{0}}\ {\isasymand}\ partner{\isacharunderscore}bits\ {\isacharparenleft}{\isasymacute}mem{\isacharunderscore}pool{\isacharunderscore}info\ {\isacharparenleft}pool\ b{\isacharparenright}{\isacharparenright}\ {\isacharparenleft}{\isasymacute}lvl\ t{\isacharparenright}\ {\isacharparenleft}{\isasymacute}bn\ t{\isacharparenright}\ \isacommand{THEN}\isanewline
\ \ 8 \ \quad \quad \quad \isacommand{FOR}\ {\isasymacute}i\ {\isacharcolon}{\isacharequal}\ {\isasymacute}i{\isacharparenleft}t\ {\isacharcolon}{\isacharequal}\ {\isadigit{0}}{\isacharparenright}{\isacharsemicolon}\ {\isasymacute}i\ t\ {\isacharless}\ {\isadigit{4}}{\isacharsemicolon}\ {\isasymacute}i\ {\isacharcolon}{\isacharequal}\ {\isasymacute}i{\isacharparenleft}t\ {\isacharcolon}{\isacharequal}\ {\isasymacute}i\ t\ {\isacharplus}\ {\isadigit{1}}{\isacharparenright}\ \isacommand{DO}\isanewline
\ \ 9 \ \quad \quad \quad \quad {\isasymacute}bb\ {\isacharcolon}{\isacharequal}\ {\isasymacute}bb\ {\isacharparenleft}t\ {\isacharcolon}{\isacharequal}\ {\isacharparenleft}{\isasymacute}bn\ t\ div\ {\isadigit{4}}{\isacharparenright}\ {\isacharasterisk}\ {\isadigit{4}}\ {\isacharplus}\ {\isasymacute}i\ t{\isacharparenright}{\isacharsemicolon}{\isacharsemicolon}\isanewline
10 \ \quad \quad \quad \quad {\isasymacute}mem{\isacharunderscore}pool{\isacharunderscore}info\ {\isacharcolon}{\isacharequal}\ set{\isacharunderscore}bit{\isacharunderscore}noexist\ {\isasymacute}mem{\isacharunderscore}pool{\isacharunderscore}info\ {\isacharparenleft}pool\ b{\isacharparenright}\ {\isacharparenleft}{\isasymacute}lvl\ t{\isacharparenright}\ {\isacharparenleft}{\isasymacute}bb\ t{\isacharparenright}{\isacharsemicolon}{\isacharsemicolon}\isanewline
11 \ \quad \quad \quad \quad {\isasymacute}block{\isacharunderscore}pt\ {\isacharcolon}{\isacharequal}\ {\isasymacute}block{\isacharunderscore}pt\ {\isacharparenleft}t\ {\isacharcolon}{\isacharequal}\ block{\isacharunderscore}ptr\ {\isacharparenleft}{\isasymacute}mem{\isacharunderscore}pool{\isacharunderscore}info\ {\isacharparenleft}pool\ b{\isacharparenright}{\isacharparenright}\ {\isacharparenleft}{\isasymacute}lsz\ t{\isacharparenright}\ {\isacharparenleft}{\isasymacute}bb\ t{\isacharparenright}{\isacharparenright}{\isacharsemicolon}{\isacharsemicolon}\isanewline
12 \ \quad \quad \quad \quad \isacommand{IF}\ {\isasymacute}bn\ t\ {\isasymnoteq}\ {\isasymacute}bb\ t\ {\isasymand}\ block{\isacharunderscore}fits\ {\isacharparenleft}{\isasymacute}mem{\isacharunderscore}pool{\isacharunderscore}info\ {\isacharparenleft}pool\ b{\isacharparenright}{\isacharparenright}\ {\isacharparenleft}{\isasymacute}block{\isacharunderscore}pt\ t{\isacharparenright}\ {\isacharparenleft}{\isasymacute}lsz\ t{\isacharparenright}\ \isacommand{THEN}\isanewline
13 \ \quad \quad \quad \quad \quad {\isasymacute}mem{\isacharunderscore}pool{\isacharunderscore}info\ {\isacharcolon}{\isacharequal}\ {\isasymacute}mem{\isacharunderscore}pool{\isacharunderscore}info\ {\isacharparenleft}{\isacharparenleft}pool\ b{\isacharparenright}\ {\isacharcolon}{\isacharequal}\ \isanewline
14 \ \quad \quad \quad \quad \quad \quad \quad remove{\isacharunderscore}free{\isacharunderscore}list\ {\isacharparenleft}{\isasymacute}mem{\isacharunderscore}pool{\isacharunderscore}info\ {\isacharparenleft}pool\ b{\isacharparenright}{\isacharparenright}\ {\isacharparenleft}{\isasymacute}lvl\ t{\isacharparenright}\ {\isacharparenleft}{\isasymacute}block{\isacharunderscore}pt\ t{\isacharparenright}{\isacharparenright}\isanewline
15 \ \quad \quad \quad \quad \isacommand{FI}\isanewline
16 \ \quad \quad \quad \isacommand{ROF}{\isacharsemicolon}{\isacharsemicolon}\isanewline
17 \ \quad \quad \quad {\isasymacute}lvl\ {\isacharcolon}{\isacharequal}\ {\isasymacute}lvl\ {\isacharparenleft}t\ {\isacharcolon}{\isacharequal}\ {\isasymacute}lvl\ t\ {\isacharminus}\ {\isadigit{1}}{\isacharparenright}{\isacharsemicolon}{\isacharsemicolon}\isanewline
18 \ \quad \quad \quad {\isasymacute}bn\ {\isacharcolon}{\isacharequal}\ {\isasymacute}bn\ {\isacharparenleft}t\ {\isacharcolon}{\isacharequal}\ {\isasymacute}bn\ t\ div\ {\isadigit{4}}{\isacharparenright}{\isacharsemicolon}{\isacharsemicolon}\isanewline
19 \ \quad \quad \quad {\isasymacute}mem{\isacharunderscore}pool{\isacharunderscore}info\ {\isacharcolon}{\isacharequal}\ set{\isacharunderscore}bit{\isacharunderscore}freeing\ {\isasymacute}mem{\isacharunderscore}pool{\isacharunderscore}info\ {\isacharparenleft}pool\ b{\isacharparenright}\ {\isacharparenleft}{\isasymacute}lvl\ t{\isacharparenright}\ {\isacharparenleft}{\isasymacute}bn\ t{\isacharparenright}{\isacharsemicolon}{\isacharsemicolon}\isanewline
20 \ \quad \quad \quad {\isasymacute}freeing{\isacharunderscore}node\ {\isacharcolon}{\isacharequal}\ {\isasymacute}freeing{\isacharunderscore}node\ {\isacharparenleft}t\ {\isacharcolon}{\isacharequal}\ Some\ {\isasymlparr}pool\ {\isacharequal}\ {\isacharparenleft}pool\ b{\isacharparenright}{\isacharcomma}\ level\ {\isacharequal}\ {\isacharparenleft}{\isasymacute}lvl\ t{\isacharparenright}{\isacharcomma}\ \isanewline
21 \ \ \ \ \ \ \ \ \ \ \ \ \ \ \ \ \ \ \ \ \ block\ {\isacharequal}\ {\isacharparenleft}{\isasymacute}bn\ t{\isacharparenright}{\isacharcomma}\ 
data\ {\isacharequal}\ block{\isacharunderscore}ptr\ {\isacharparenleft}{\isasymacute}mem{\isacharunderscore}pool{\isacharunderscore}info\ {\isacharparenleft}pool\ b{\isacharparenright}{\isacharparenright}\ \isanewline
22 \ \ \ \ \ \ \ \ \ \ \ \ \ \ \ \ \ \ \ \ \ \ \  {\isacharparenleft}{\isacharparenleft}{\isacharparenleft}ALIGN{\isadigit{4}}\ {\isacharparenleft}max{\isacharunderscore}sz\ {\isacharparenleft}{\isasymacute}mem{\isacharunderscore}pool{\isacharunderscore}info\ {\isacharparenleft}pool\ b{\isacharparenright}{\isacharparenright}{\isacharparenright}{\isacharparenright}\ div\ {\isacharparenleft}{\isadigit{4}}\ {\isacharcircum}\ {\isacharparenleft}{\isasymacute}lvl\ t{\isacharparenright}{\isacharparenright}{\isacharparenright}{\isacharparenright}\ 
{\isacharparenleft}{\isasymacute}bn\ t{\isacharparenright}\ {\isasymrparr}{\isacharparenright}\isanewline
23 \ \quad \quad \isacommand{ELSE}\isanewline
24 \ \quad \quad \quad \isacommand{IF}\ block{\isacharunderscore}fits\ {\isacharparenleft}{\isasymacute}mem{\isacharunderscore}pool{\isacharunderscore}info\ {\isacharparenleft}pool\ b{\isacharparenright}{\isacharparenright}\ {\isacharparenleft}{\isasymacute}blk\ t{\isacharparenright}\ {\isacharparenleft}{\isasymacute}lsz\ t{\isacharparenright}\ \isacommand{THEN}\isanewline
25 \ \quad \quad \quad \quad {\isasymacute}mem{\isacharunderscore}pool{\isacharunderscore}info\ {\isacharcolon}{\isacharequal}\ {\isasymacute}mem{\isacharunderscore}pool{\isacharunderscore}info\ {\isacharparenleft}{\isacharparenleft}pool\ b{\isacharparenright}\ {\isacharcolon}{\isacharequal}\ \isanewline
26 \ \quad \quad \quad \quad \quad append{\isacharunderscore}free{\isacharunderscore}list\ {\isacharparenleft}{\isasymacute}mem{\isacharunderscore}pool{\isacharunderscore}info\ {\isacharparenleft}pool\ b{\isacharparenright}{\isacharparenright}\ {\isacharparenleft}{\isasymacute}lvl\ t{\isacharparenright}\ {\isacharparenleft}{\isasymacute}blk\ t{\isacharparenright}\ {\isacharparenright}\isanewline
27 \ \quad \quad \quad \isacommand{FI}{\isacharsemicolon}{\isacharsemicolon}\isanewline
28 \ \quad \quad \quad {\isasymacute}free{\isacharunderscore}block{\isacharunderscore}r\ {\isacharcolon}{\isacharequal}\ {\isasymacute}free{\isacharunderscore}block{\isacharunderscore}r\ {\isacharparenleft}t\ {\isacharcolon}{\isacharequal}\ False{\isacharparenright}\isanewline
29 \ \quad \quad \isacommand{FI}\isanewline
30 \ \quad \isacommand{END}\isanewline
31 \ \isacommand{OD}
\end{isabellec} 

%\end{minipage}
\caption{The {\slang} Specification of \emph{free\_block}}
\label{fig:freeblock}
\end{flushleft}
\end{figure}

{\figprefix} \ref{fig:freeblock} illustrates the {\slang} specification of the \emph{free\_block} function invoked by \emph{k\_mem\_pool\_free} when releasing a memory block. The code accesses  the following variables: $lsz$, $lsize$, and $lvl$ to keep information about the current level; $blk$, $bn$, and $bb$ to represent the address and number of the block currently being accessed; $freeing\_node$ to represent the node being freeing; and $i$ to iterate blocks. Additionally, the model includes the component $free\_block\_r$ to model the recursion condition. To simplify the representation the model uses predicates and functions to access and modify the state. We refer readers to the Isabelle/HOL sources for the complete specification of these functions and the complete formal model. 

In the C code, \emph{free\_block} is a recursive function with two conditions: (1) the block being released belongs to a level higher than zero, since blocks at level zero cannot be merged; and (2) the partners bits of the block being released are FREE so they can be merged into a bigger block. We represent (1) with the predicate ${\isasymacute}lvl\ t\ {\isachargreater}\ {\isadigit{0}}$ and (2) with the predicate $partner\_bit\_free$. The formal specification follows the same structure translating the recursive function into a loop that is controlled by a variable mimicking the recursion. The recursive process of \emph{free\_block} is illustrated in {\figprefix} \ref{fig:freeblockproc}. 

\begin{figure}[t]
\begin{center}
\includegraphics[width=4.0in]{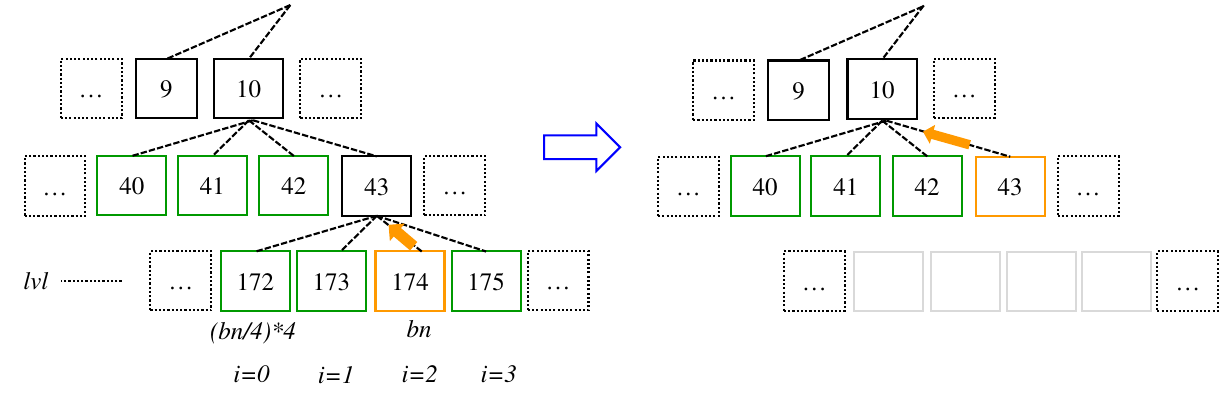}
\end{center}
\caption{Coalescing Memory Blocks in \emph{free\_block}}
\label{fig:freeblockproc}
\end{figure}

The formal specification for  \emph{free\_block} first releases an allocated memory block $bn$ setting it to \emph{FREEING}. Then, the loop statement sets \emph{free\_block} to \emph{FREE} (Line 5 in {\figprefix} \ref{fig:freeblock}), and also checks that the iteration/recursive condition holds in Line 7. If the condition holds, the partner bits are set  to \emph{NOEXIST}, and remove their addresses from the free list for this level (Lines 12 - 15). Then, it sets the parent block bit to \emph{FREEING} (Lines 17 - 22), and updates the variables controlling the current block and level numbers, before going back to the beginning of the loop again. If the iteration condition is not true it sets the bit to \emph{FREE} and add the block to the free list (Lines 24 - 28) and sets the loop condition to false to end the procedure. 
This function is illustrated in {\figprefix} \ref{fig:freeblockproc}. The block $174$ is released by a thread and since its partner blocks (block $172$, $173$ and $175$) are free, Zephyr coalesces the four blocks and sets their parent block $43$ as \emph{FREEING}. The coalescence continues iteratively if the partners of block $43$ are all free. 

The main part of the C code of the \emph{k\_mem\_pool\_free} service and its complete formalization are shown in Appendices \ref{appx:mempoolfree_c} and \ref{appx:mempoolfree} respectively.

\subsubsection{Formal specification of the memory management}
The {\slang} specification of the memory management of Zephyr is finally defined as follows. The events for the scheduler and the timer is simple. The \emph{schedule} event chooses a \emph{READY} thread \emph{t} to be executing and set the current thread to \emph{READY}. The \emph{tick} event just increases the \emph{tick} variable in system state by one. The \emph{tick} variable is used for the \emph{time out} waiting mode for memory allocation. 
\begin{equation}
\label{eq:mem_spec}
\footnotesize
\begin{aligned}
Mem\_Spec \equiv \lambda \symbcore.\ \textbf{case}\ \symbcore\ \textbf{of}\ & (\textbf{\isasymT}\ t) \ \isasymRightarrow\ \bigcup {b}.\ \{free(b)@t\} \cup \bigcup {(p,sz,to)}.\ \{alloc(p,sz,to)@t\} \\
\isacharbar \ & \textbf{\isasymS}\ \isasymRightarrow\ \bigcup {t}.\ \{schedule(t)\} \\
\isacharbar \ & \textbf{Timer}\ \isasymRightarrow\ \{tick\}
\end{aligned}
\end{equation}

\section{Compositional Reasoning about Integrity in {\slang}}
\label{sect:ifs}
In this section, we present a compositional reasoning approach for the verification of integrity in {\slang}. 
We use the notion of integrity from \cite{rushby92,Murray12}, which provides a formalism for the specification of security policies.
We first define the integrity of {\slang} specification. Then, we show that reasoning about the integrity of the system can be decomposed to events. %Finally, we instantiate the approach to the memory management of Zephyr. 
For convenience, we first introduce the operational semantics and computations of {\slang} in brief, which is the foundation of the integrity and its compositional reasoning. 

\subsection{Operational Semantics and Computations of {\slang}}
The semantics of {\slang} is defined via transition rules between configurations. We define a configuration $\symbconf$ in {\slang} as a triple  $(\symbspec, \symbstate, \symbevtctx)$ where $\symbspec$ is a specification, $\symbstate: S$ is a system state, and $\symbevtctx : \symbCore \rightarrow \symbEvt$ is an event context. The event context indicates which event is currently being executed in an event system $\symbcore$. $\symbspec_{\symbconf}$, $\symbstate_{\symbconf}$, and $\symbevtctx_{\symbconf}$ represent the projection of each component in the tuple $\symbconf =(\symbspec, \symbstate, \symbevtctx)$. 
Transition rules in events, event systems, and parallel event systems have the form $\transenv (\symbspec_1, \symbstate_1, \symbevtctx_1) \tranb{\symbactk} (\symbspec_2, \symbstate_2, \symbevtctx_2)$, where $\symbactk = \actk{\symbact}{\symbcore}$ is a label indicating the type of transition, the subscript ``$\square$'' ($_e$, $_{es}$ or $_{pes}$) indicates the transition objects, and $\progenv$ is used for some static configuration for programs (e.g. an environment for procedure declarations). Here $\symbact$ indicates a program action $c$ or an occurrence of an event $\symbEvt$. $@\symbcore$ means that the action occurs in event system $\symbcore$. %An action that occurs in event system $\symbcore$ is denoted as $\symbactk = \actk{\symbact}{\symbcore}$. 
%The program transition is denoted as $\tranp$ in the $\textsc{TrgdEvt}$ rule. 
Environment transition rules have the form $\transenv (\symbspec, \symbstate, \symbevtctx) \tranb{env} (\symbspec, \symbstate', \symbevtctx')$. Intuitively, a transition made by the environment may change the state but not the event context nor the specification. 
%Transition rules of {\slang} are shown in {\figprefix} \ref{fig:semantics}. 
The parallel composition of event systems is fine-grained since small steps in events are interleaved in the semantics of {\slang}. 
%This design relaxes the atomicity of events in other approaches (e.g., Event-B \cite{Abrial07}).

A \emph{computation} of {\slang} is a sequence of transitions. 
We define the set of computations of all parallel event systems with static information $\progenv$ as $\compfun(\progenv)$, which is a set of lists of configurations inductively defined as follows. The singleton list is always a computation (1). Two consecutive configurations are part of a computation if they are the initial and final configurations of an environment (2) or action transition (3). 

\[\footnotesize
\left\{
\begin{aligned}
& (1)[(\symbpes, \symbstate, \symbevtctx)] \in  \compfun(\progenv) \\
& (2)(\symbpes, \symbstate_1, \symbevtctx_1)\#cs \in  \compfun(\progenv)  \Longrightarrow (\symbpes, \symbstate_2, \symbevtctx_2)\#(\symbpes, \symbstate_1, \symbevtctx_1)\#cs \in  \compfun(\progenv) \\
& (3)\transenv (\symbpes_2, \symbstate_2, \symbevtctx_2) \trant{\symbactk}{pes} (\symbpes_1, \symbstate_1, \symbevtctx_1) \wedge (\symbpes_1, \symbstate_1, \symbevtctx_1)\#cs \in  \compfun(\progenv) \\
& \quad \quad \quad \Longrightarrow (\symbpes_2, \symbstate_2, \symbevtctx_2)\#(\symbpes_1, \symbstate_1, \symbevtctx_1)\#cs \in  \compfun(\progenv)
\end{aligned}
\right.
\]

Computations for events and event systems are defined in a similar way. We use $\compfun(\progenv,\symbpes)$ to denote the set of computations of a parallel event system $\symbpes$. The function $\compfun(\progenv,\symbpes, \symbstate, \symbevtctx)$ denotes the computations of $\symbpes$ starting up from an initial state $\symbstate$ and event context $\symbevtctx$.

In {\slang}, the semantics is compositional. A computation $\symbcomp$ of $\symbpes$ could be decomposed into a set of computations $\tilde{\symbcomp}$ of its event systems. Computations in $\tilde{\symbcomp}$ have the same state and event context sequence. They do not have component transitions at the same time. $\symbcomp$ also has the same state and event context sequence as $\tilde{\symbcomp}$. Furthermore, in $\symbcomp$ a transition is labelled as $\symbactk$ if this is the label in one of the computations $\tilde{\symbcomp}$ at the corresponding position; a transition is an environment transition if this is the case in all computations $\tilde{\symbcomp}$ at the corresponding position. We use the \emph{conjoin} notation $\symbcomp\ \isasympropto\ \tilde{\symbcomp}$ to present this compositionality, and $\tilde{\symbcomp}^{\symbcore}$ denotes the computation of $\symbcore$. 
%Formally, any computation $\symbcomp \in \compfun(\progenv,\symbpes, \symbstate, \symbevtctx)$, there is a set of computations of its event systems $\tilde{\symbcomp}$ conjoined by $\symbcomp$ ($\symbcomp\ \isasympropto\ \tilde{\symbcomp}$). The computation of $\symbcore$ is denoted as $\tilde{\symbcomp}^{\symbcore}$ and $\tilde{\symbcomp}^{\symbcore} \in \compfun(\progenv, \symbpes(\symbcore), \symbstate, \symbevtctx)$. 

\subsection{Integrity in {\slang}}
\label{subsect:ifs}
In general, the definition of integrity relies on a security configuration for a system, which is actually the security policies, as well as a state machine of the system. 

\subsubsection{Security Configuration}
In order to discuss the security of a {\slang} specification $\symbpes$, we assume a set of security domains $\symbDomain$ and a security policy $\interf$ that restricts the allowable flow of information among those domains. The security policy $\interf$ is a reflexive relation on $\symbDomain$. 
$\symbdomain_1 \interf \symbdomain_2$ means that actions performed by $\symbdomain_1$ can influence subsequent outputs seen by $\symbdomain_2$. $\ninterf$ is the complement relation of $\interf$. 
We call $\interf$ and $\ninterf$ the \emph{interference} and \emph{noninterference} relations respectively.

Each event has associated an execution domain responsible of invoking the event. Traditional formulations in information-flow security assume a static mapping from events to domains, such that the domain of an event can be determined solely from the event itself \cite{rushby92,Oheimb04}. For flexibility, we use a dynamic mapping, which is represented by a function $dom\_e: S \times \symbCore \times \symbEvt \rightarrow \symbDomain$, and $dom\_e(\symbstate, \symbcore,\symbevt)$ means the execution domain of event $\symbevt$ on context $\symbcore$ in state $\symbstate$. 
The $\symbpes$ is \emph{view-partitioned} if, for each domain $\symbdomain \in \symbDomain$, there is an equivalent relation $\dsim{\symbdomain}$ on $\symbState$. For convenience, we define $\symbconf_1 \dsim{\symbdomain} \symbconf_2 \defi \symbstate_{\symbconf_1} \dsim{\symbdomain} \symbstate_{\symbconf_2}$.
%, where $\symbconf$ is a configuration in the {\slang} semantics as a triple  $(\symbspec, \symbstate, \symbevtctx)$ where $\symbspec$ is a specification, $\symbstate$ is a state, and $\symbevtctx : \symbCore \rightarrow \symbEvt$ is an event context. The event context indicates which event is currently being executed in an event system $\symbcore$. $\symbstate_{\symbconf}$ indicates the projection of the state in the configuration $\symbconf$. 
%For a set of domains $D$, let $\mathcal{C}_1 \ddsim{D} \mathcal{C}_2 \equiv \forall d \in D.\ \mathcal{C}_1 \dsim{d} \mathcal{C}_2$.
%An observation function of a domain $\symbdomain$ to a state $\symbstate$ is defined as $ob(\symbstate,\symbdomain)$. For convenience, we define $ob(\symbconf,\symbdomain) \defi ob(\symbstate_\symbconf,\symbdomain)$.

\subsubsection{State Machine Representation of {\slang} Specification}
The state-event IFS is usually defined on a state machine. Here, we construct an equivalent state machine for a {\slang} specification. 
The state of the machine is the configuration in the {\slang} semantics. The security of {\slang} consider small-step actions of systems. A small-step action in the machine is identified by the label of a transition, the event that the action belongs to, and the domain that triggers the event. A small-step action is thus in the form of a triple $(\symbactk, \symbevt, \symbdomain)$. 
The label of the transition can represent an occurrence of an event $\symbevt_{\symbaction}$, or $c$ when is an internal step of an event already triggered. In that case $\symbevtctx_\symbconf(\symbcore)$ stores the event that is being executed. 
We construct a nondeterministic state machine for a closed {\slang} specification as follows. Here, a closed specification means that we dont consider the environment of the whole system, i.e. the rely condition of the system is the identity relation. 

\begin{definition}
\label{def:statemachine}
A state machine of a specification $\symbpes$ is a quadruple $\symbSM=\langle \symbConf, \symbAction, step, \symbconf_0 \rangle$, where 
\begin{itemize}
\item $\symbConf$ is the set of configurations.

\item $\symbAction$ is the set of actions. An action is a triple $\symbaction = \langle \symbactk, \symbevt, \symbdomain \rangle$, where $\symbactk$ is a transition label in the {\slang} semantics, $\symbevt$ is an event, and $\symbdomain$ is a domain. The notations $\symbactk_{\symbaction}$, $\symbevt_{\symbaction}$ and $\symbdomain_{\symbaction}$ respectively denote the projections of the components of an action $\symbaction$. 

\item $step: \symbAction \rightarrow \mathbb{P}(\symbConf \times \symbConf)$ is the transition function, where $step(\symbaction) = \{(\symbconf,\symbconf') \mid \transenv \symbconf \trant{\symbactk_\symbaction}{pes} \symbconf' \wedge ((\symbactk_\symbaction = \actk{\symbevt_\symbaction}{\symbcore} \wedge dom\_e(\symbstate_\symbconf, \symbcore, \symbevt_\symbaction) = \symbdomain_\symbaction) \vee (\symbactk_\symbaction = \actk{\symbpcomp}{\symbcore} \wedge \symbevt_\symbaction = \symbevtctx_\symbconf(\symbcore) \wedge dom\_e(\symbstate_\symbconf, \symbcore, \symbevt_\symbaction) = \symbdomain_\symbaction))\}$. 

\item $\symbconf_0$ is the initial configuration $(\sharp_0, s_0,x_0)$.
\end{itemize}
\end{definition}

Based on the function $step$, we define the function $run$ as follows to represent the execution of a sequence of actions. 
\[
\left\{
\begin{aligned}
& run(Nil) = Id \\
& run(\symbaction \# \symbactions) = step(\symbaction) \circ run(\symbactions)
\end{aligned}
\right.
\]

We prove the following lemma to ensure that the state machine is an equivalent representation of the {\slang} specification. 

\begin{lemma}[Equivalence of {\slang} and Its State Machine Representation]
\label{lm:sm_equiv}
The state machine defined in {\defprefix} \ref{def:statemachine} is an equivalent representation of {\slang}, i.e., 
\begin{itemize}
\item If $(\symbconf_1,\symbconf_2) \in run(\symbactions)$, then $\exists \symbcomp. \ \symbcomp \in \compfun(\progenv,\symbpes) \wedge \symbcomp_0 = \symbconf_1 \wedge last(\symbcomp) = \symbconf_2 \wedge (\forall j < len(\symbcomp) - 1. \ \transenv \symbcomp_j \trant{\symbactk_{\symbactions_j}}{pes} \symbcomp_{(j+1)})$, and 
\item If $\symbcomp \in \compfun(\progenv,\symbpes) \wedge \symbcomp_0 = \symbconf_1 \wedge last(\symbcomp) = \symbconf_2 \wedge (\forall j < len(\symbcomp) - 1. \ \neg (\transenv \symbcomp_j \tranenvpes \symbcomp_{(j+1)}))$, then $\exists \symbactions. \ (\symbconf_1,\symbconf_2) \in run(\symbactions) \wedge (\forall j < len(\symbcomp) - 1. \ \transenv \symbcomp_j \trant{\symbactk_{\symbactions_j}}{pes} \symbcomp_{(j+1)})$
\end{itemize}
where $\compfun(\progenv,\symbpes)$ is the set of computations of a parallel event system $\symbpes$.
\end{lemma}

We consider closed specifications where there is no environment transition in the computations of $\symbpes$, i.e., $\forall j < len(\symbcomp) - 1. \ \neg (\transenv \symbcomp_j \tranenvpes \symbcomp_{(j+1)})$. 

\subsubsection{Integrity}

Following the definitions of integrity in \cite{rushby92,Murray12}, we define an intuitive integrity property in {\slang} as follows, which concerns the small-step execution of the system.

\begin{definition}[Integrity]
\label{def:integrity}
Integrity of a parallel event system $\symbpes$ with a static configuration $\progenv$ from the initial state $\symbstate_0$ and $\symbevtctx_0$ is defined as 
\[
\forall \symbaction, \symbdomain, \symbconf, \symbconf'. \ \reachablef(\symbconf) \wedge  
\symbdomain_\symbaction \ninterf \symbdomain \wedge (\symbconf, \symbconf') \in step(\symbaction) \longrightarrow \symbconf \dsim{\symbdomain} \symbconf'
\]

where $\reachablef(\symbconf) \equiv \exists as.\ (\symbconf_0,\symbconf) \in run(as)$ is a function to show that configuration $\symbconf$ is reachable from the initial configuration of the system $\symbconf_0$.
\end{definition}

The intuition of integrity is that a small-step action $\symbaction$ executed in a configuration $\symbconf$ can only affect those
domains for which the domain executing $\symbaction$ is allowed to send information.

\subsection{Compositional Reasoning}
\label{subsect:compreason}

In order to decompose the integrity reasoning of the system to its events, we define a form of integrity on event as follows. 

\begin{definition}[Integrity on Events]
\label{def:integrity_evt}
The integrity on the events composing a parallel event system $\symbpes$ is defined as
\[
\forall \symbevt \ \symbdomain \ \symbstate \ \symbstate' \ \symbcore. \ \symbevt \in evts(\symbpes) \wedge (\symbstate, \symbstate') \in G_{\Gamma(\symbevt)} \wedge (dom\_e(\symbstate, \symbcore, \symbevt) \ninterf d) \longrightarrow \symbstate \dsim{\symbdomain} \symbstate'
\]
\end{definition}

where the $evts(\symbpes)$ function returns all the events defined in a specification $\symbpes$. We assume a function $\Gamma: evts(\symbpes) \rightarrow RGCond$, where $RGCond$ is the type of the rely-guarantee specification, to specify the rely-guarantee specification of events in $\symbpes$. $G_{\Gamma(\symbevt)}$ is the guarantee condition in the rely-guarantee specification of an event $\symbevt$. 

The integrity on events requires that when an event $\symbevt$ is executed, the interaction of $\symbevt$ with the environment affects only those domains for which the domain executing $\symbevt$ is allowed to send information to, according to the relation  $\interf$. Different from the integrity on actions of parallel in {\defprefix} \ref{def:integrity}, the integrity here considers the global effects of events to the environment.  

Next, we show the compositionality of integrity, i.e. the integrity on events implies the integrity on parallel event systems. 
First, {\lemmaprefix}~\ref{lm:evtctx_consist} shows the consistency of the event context in computations under a closed $\symbpes$, i.e.,  program transitions of $\symbpes$, $\transenv \symbcomp_i \trant{\actk{c}{\symbcore}}{pes} \symbcomp_{i+1}$, should be a transition of the event currently being under execution.  

\begin{lemma}
\label{lm:evtctx_consist}
For any closed $\symbpes$, if  $\forall \symbevt \in evts(\symbpes). \ is\_basic(\symbevt)$, that is, all events in $\symbpes$ are basic events,
then for any computation $\symbcomp$ of $\symbpes$, we have 
\[
\forall i < len(\symbcomp) - 1, \symbcore. \ (\transenv \symbcomp_i \trant{\actk{c}{\symbcore}}{pes} \symbcomp_{i+1}) \longrightarrow (\exists \symbevt \in evts(\symbpes). \ \symbevtctx_{\symbcomp_i}(\symbcore) = \symbevt)
\]
\end{lemma}
\begin{proof}
For the computation $\symbcomp$, we have its conjoined computations $\tilde{\symbcomp}$ such that $\symbcomp\ \isasympropto\ \tilde{\symbcomp}$. Hence, for a program transition $\transenv \symbcomp_i \trant{\actk{c}{\symbcore}}{pes} \symbcomp_{i+1}$, we have that $\transenv \tilde{\symbcomp}_i^{\symbcore} \tranes{c}{\symbcore} \tilde{\symbcomp}_{i+1}^{\symbcore}$. Since, all events in $\symbpes$ are basic events, all events in the event system $\symbpes(\symbcore)$ are basic events too. Thus, there must be an event occurrence transition $\transenv \tilde{\symbcomp}_m^{\symbcore} \tranes{\symbevt}{\symbcore} \transenv \tilde{\symbcomp}_{m+1}^{\symbcore}$ where $m < i$ and $\symbevt \in evts(\symbpes)$. The transition sets the event context of $\symbcore$ to $\symbevt$, and all transitions between $m$ and $i$ are program transitions or environment transitions which will not change the event context. 
%This is proven by induction on the length of a conjoined computation $\symbcomp^{\symbcore}$ of $\symbcore$ corresponding to a computation $\symbcomp$ of $\symbpes$. 
%If $i=0$ then $\symbcomp^{\symbcore}_i$ is an initial configuration such that  $\symbevtctx_{\symbcomp^{\symbcore}_0}(\symbcore) = {\symbevt}_0$ with ${\symbevt}_0\in evts(\symbpes)$. For $i > 0$ then the event system semantics may have transited because it triggered an event $\symbevt \in evts(\symbpes)$, and then  $\symbevtctx_{\symbcomp^{\symbcore}_{i+1}}(\symbcore) = \symbevt$; or it may have executed an already triggered event or performed an internal event system operation, in which case $\symbevtctx_{\symbcomp^{\symbcore}_{i+1}}(\symbcore) = \symbevtctx_{\symbcomp^{\symbcore}_{i}}(\symbcore)$.  Then, because $\symbcore$ is an event system of $\symbpes$, for all parallel event system computation $\symbcomp$ there exists an event system computations such that component transitions in the event system $\symbcore$ are transitions from $\symbcore$ in the event system $\symbcomp$, where the event context component of the configurations in the parallel event system is the same than in the configurations in the event system $\symbcore$, that is $\forall i \leq len(\symbcomp).\ \symbevtctx_{\symbcomp_i} = \symbevtctx_{\symbcomp^{\symbcore}_i}$. 
\end{proof}

Second, {\lemmaprefix} \ref{lm:guar_consist} shows the compositionality of guarantee conditions of events in a valid and closed parallel event system, i.e., any component transition must preserve the guarantee condition of the current event.

\begin{lemma}
\label{lm:guar_consist}
For any closed $\symbpes$, if 
\begin{enumerate}
\item $\symbconf_0 = (\symbpes, \symbstate_0, \symbevtctx_0)$.
\item events in $\symbpes$ are basic events, i.e., $\forall \symbevt \in evts(\symbpes). \ is\_basic(\symbevt)$.
\item events in $\symbpes$ satisfy their rely-guarantee specification, i.e., $\forall \symbevt \in evts(\symbpes). \ \rgsat{\symbevt}{\Gamma(\symbevt)}$. 
\item from (3) we have $\rgsat{\symbpes}{\rgcond{\{\symbstate_0\}}{\{\}}{UNIV}{UNIV}}$.
\end{enumerate}
then for any computation $\symbcomp \in \compfun(\symbpes, \symbstate_0, \symbevtctx_0)$, we have 
\[
\forall i < len(\symbcomp) - 1, \symbcore. \ (\transenv \symbcomp_i \trant{\actk{c}{\symbcore}}{pes} \symbcomp_{i+1}) \longrightarrow (\symbstate_{\symbcomp_i}, \symbstate_{\symbcomp_{i+1}}) \in G_{\Gamma(\symbevtctx_{\symbcomp_i}(\symbcore))}
\]
\end{lemma}
\begin{proof}
From assumption (4), we know that the rely-guarantee specifications of each event in assumption (3) are compositional, i.e. the execution of each event in all computations of $\compfun(\symbpes, \symbstate_0, \symbevtctx_0)$ preserves its specification. 
For the computation $\symbcomp$, we have its conjoined computations $\tilde{\symbcomp}$ such that $\symbcomp\ \isasympropto\ \tilde{\symbcomp}$. Hence, for a program transition $\transenv \symbcomp_i \trant{\actk{c}{\symbcore}}{pes} \symbcomp_{i+1}$, we have that $\transenv \tilde{\symbcomp}_i^{\symbcore} \tranes{c}{\symbcore} \tilde{\symbcomp}_{i+1}^{\symbcore}$. Next, we apply induction on $\symbspec_{\tilde{\symbcomp}_{0}^{\symbcore}}$:
\begin{enumerate}

\item $\symbspec_{\tilde{\symbcomp}_{0}^{\symbcore}} = \evtsysdef$: the execution of an event system is a sequence of executions of its composed events. Thus, the program transition is a transition of its one event. We split the computation $\tilde{\symbcomp}^{\symbcore}$ to a set of computations of events, and the program transition is in a computation of an event $\symbevt \in \evtsysdef$. By {\lemmaprefix} \ref{lm:evtctx_consist} and $\rgsat{\symbevt}{\Gamma(\symbevt)}$, we have the conclusion. 

\item $\symbspec_{\tilde{\symbcomp}_{0}^{\symbcore}} = \evtseq{\symbevt}{\symbevtsys}$: first, if the transition is of the execution of $\symbevt$, we have $\symbevtctx_{\tilde{\symbcomp}_i^{\symbcore}}(\symbcore) = \symbevt$ by {\lemmaprefix} \ref{lm:evtctx_consist} and the semantics of event occurrence. Moreover, the execution of $\evtseq{\symbevt}{\symbevtsys}$ before position $i$ is the same as of $\symbevt$ and $\rgsat{\symbevt}{\Gamma(\symbevt)}$. We have $(\symbstate_{\symbcomp_i}, \symbstate_{\symbcomp_{i+1}}) \in G_{\Gamma(\symbevtctx_{\symbcomp_i}(\symbcore))}$. Second, if the transition is of the execution of $\symbevtsys$, we have the conclusion by the inductive case (1). 

\end{enumerate}
%Because for any event $\symbevt$, $\symbevt \in evts(\symbpes)$, $\rgsat{\symbevt}{\Gamma(\symbevt)}$, then for any event system $\symbcore$, the state of component transitions of computations $\symbcomp^{\symbcore}$ in ${\symbcore}$ belong to the component guarantee, that is $(\symbstate_{\symbcomp_{i}^{\symbcore}}, \symbstate_{\symbcomp_{i+1}^{\symbcore}}) \in G_{\Gamma(\symbevtctx_{\symbcomp_i}(\symbcore))}$, and because of Lemma~\ref{lm:evtctx_consist} $\symbevtctx_{\symbcomp_i}(\symbcore) = \symbevt$. Then, because $\symbcore$ is an event system of $\symbpes$, for all parallel event system computation $\symbcomp$ there exists an event system computations such that component transitions in the event system $\symbcore$ are transitions from $\symbcore$ in the event system $\symbcomp$, where the state component of the configurations in the parallel event system is the same than in the configurations in the event system ${\symbcore}$, that is $\forall i \leq len(\symbcomp).\  \symbstate_{\symbcomp_i} = \symbstate_{\symbcomp^{\symbcore}_i}$. 
\end{proof}

By these two lemmas and the equivalence in {\lemmaprefix} \ref{lm:sm_equiv}, we have the following theorem for the compositionality of integrity. 

\begin{theorem}[Compositionality of Integrity]
\label{thm:comp_integrity}
For a closed parallel event system $\symbpes$, if 
\begin{enumerate}
\item $\symbconf_0 = (\symbpes, \symbstate_0, \symbevtctx_0)$.
\item events in $\symbpes$ are basic events, i.e., $\forall \symbevt \in evts(\symbpes). \ is\_basic(\symbevt)$.
\item events in $\symbpes$ satisfy their rely-guarantee specification, i.e., $\forall \symbevt \in evts(\symbpes). \ \rgsat{\symbevt}{\Gamma(\symbevt)}$. 
\item $\rgsat{\symbpes}{\rgcond{\{\symbstate_0\}}{\{\}}{UNIV}{UNIV}}$.
\item $\symbpes$ satisfies the integrity on its events. 
\end{enumerate}
then $\symbpes$ preserves the integrity property. 
\end{theorem}

We require that all events in $\symbpes$ are basic events to ensure the event context in computations of $\symbpes$ is consistent before the execution of an event. It is a reasonable assumption since anonymous events are only used to represent intermediate specifications during the execution of events and they do not change modify the event context. 
The assumption (4) is to ensure the compositionality of the rely-guarantee specifications of each event in assumption (3), i.e. the execution of each event in all computations of $\compfun(\symbpes, \symbstate_0, \symbevtctx_0)$ preserves its specification. It is a highly relaxed condition and it is easy to prove. First, we only consider closed concurrent systems starting from the initial state $\symbstate_0$. Thus, the pre-condition only has the initial state and the rely condition is empty. Second, we do not constrain the behaviour of the parallel event system, thus the guarantee condition is the universal set. %Second, we concern the environment affect of an event to other events, but not the overall modification, and thus the guarantee condition is the universal set. 
Third, the integrity only concerns the action transition, but not the final state. Thus, the post-condition is the universal set.  

\section{Rely-guarantee Proof of Zephyr}
\label{sect:proof}

We have proven correctness of the buddy memory management in Zephyr using the rely-guarantee proof system of {\slang}. 
We ensure functional correctness of each kernel service w.r.t. the defined pre/post conditions, termination of loop statements in the kernel services, the separation of local variables of threads, safety by invariant preservation, and security by memory separation. 
In this section, we introduce how these properties are specified and verified using the {\slang} rely-guarantee proof system.

Actually, the safety and security properties verified in this article can be embedded in the guarantee conditions of events of the memory management specification. In this section, we first present the rely-guarantee specification of events. 

\subsection{Correctness of the Specification}
\label{subsect:corspec}

%\qmark{why these cor specification? from invariants, integrity properties. }

Using the compositional reasoning of {\slang}, correctness of Zephyr memory management can be specified and verified with the rely-guarantee specification of each event. 

%The partial correctness of a kernel service is specified by its pre/post-conditions. 

%The termination of each service will be discussed in next subsection. Other correctness conditions are represented in the guarantee condition. 

%Invariant preservation, memory configuration, and separation of local variables is specified in the guarantee condition of each service. 

The guarantee conditions for both memory services are the same, which is defined as:

\begin{isabellec}
\isacommand{Mem{\isacharunderscore}pool{\isacharunderscore}guar}\ t\ {\isasymequiv} 
%$\stackrel{(1)}{Id}$
$\overbrace{Id}^{(1)}$
 \ {\isasymunion}\ {\isacharparenleft}
$\overbrace{gvars\_conf\_stable}^{(2)}$
 {\isasyminter}
 
\quad {\isacharbraceleft}{\isacharparenleft}s{\isacharcomma}r{\isacharparenright}{\isachardot}\ {\isacharparenleft}
$\overbrace{cur\ s\ {\isasymnoteq}\ Some\ t\ {\isasymlongrightarrow}\ gvars{\isacharunderscore}nochange\ s\ r\ {\isasymand}\ lvars{\isacharunderscore}nochange\ t\ s\ r}^{(3.1)}$
{\isacharparenright}

\quad \quad \quad {\isasymand}\ {\isacharparenleft}
$\overbrace{cur\ s\ {\isacharequal}\ Some\ t\ {\isasymlongrightarrow}\ inv\ s\ {\isasymlongrightarrow}\ inv\ r}^{(3.2)}$
{\isacharparenright}\ {\isasymand}\ {\isacharparenleft}
$\overbrace{{\isasymforall}t{\isacharprime}{\isachardot}\ t{\isacharprime}\ {\isasymnoteq}\ t\ {\isasymlongrightarrow}\ lvars{\isacharunderscore}nochange\ t{\isacharprime}\ s\ r}^{(4)}$
{\isacharparenright}

\quad \quad \quad {\isasymand}\ {\isacharparenleft}
$\overbrace{{\isasymforall}t{\isacharprime}{\isachardot}\ t{\isacharprime}\ {\isasymnoteq}\ t\ {\isasymlongrightarrow}\ mblocks\ s\ t{\isacharprime} = mblocks\ r\ t{\isacharprime}}^{(5)}$
{\isacharparenright}{\isacharbraceright}
{\isasyminter} $\overbrace{{\isasymlbrace}{\isasymordmasculine}tick\ {\isacharequal}\ {\isasymordfeminine}tick{\isasymrbrace}
{\isacharparenright}}^{(6)}$
\isanewline
\end{isabellec}

This relation states that \emph{alloc} and \emph{free} services may not change the state (1), e.g., a blocked await or selecting branch on a conditional statement. If it changes the state then: (2) the static configuration of memory pools in the model do not change; (3.1) if the scheduled thread is not the thread invoking the event then variables for that thread do not change (since it is blocked in an \emph{Await} as explained in {\sectprefix} \ref{sect:mem_spec}); (3.2) if it is, then the relation preserves the memory invariant, and consequently each step of the event needs to preserve the invariant, where $inv$ is the conjunction of all invariant properties defined in {\sectprefix} \ref{subsect:inv}; (4) a thread does not change the local variables of other threads; (5) it is the memory separation property at the thread level, which means that a thread does not change the allocated memory blocks of other threads; and (6) threads do not change the \emph{tick} of the timer. Here, ${\isasymlbrace}{\isasymordmasculine}tick\ {\isacharequal}\ {\isasymordfeminine}tick{\isasymrbrace}$ defines a set of state pairs, where operator ${\isasymordmasculine}tick$ represents the $tick$ field in the first state of a pair and ${\isasymordfeminine}tick$ in the second state. It is equivalent to $\{(s,t).\ tick\ s = tick\ t\}$. 

We could find that the preservation of the memory configuration during small steps of kernel services is represented as condition (2), the invariant preservation as condition (3.2), the separation of local variables of threads as condition (4), and the memory separation as condition (5). 

The rely conditions for the both memory services are the same, which is defined as:

\begin{isabellec} \isanewline
\isacommand{Mem{\isacharunderscore}pool{\isacharunderscore}rely}\ t\ {\isasymequiv}\ Id\ {\isasymunion}\ 
{\isacharparenleft}gvars{\isacharunderscore}conf{\isacharunderscore}stable {\isasyminter}\ {\isacharbraceleft}{\isacharparenleft}s{\isacharcomma}r{\isacharparenright}{\isachardot}\ (inv\ s\ {\isasymlongrightarrow}\ inv\ r) {\isasymand} lvars{\isacharunderscore}nochange\ t\ s\ r 

\quad {\isasymand} {\isacharparenleft}cur\ s\ {\isacharequal}\ Some\ t\ {\isasymlongrightarrow}\ mem{\isacharunderscore}pool{\isacharunderscore}info\ s\ {\isacharequal}\ mem{\isacharunderscore}pool{\isacharunderscore}info\ r
{\isasymand}\ {\isacharparenleft}{\isasymforall}t{\isacharprime}{\isachardot}\ t{\isacharprime}\ {\isasymnoteq}\ t\ {\isasymlongrightarrow}\ lvars{\isacharunderscore}nochange\ t{\isacharprime}\ s\ r{\isacharparenright}{\isacharparenright}

\quad {\isasymand} mblocks\ s\ t = mblocks\ r\ t
{\isacharbraceright}{\isacharparenright} \isanewline
\end{isabellec}

This relation states that the \emph{alloc} and \emph{free} services assume that the environment does not change the state (\emph{Id}), otherwise (1) the environment does not change the static configuration of memory pools; (2) the environment preserves the invariant; (3) the environment does not change the local variables of thread $t$; (4) if the scheduled thread is the thread invoking the event then memory is not changed and local variables of other threads are not changed; and (5) the environment does not change the allocated memory blocks of thread $t$. 

The functional correctness is specified as the pre/post-conditions. For instance, we prove that when starting in a valid memory configuration given by the invariant, then if the allocation service does not returns an error code then it returns a valid memory block with size bigger or equal than the requested capacity. The property is specified by the following postcondition of the allocation service: 

\begin{isabellec}
	\isanewline
\isacommand{Mem{\isacharunderscore}pool{\isacharunderscore}alloc{\isacharunderscore}pre}\ t\ {\isasymequiv}\ {\isacharbraceleft}s{\isachardot}\ inv\ s\ {\isasymand}\ allocating{\isacharunderscore}node\ s\ t\ {\isacharequal}\ None\ {\isasymand}\ freeing{\isacharunderscore}node\ s\ t\ {\isacharequal}\ None{\isacharbraceright}

\isacommand{Mem{\isacharunderscore}pool{\isacharunderscore}alloc{\isacharunderscore}post}\ t\ p\ sz\ timeout\ {\isasymequiv}\ \isanewline
\ \ {\isacharbraceleft}s{\isachardot}\ inv\ s\ {\isasymand}\ allocating{\isacharunderscore}node\ s\ t\ {\isacharequal}\ None\ {\isasymand}\ freeing{\isacharunderscore}node\ s\ t\ {\isacharequal}\ None

\ \ \ \ \ \ {\isasymand}\ {\isacharparenleft}timeout\ {\isacharequal}\ FOREVER\ {\isasymlongrightarrow}

\quad \quad \quad {\isacharparenleft}ret\ s\ t\ {\isacharequal}\ ESIZEERR\ {\isasymand}\ mempoolalloc{\isacharunderscore}ret\ s\ t\ {\isacharequal}\ None\ {\isasymor}

\quad \quad \quad \ ret\ s\ t\ {\isacharequal}\ OK\ {\isasymand}\ {\isacharparenleft}{\isasymexists}mblk{\isachardot}\ mempoolalloc{\isacharunderscore}ret\ s\ t\ {\isacharequal}\ Some\ mblk\ {\isasymand}\ mblk{\isacharunderscore}valid\ s\ p\ sz\ mblk{\isacharparenright}{\isacharparenright}{\isacharparenright}

\ \ \ \ \ \ {\isasymand}\ {\isacharparenleft}timeout\ {\isacharequal}\ NOWAIT\ {\isasymlongrightarrow}

\quad \quad \quad {\isacharparenleft}{\isacharparenleft}ret\ s\ t\ {\isacharequal}\ ENOMEM\ {\isasymor}\ ret\ s\ t\ {\isacharequal}\ ESIZEERR{\isacharparenright}\ {\isasymand}\ mempoolalloc{\isacharunderscore}ret\ s\ t\ {\isacharequal}\ None{\isacharparenright} {\isasymor} 

\quad \quad \quad \ {\isacharparenleft}ret\ s\ t\ {\isacharequal}\ OK\ {\isasymand}\ {\isacharparenleft}{\isasymexists}mblk{\isachardot}\ mempoolalloc{\isacharunderscore}ret\ s\ t\ {\isacharequal}\ Some\ mblk\ {\isasymand}\ mblk{\isacharunderscore}valid\ s\ p\ sz\ mblk{\isacharparenright}{\isacharparenright}{\isacharparenright}

\ \ \ \ \ \ {\isasymand}\ {\isacharparenleft}timeout\ {\isachargreater}\ {\isadigit{0}}\ {\isasymlongrightarrow}

\quad \quad \quad {\isacharparenleft}{\isacharparenleft}ret\ s\ t\ {\isacharequal}\ ETIMEOUT\ {\isasymor}\ ret\ s\ t\ {\isacharequal}\ ESIZEERR{\isacharparenright}\ {\isasymand}\ mempoolalloc{\isacharunderscore}ret\ s\ t\ {\isacharequal}\ None{\isacharparenright} {\isasymor} 

\quad \quad \quad \ {\isacharparenleft}ret\ s\ t\ {\isacharequal}\ OK\ {\isasymand}\ {\isacharparenleft}{\isasymexists}mblk{\isachardot}\ mempoolalloc{\isacharunderscore}ret\ s\ t\ {\isacharequal}\ Some\ mblk\
{\isasymand}\ mblk{\isacharunderscore}valid\ s\ p\ sz\ mblk{\isacharparenright}{\isacharparenright}{\isacharparenright}{\isacharbraceright}
\isanewline
\end{isabellec}

If a thread requests a memory block in mode \emph{FOREVER}, it may  successfully allocate a valid memory block, or fail (\emph{ESIZEERR}) if the request size is larger than the size of the memory pool. If the thread is requesting a memory pool in mode \emph{NOWAIT}, it may also get  \emph{ENOMEM} as a result if there is no available blocks. But if the thread is requesting in mode \emph{TIMEOUT}, it will get the result of \emph{ETIMEOUT} if there is no available blocks in \emph{timeout} milliseconds. 

The property is indeed weak since even if the memory has a block able to allocate the requested size before invoking the allocation service, another thread running concurrently may have taken the block first during the execution of the service. 
For the same reason, the released block may be taken by another concurrent thread before the end of the release services.

\subsection{Proof of Partial Correctness}
\label{subsect:corproof}
In the {\slang} system, verification of a rely-guarantee specification is carried out by inductively applying the proof rules for each system event and discharging the proof obligations the rules generate. Typically, these proof obligations require to prove stability of the pre- and post-condition to check that changes of the environment preserve them, and to show that a statement modifying a state from the precondition gets a state belonging to the postcondition. %A detailed proof sketch of the \emph{free} service is shown in Appendix \ref{appx:mempoolfree}.

The final theorem of functional correctness is as follows.

\begin{theorem}[Functional Correctness of Memory Management]
\label{thm:cor}
\[
\rgsat{Mem\_Spec}{\rgcond{\{s_0\}}{\{\}}{Guar}{UNIV}}
\]
where $Guar = tick\_guar\ \cup\ schedule\_guar\ \cup\ (\bigcup {t}.\ Mem\_pool\_guar\ t)$. 
\end{theorem}

%We omit the guarantee condition of the release service in $Guar$, since its the same as that of the allocation service. 
We consider that the memory management is a closed system, i.e., the environment is the empty set. In the initial state $s_0$, we assume that (1) the memory blocks at level 0 of all memory pools are free and not split, (2) the current thread is \textbf{None}, (3) the state of all threads are \emph{READY}, and (4) the wait queue of thread of each memory pool is empty. We have that $s_0$ satisfies the invariants. 

By the \textsc{Par} rule in {\figprefix} \ref{fig:proofrule}, the proof of {\theoremprefix} \ref{thm:cor} can be decomposed to the satisfiability for each event of the proof of correctness of the specification introduced in section~\ref{subsect:corspec}. We proved the following lemmas for the memory services. A detailed proof sketch and intermediate conditions of {\lemmaprefix} \ref{lm:cor_free} is shown in Appendix \ref{appx:mempoolfree}. 

\begin{lemma}[Functional Correctness of the Allocation Service]
\label{lm:cor_alloc}
\begin{equation*}
\begin{aligned}
\progenv \vdash Mem\_pool\_alloc\ t\ p\ sz\ to\ \mathbf{sat} \ \langle & Mem\_pool\_alloc\_pre\ t, Mem\_pool\_rely\ t, \\
& Mem\_pool\_guar\ t, Mem\_pool\_alloc\_post\ t\ p\ sz\ to \rangle
\end{aligned}
\end{equation*} 
\end{lemma}

\begin{lemma}[Functional Correctness of the Release Service]
\label{lm:cor_free}
\begin{equation*}
\begin{aligned}
\progenv \vdash Mem\_pool\_free\ t\ p\ sz\ to\ \mathbf{sat} \ \langle & Mem\_pool\_free\_pre\ t, Mem\_pool\_rely\ t, \\
& Mem\_pool\_guar\ t, Mem\_pool\_free\_post\ t\ p\ sz\ to \rangle
\end{aligned}
\end{equation*}
\end{lemma}

\subsection{Proof of Termination}
\label{subsect:termi}

To prove loop termination, loop invariants are parametrized with a logical variable $\alpha$. It suffices to show total correctness of a loop statement by the following proposition where $loopinv(\alpha)$ is the parametrize invariant, in which the logical variable is used to find a convergent relation to show that the number of iterations of the loop is finite. 
\[ \fontsize{9pt}{0cm}
\begin{aligned} 
&\rgsat{P}{\rgcond{loopinv(\alpha) \cap \stset{\alpha > 0}}{R}{G}{\exists \beta < \alpha. \ loopinv(\beta)}} \\
& \wedge loopinv(\alpha) \cap \stset{\alpha > 0} \subseteq \stset{b} \wedge loopinv(0) \subseteq \stset{\neg b} \\
& \wedge \forall s \in loopinv(\alpha).\ (s,t) \in R \longrightarrow \exists \beta \leqslant \alpha. \ t \in loopinv(\beta)
\end{aligned}
\]

For instance, to prove termination of the loop statement in \emph{free\_block} shown in {\figprefix} \ref{fig:freeblock}, we define the loop invariant with the logical variable $\alpha$ as follows. Here, ${\isasymlbrace} {\isasymacute}inv {\isasymrbrace}$ defines a set of states, each of which satisfies the $inv$ predicate. It is equivalent to $\{s.\ inv\ s\}$. 

\begin{isabellec}
	\isanewline
\isacommand{mp{\isacharunderscore}free{\isacharunderscore}loopinv}\ t\ b\ {\isasymalpha}\ {\isasymequiv} {\isasymlbrace} ... {\isasymand}{\isasymacute}inv {\isasymand} level\ b\ {\isacharless}\ length\ {\isacharparenleft}{\isasymacute}lsizes\ t{\isacharparenright}

\quad {\isasymand}\ {\isacharparenleft}{\isasymforall}ii{\isacharless}length\ {\isacharparenleft}{\isasymacute}lsizes\ t{\isacharparenright}{\isachardot}\ {\isasymacute}lsizes\ t\ {\isacharbang}\ ii\ {\isacharequal}\ {\isacharparenleft}max{\isacharunderscore}sz\ {\isacharparenleft}{\isasymacute}mem{\isacharunderscore}pool{\isacharunderscore}info\ {\isacharparenleft}pool\ b{\isacharparenright}{\isacharparenright}{\isacharparenright}\ div\ {\isacharparenleft}{\isadigit{4}}\ {\isacharcircum}\ ii{\isacharparenright}{\isacharparenright}

\quad {\isasymand}\ {\isasymacute}bn\ t\ {\isacharless}\ length\ {\isacharparenleft}bits\ {\isacharparenleft}levels\ {\isacharparenleft}{\isasymacute}mem{\isacharunderscore}pool{\isacharunderscore}info\ {\isacharparenleft}pool\ b{\isacharparenright}{\isacharparenright}{\isacharbang}{\isacharparenleft}{\isasymacute}lvl\ t{\isacharparenright}{\isacharparenright}{\isacharparenright}

\quad {\isasymand}\ {\isasymacute}bn\ t\ {\isacharequal}\ {\isacharparenleft}block\ b{\isacharparenright}\ div\ {\isacharparenleft}{\isadigit{4}}\ {\isacharcircum}\ {\isacharparenleft}level\ b\ {\isacharminus}\ {\isasymacute}lvl\ t{\isacharparenright}{\isacharparenright} {\isasymand} {\isasymacute}lvl\ t\ {\isasymle}\ level\ b

\quad {\isasymand}\ {\isacharparenleft}{\isasymacute}free{\isacharunderscore}block{\isacharunderscore}r\ t\ {\isasymlongrightarrow}\  {\isacharparenleft}{\isasymexists}blk{\isachardot}\ {\isasymacute}freeing{\isacharunderscore}node\ t\ {\isacharequal}\ Some\ blk\ {\isasymand}\ pool\ blk\ {\isacharequal}\ pool\ b

\quad \quad \quad \quad \quad \quad \quad \quad \quad \quad \quad \quad \quad \quad \quad 
{\isasymand}\ level\ blk\ {\isacharequal}\ {\isasymacute}lvl\ t\ {\isasymand}\ block\ blk\ {\isacharequal}\ {\isasymacute}bn\ t{\isacharparenright}

\quad \quad \quad \quad \quad \quad \quad \quad \quad \quad {\isasymand}\ {\isasymacute}alloc{\isacharunderscore}memblk{\isacharunderscore}data{\isacharunderscore}valid\ {\isacharparenleft}pool\ b{\isacharparenright}\ {\isacharparenleft}the\ {\isacharparenleft}{\isasymacute}freeing{\isacharunderscore}node\ t{\isacharparenright}{\isacharparenright}{\isacharparenright}

\quad {\isasymand}\ {\isacharparenleft}{\isasymnot}\ {\isasymacute}free{\isacharunderscore}block{\isacharunderscore}r\ t\ {\isasymlongrightarrow}\ {\isasymacute}freeing{\isacharunderscore}node\ t\ {\isacharequal}\ None{\isacharparenright} \ {\isasymrbrace}\ {\isasyminter} 

\quad {\isasymlbrace}\ {\isasymalpha}\ {\isacharequal}\ {\isacharparenleft}if\ {\isasymacute}freeing{\isacharunderscore}node\ t\ {\isasymnoteq}\ None\ then\ {\isasymacute}lvl\ t\ {\isacharplus}\ {\isadigit{1}}\ else\ {\isadigit{0}}{\isacharparenright}\ {\isasymrbrace}
\isanewline 
\end{isabellec}

where $freeing\_node$ and $lvt$ are local variables respectively storing the node being free and the level that the node belongs to. 
In the body of the loop, if $lvl\ t\ {\isachargreater}\ {\isadigit{0}}$ and $partner\_bit$ is \emph{true}, then $lvl = lvl - 1$ at the end of the body. Otherwise, $freeing\_node \ t= None$. So at the end of the loop body, $\alpha$ decreases or $\alpha = 0$. If $\alpha = 0$, we have $freeing\_node\ t = None$, and thus the negation of the loop condition  $\neg free\_block\_r \ t$, concluding termination of \emph{free\_block}.

Due to concurrency, it is necessary to consider fairness to prove termination of the loop statement in \emph{k\_mempool\_alloc} from Line 34 to 47 in {\figprefix} \ref{fig:mem_alloc_code}. On one hand, when a thread requests a memory block in the \emph{FOREVER} mode, it is possible that there will never be available blocks since other threads do not release allocated blocks. On the other hand, even when other threads release blocks, it is possible that the available blocks are always raced by threads. 

\subsection{Proof of Safety}
\label{subsect:safetyproof}

We have proved the following theorem to show that, as a closed system, the memory management of Zephyr is safe w.r.t invariant $inv$ executing from an initial state $s_0$. According to {\theoremprefix} \ref{thm:invariant}, the memory management also preserves the safety in any environment $R$, if $R$ is stable with the invariant. 

\begin{theorem}[Safety of Memory Management]
\label{thm:safety_mem}
The {\slang} specification \textbf{\emph{Mem\_Spec}} in {\equationprefix} (\ref{eq:mem_spec}) satisfies the invariant \emph{\textbf{inv}} w.r.t. the initial state $s_0$ in an empty environment $\{\}$. 
\end{theorem}
\begin{proof}
To prove the theorem, by {\theoremprefix} \ref{thm:invariant} and {\theoremprefix} \ref{thm:cor}, we only need to show that $inv(s_0)$, $stable(\{s.\ inv(s)\}, \{\})$, and $stable(\{s.\ inv(s)\}, Guar)$ where $Guar$ is defined in {\theoremprefix} \ref{thm:cor}. 
The first two predicates are obviously satisfied as we have shown before. The third predicate is satisfied since $Mem\_pool\_guar\ t$ is stable with the invariant and the guarantee conditions of $tick$ and $schedule$ do not change the memory. 
\end{proof}

\subsection{Proof of Security}
\label{subsect:securityproof}

%\subsection{Instantiation for Zephyr Memory Management}
To apply the compositional reasoning approach to proof integrity of a parallel event system, we instantiate the security configuration in {\sectprefix} \ref{subsect:ifs} for the specification of the memory management, i.e. $Mem\_Spec$ in {\equationprefix} (\ref{eq:mem_spec}). 

The $\interf$ relation of domains is instantiated as the \emph{interference} function as follows. The first two rules mean that the Timer can only interfere itself. The third rule means that a thread can interfere with itself but not other threads. The \emph{scheduler} can interfere with all threads and the \emph{timer}. A special case in the memory management of Zephyr is that threads can interfere with the \emph{scheduler} too, because the memory allocation service may block the current thread and reschedule to other threads. It is very different from the interference relation in separation kernels (e.g. \cite{Murray13}) and ARINC 653 OSs (e.g. \cite{zhao19tdsc}), where the scheduler must not be interfered by partitions or processes for the purpose of temporal separation. 
\[
\left\{
\begin{aligned}
& \textbf{Timer} \interf c = (c = Timer) \\
& c \interf \textbf{Timer}  = (c = Timer) \\
& (\textbf{\isasymT}\ t) \interf (\textbf{\isasymT}\ r)  = (t = r) \\
& otherwise . . . = True
\end{aligned}
\right.
\]

%%\[
%%\left\{
%%\begin{aligned}
%%& interfere\ (\textbf{Some}\ \textbf{Timer})\ (\textbf{Some}\ c) = (c = Timer) \\
%%& interfere\ (\textbf{Some}\ c)\ (\textbf{Some}\ \textbf{Timer})  = (c = Timer) \\
%%& interfere\ (\textbf{Some}\ (\textbf{\isasymT}\ t))\ (\textbf{Some}\ (\textbf{\isasymT}\ r))  = (t = r) \\
%%& otherwise . . . = True
%%\end{aligned}
%%\right.
%%\]

The state equivalence relation $s \dsim{\symbdomain} r$ is instantiated as follows. It requires that two states are equivalent to a thread \emph{t} iff the allocated memory blocks of \emph{t} in the two states are the same. 
\[
\left\{
\begin{aligned}
& s \dsim{\textbf{\isasymS}} r = (cur\ s = cur\ r) \\
& s \dsim{(\textbf{\isasymT}\ t)} r = (mblocks\ s\ t = mblocks\ r\ t) \\
& s \dsim{\textbf{Timer}} r = (tick\ s = tick\ r)
\end{aligned}
\right.
\]
%%\[
%%\left\{
%%\begin{aligned}
%%& s \dsim{\textbf{Some}\ \textbf{\isasymS}} r = (cur\ s = cur\ r) \\
%%& s \dsim{\textbf{Some}\ (\textbf{\isasymT}\ t)} r = (mblocks\ s\ t = mblocks\ r\ t) \\
%%& s \dsim{\textbf{Some}\ \textbf{Timer}} r = (tick\ s = tick\ r) \\
%%& s \dsim{\textbf{None}} r = True
%%\end{aligned}
%%\right.
%%\]

The domain function \emph{dom\_e} is instantiated as follows, which is straightforward. 
\[
\left\{
\begin{aligned}
& dom\_e\ s\ (\textbf{\isasymT}\ t)\ (alloc(p,sz,to)@t) = \textbf{\isasymT}\ t \\
& dom\_e\ s\ (\textbf{\isasymT}\ t)\ (free(b)@t) = \textbf{\isasymT}\ t \\
& dom\_e\ s\ \textbf{\isasymS}\ (schedule(t)) = \textbf{\isasymS} \\
& dom\_e\ s\ \textbf{Timer}\ tick = \textbf{Timer}
\end{aligned}
\right.
\]
%%\[
%%\left\{
%%\begin{aligned}
%%& domain\ s\ (\textbf{\isasymT}\ t)\ (alloc(p,sz,to)@t) = \textbf{Some}\ (\textbf{\isasymT}\ t) \\
%%& domain\ s\ (\textbf{\isasymT}\ t)\ (free(b)@t) = \textbf{Some}\ (\textbf{\isasymT}\ t) \\
%%& domain\ s\ \textbf{\isasymS}\ (schedule(t)) = \textbf{Some}\ \textbf{\isasymS} \\
%%& domain\ s\ \textbf{Timer}\ tick = \textbf{Some}\ \textbf{Timer} \\
%%& otherwise . . . = \textbf{None}
%%\end{aligned}
%%\right.
%%\]

%memory separation at thread level implies the integrity

By the compositional reasoning approach in {\sectprefix} \ref{subsect:compreason}, to prove the integrity of the memory management of Zephyr, we first have to show the integrity of memory services shown as the following lemma. 

\begin{lemma}[Integrity of Memory Services]
\label{lm:integrity_evt_mem}
\[
\forall ev\ u\ s_1\ s_2\ k.\ ev \in evts(Mem\_Spec) \wedge (s_1,s_2) \in G_{\Gamma(ev)} \wedge (dom\_e\ s_1\ k\ ev) \ninterf u \longrightarrow s_1\ \dsim{u}\ s_2
\]
\end{lemma}
\begin{proof}
By $mem\_Spec$ in {\equationprefix} \ref{eq:mem_spec}, the proof of this lemma is discharged by cases of $ev$ and then induction of the domain $u$. 
\begin{itemize}
\item if $ev$ is an event $Mem\_pool\_alloc$ or $Mem\_pool\_free$, then $dom\_e\ s_1\ k\ ev$ is a thread $t$. 
	\begin{itemize}
		\item if $u$ is the scheduler, then its proved since $t$ can interfere with the scheduler. 
		\item if $u$ is the timer, then its proved since the execution of memory services does not change the $tick$ variable. 
		\item if $u$ is a thread $r$, then we have $t \neq r$. It is proved by that the memory separation property in the guarantee condition {\isacommand{Mem{\isacharunderscore}pool{\isacharunderscore}guar}} does not change the allocated memory of other threads. 
	\end{itemize}
\item if $ev$ is the $schedule$ event, then $dom\_e\ s_1\ k\ ev$ is the scheduler {\isasymS}.
	\begin{itemize}
		\item if $u$ is the scheduler, then its proved since a domain can interfere with itself. 
		\item if $u$ is the timer, then its proved since the $schedule$ event does not change the $tick$ variable. 
		\item if $u$ is a thread $t$, then its proved since the $schedule$ event does not change the memory.  
	\end{itemize}	 
\item if $ev$ is the $tick$ event, then its proved since the $tick$ event does not change the current thread and the memory. 
\end{itemize}
\end{proof}

Finally, we have the following theorem to show the integrity of the memory management. 
\begin{theorem}[Integrity of Memory Management]
\label{thm:integrity_mem}
The {\slang} specification \textbf{\emph{Mem\_Spec}} in {\equationprefix} (\ref{eq:mem_spec}) satisfies the integrity property. 
\end{theorem}
\begin{proof}
The proof is straightforward by {\theoremprefix} \ref{thm:comp_integrity}, {\theoremprefix} \ref{thm:cor} and {\lemmaprefix} \ref{lm:integrity_evt_mem}. The rest to prove is $\forall ev \in evts(Mem\_Spec).\ is\_basic(ev)$, which is straightforward since the events are all basic events according to the definition of $Mem\_Spec$. 
\end{proof}

\section{Result and Evaluation} % 2 pages
\label{sect:result}

\subsection{Verification Effort}

The verification conducted in this work is on Zephyr v1.8.0. The C code of the buddy memory management is $\approx$ 400 lines, not counting blank lines and comments. {\tableprefix} \ref{tbl:stat} shows the statistics for the effort and size of the proofs in the Isabelle/HOL theorem prover. In total, the models and mechanized verification consists of $\approx$ 34,000 lines of specification and proofs (LOSP), and the total effort is $\approx$ 26 person-months (PM), where the security proof in {\slang} takes 4 PMs. The specification and proof of {\slang} are reusable for the verification of other systems. We develop $\approx$ 18,200 LOSP for the concurrent memory management of Zephyr, 40 times more than the lines of the C code due to the in-kernel concurrency, where invariant proofs represent the largest part. This takes 14 PMs. Since  the safety and security properties are represented by the guarantee conditions of the memory management services, the final theorems to show the safety and security are relatively small, taking 400 LOSP.

\begin{table}[t]
%\vspace{-8pt}
\centering
\footnotesize
\caption{Specification and Proof Statistics} %\tiny,\scriptsize,\footnotesize,\small,
\begin{tabular} {|c|c|c|c|c|c|}
\hline
\multicolumn{3}{|c|}{\textbf{{\slang} Language}} & \multicolumn{3}{c|}{\textbf{Memory Management}} \\
\hline
\textbf{Item} & \textbf{LOSP} & \textbf{PM} & \textbf{Item} & \textbf{LOSP} & \textbf{PM}\\
\hline
\textit{Language and Proof Rules} & 700 & \multirow{4}{*}{8} & \textit{Specification} & 400 & \multirow{5}{*}{14} \\
\cline{1-2} \cline{4-5} 
\textit{Lemmas of Language/Semantics} & 3,000 & & \textit{Auxiliary Lemmas/Invariant} & 1,700 & \\
\cline{1-2} \cline{4-5} 
\textit{Soundness} & 7,100 & & \textit{Rely-guarantee Proof of Allocation} & 10,700 & \\
\cline{1-2} \cline{4-5} 
\textit{Invariant} & 100 & & \textit{Rely-guarantee Proof of Free} & 5,000 & \\
\cline{1-5} 
\textit{Security} & 4,800 & 4 & \textit{Proof of Safety and Security} & 400 & \\
\hline
\textbf{Total} & 15,700 & 12 & \textbf{Total} & 18,200 & 14 \\
\hline
\end{tabular}
%\vspace{-6pt}
\label{tbl:stat}
\end{table} 

\subsection{Bugs Found in Zephyr}
During the formal verification, we found 3 bugs in the C code of Zephyr. 
The first two bugs are critical and have been repaired in the latest release of Zephyr. To avoid the third one, callers to \emph{k\_mem\_pool\_alloc} have to constrain the argument \emph{t\_size size}. 
%The integrity issue requires deeper changes on the design of Zephyr. 

%\begin{enumerate}
%\item 
\textbf{(1) Incorrect block split}: this bug is located in the loop in Line 20 of the \emph{k\_mem\_pool\_alloc} service, shown in {\figprefix} \ref{fig:mem_alloc_code}. The \emph{level\_empty} function checks if a pool $p$ has blocks in the free list at level \emph{alloc\_l}. Concurrent threads may release a memory block at that level, making the call to \emph{level\_empty(p, alloc\_l)} to return \emph{false} and stopping the loop. In such case, it allocates a memory block of a bigger capacity at a level $i$ but it still sets the level number of the block as \emph{alloc\_l} at Line 23. The service allocates a larger block to the requesting thread causing an internal fragmentation of $max\_sz / 4 ^ i - max\_sz / 4 ^ {alloc\_l}$ bytes. When this block is released, it will be inserted into the free list at level \emph{alloc\_l}, but not at level $i$, causing an external fragmentation of $max\_sz / 4 ^ i - max\_sz / 4 ^ {alloc\_l}$. 
The bug is fixed by removing the condition \emph{level\_empty(p, alloc\_l)} in our specification. 

%\item 
\textbf{(2) Incorrect return from \emph{k\_mem\_pool\_alloc}}: 
this bug is found at Line 36 in {\figprefix}~\ref{fig:mem_alloc_code}. When a suitable free block is allocated by another thread, the \emph{pool\_alloc} function returns \emph{EAGAIN} at Line 18 to ask the thread to retry the allocation. When a thread invokes \emph{k\_mem\_pool\_alloc} in \emph{FOREVER} mode and this case happens, the service returns \emph{EAGAIN} immediately. However, a thread invoking \emph{k\_mem\_pool\_alloc} in \emph{FOREVER} mode should keep retrying when it does not succeed. We repair the bug by removing the condition $ret == EAGAIN$ at Line 36. As explained in the comments of the C Code, \emph{EAGAIN} should not be returned to threads invoking the service. Moreover, the \emph{return EAGAIN} at Line 48 is actually the case of time out. Thus, we introduce a new return code \emph{ETIMEOUT} in our specification. 

%\item 
\textbf{(3) Non-termination of \emph{k\_mem\_pool\_alloc}}:   
we have discussed that the loop statement at Lines 34 - 47 in {\figprefix} \ref{fig:mem_alloc_code} does not terminate. However, it should terminate in certain cases, which are actually violated in the C code. 
When a thread requests a memory block in \emph{FOREVER} mode and the requested size is larger than \emph{max\_sz}, the maximum size of blocks, the loop at Lines 34 - 47 in {\figprefix} \ref{fig:mem_alloc_code} never finishes since \emph{pool\_alloc} always returns \emph{ENOMEM}. The reason is that the ``\emph{return ENOMEM}'' at Line 15 does not distinguish two cases, $alloc\_l < 0$ and $free\_l < 0$. In the first case, the requested size is larger than \emph{max\_sz} and the kernel service should return immediately. In the second case, there are no free blocks larger than the requested size and the service tries forever until some free block available. We repair the bug by splitting the \emph{if} statement at Lines 13 - 16 into these two cases and introducing a new return code \emph{ESIZEERR} in our specification. Then, we change the condition at Line 36  to check that the returned value is \emph{ESIZEERR} instead of \emph{ENOMEM}.

\subsection{Further Related Work}

%\paragraph{Formal verification of concurrent OSs}
\citeauthor{Klein09a} \cite{Klein09a} presented the first formal verification of the functional correctness and security properties of a general-purpose OS kernel in Isabelle/HOL took roughly 20 person years for 10,000 lines of C code. To reduce the cost of formal verification, \citeauthor{Yang11} \cite{Yang11} demonstrated mechanical verification of Verve, an operating system and run-time system to ensure both the safety and correctness using Boogie and the Z3 SMT solver, which only 2-3 lines of proof annotation per executable statement. 
\citeauthor{Nelson17} \cite{Nelson17} proposed an approach to designing, implementing, and formally verifying the functional correctness of Hyperkernel with a high degree of proof automation and low proof burden with the Z3 SMT solver.
However, all these works did not consider concurrent OS kernels. 

Examples of recent progress in formal verification of concurrent OS kernels are CertiKOS with multicore support \cite{Gu16},
a practical verification framework for preemptive OS kernels to reason about interrupts \cite{Xu16}, and a compositional verification of
interruptible OS kernels with device drivers \cite{Chen16}. %All of the three projects use the Coq theorem prover. 
The Verisoft team \cite{Leinen09,Alka10} applied the VCC framework to formally verify Hyper-V, which is a widely deployed multiprocessor hypervisor by Microsoft consisting of 100 kLOC of concurrent C code and 5 kLOC of assembly. 

To ease the formal verification, a large portion of related work make assumptions on the targeted OS kernel. The formal verification of seL4 \cite{Klein09a} changed the C code of L4 microkernel and thus disabled the in-kernel concurrency. \citeauthor{Nelson17} made the Hyperkernel \cite{Nelson17} interface finite, avoiding unbounded loops, recursion, or complex data structures. 
As opposed to brand new research systems developed for verifiability (e.g. CertiKOS \cite{Gu16}), this article presented the first verification of a 3rd-party existing and realistic concurrent OS. Our formal specification in {\slang} completely corresponds to the execution behavior with fine-grained concurrency of the Zephyr C code. 

%In \cite{Chen16}, the executions of CPU instructions and device transitions can interleave arbitrarily but without shared state between interrupt handlers for device drivers and non-handler kernel code. 

Formal verification of OS memory management has been studied in sequential and concurrent OS kernels, such as CertiKOS \cite{Vaynberg12,Gu16}, seL4 \cite{Klein04,Klein09a}, Verisoft \cite{Alkassar08}, and in the hypervisors from \cite{Blan15,Bolig16}, where only the works in \cite{Gu16,Blan15} considered concurrency. Comparing to buddy memory allocation, the data structures and algorithms verified in \cite{Gu16} are relatively simpler, without block split/coalescence and multiple levels of free lists and bitmaps. The work in \cite{Blan15} only considered virtual mapping but not allocation or deallocation of memory areas. 

Algorithms and implementations of dynamic memory allocation have been formally specified and verified in an extensive number of works \cite{Yu03,Fang17a,Marti06,Su16,Fang17b,Fang18}. However, the buddy memory allocation was only studied in \cite{Fang18}, which did not consider concrete data structures (e.g. bitmaps) and concurrency. %To the best of our knowledge, 
A memory model \cite{Saraswat07} provides the necessary abstraction to separate the behaviour of a program from the behaviour of the memory it reads and writes. There are many formalizations of memory models in the literature, e.g., \cite{Leroy2008,Tews2009,Gallardo2009,Sevcik13,Mansky15}, where some of them only created an abstract specification of the services for memory allocation and release \cite{Gallardo2009,Sevcik13,Mansky15}. 
Our article presents the first formal specification and mechanized proof for concurrent memory allocation of a realistic operating system. 

Integrity is a sort of information-flow security (IFS) which deals with the problem of preventing improper release and modification of information in complex systems. 
Language-based IFS \cite{sabel03} defines security policies on programming languages and concerns the data confidentiality among program
variables. Compositional verification of language-based IFS has been conducted in \cite{Mantel11,Murray16,Murray18}. 
Formal verification of IFS on OS kernels need to considers the events (e.g. kernel services, interrupt handlers) rather than on pure programs. Therefore, state-event IFS \cite{rushby92,Oheimb04,Murray12} is usually applied to OS kernels (e.g. \cite{Murray13,Dam13,Costanzo16,zhao19tdsc}). 
However, formal verification of state-event IFS for concurrent systems (e.g. OS kernels) has not been addressed in literature. 
Our article presents the first integrity verification of concurrent OS kernels. 

%concurrent IFS
%memory separation of ED separation kernel

%The Hyperkernel verifier does not reason about multicore or interrupts in the kernel
%concurrent OS verification: CertiKOS etc. 

%The Verisoft team \cite{Leinen09,Alka10} applied the VCC framework [15] to formally verify Hyper-V, which is a widely deployed multiprocessor hypervisor by Microsoft consisting of 100 kLOC of concurrent C code and 5 kLOC of assembly. However, only 20\% of the code is verified; it is also only verified for function contracts and type invariants, not the full functional correctness property.

%%\paragraph{Compositional verification of information-flow security}
%%there are compositional verification approaches of ifs. we applied the integrity to a realistic RTOS at implementation level. 
%%
%%verification integrity, ifs, etc. for concurrent systems

\subsection{Limitations and Discussion}

The state of the art of formal verification of OS kernels focuses on the implementation level (e.g. \cite{Klein09a,Yang11,Gu16,Nelson17}). One limitation of this work is that formal verification is enforced at the low-level design specification. 

The first concern of this decision is that our work aims at the highest evaluation assurance level (EAL 7) of Common Criteria (CC) \cite{cc}, which was declared as the candidate standard for security certification by the Zephyr project. With regard to the EAL 7, a main requirement of functional specification addressed by formal methods is a complete formal and modular design of Target of Evaluation (TOE) with security proofs, rather than mandating the formal verification at source code level. In this article, we develop a fine-grained low level formal specification of Zephyr. The specification closely follows the Zephyr C code, and thus is able to do the \emph{code-to-spec} review required by the EAL 7 evaluation, covering all the data structures and imperative statements present in the implementation.

Second, formally verifying functional correctness, safety and security of concurrent C programs, in particular the memory management of Zephyr, is not well supported by the state of the art of C verifiers (e.g. VCC \cite{Kroe14} \cite{Cohen09}, Frame-C \cite{Kirchner15}, CBMC). 
%Formal C verifiers are not applicable to Zephyr memory management yet. 
\emph{Simpl} \cite{Schirm06} is a generic imperative language embedded into Isabelle/HOL that was designed as an intermediate language for program verification. In the seL4 project, the C code was translated into \emph{Simpl} and then a state monad representation by the \emph{CParser} and \emph{AutoCorres} tools, which do not support concurrent C programs. Though we have extended \emph{Simpl} to \emph{CSimpl} by concurrent statements and a rely-guarantee proof system in \cite{Sanan2017}, a new parser for \emph{CSimpl} is still under development. 

Another limitation and thus a challenge for concurrent OS kernels is scalability. Our work has been carried out in the Isabelle/HOL interactive theorem prover, and thus is labor-intensive. We developed specification and proof, 40 times more than the lines of the C code due to the in-kernel concurrency, where the complicated invariant proofs represent the largest part. It is certainly that Isabelle/HOL has high degree of proof automation by integrating various SMT solvers and automatic provers. Proof automation for OS kernels has been made significant progress in recent years, successful examples of this direction are \cite{Yang11,Nelson17} closely reaching 100\%. However, automatic verification of concurrent OS kernels is still a challenging \cite{Nelson17}. Moreover, higher degree of automation implies simpler specification or properties of kernels. For instance, to enable automated verification with SMT solvers, \citeauthor{Nelson17} made the Hyperkernel interface finite, avoiding unbounded loops, recursion, or complex data structures. 

With regards to security, currently {\slang} only supports integrity of concurrent reactive systems. In order for {\slang} to support confidentiality, it is necessary some modification on its core, and more specifically in the labels associated with actions. This change allows the security framework to obtain the necessary information to properly generate the possible parallel execution traces. Additionally, seeking compositional reasoning at both the security (by applying the step consistency unwinding condition), and concurrent (by applying rely-guarantee) levels, we are adding a simulation framework for {\slang} similar to CSim$^2$ for CSimpl~\cite{Sanan20}. This new framework will provide a new set of inference rules for the compositional verification of confidentiality preservation of {\slang} specifications.

Finally, for the purpose of EAL7 evaluation of the concurrent Zephyr RTOS, modular compositionality of the specification and proof is a necessary approach, such as the technique of certified abstract layers \cite{Gu15popl,Gu18pldi}. Due to the concept of ``event'' in {\slang} and their sequential composition for event systems as well as their parallel composition for the whole system, the formal specification of {\slang} is compositional. Moreover, the formal proof is compositional as well thanks to the compositional reasoning approach in the  {\slang} proof system for functional correctness, safety and security.  %A limitation of {\slang} is how to compose the state of different modules of OS kernels whilst making few changes to the functional specification and formal proof. 

%%
%%memory separation: data in memory blocks of a thread should be cleaned before release the blocks. 
%%we only verify integrity which means a thread cannot change memory of other threads, how about memory reading?

\section{Conclusion and Future Work} 
\label{sect:concl}

In this article, we have developed a formal specification at low-level design of the concurrent buddy memory management of Zephyr RTOS. Using the rely-guarantee technique in the {\slang} framework, we have formally verified a set of critical properties for OS kernels such as correctness, safety and security. Finally, we identified some critical bugs in the C code of Zephyr. 

Our work explores the challenges and cost of certifying concurrent OSs for the highest-level assurance. The definition of properties and rely-guarantee relations is complex and the verification task becomes expensive. We used $\approx$ 40 times of LOS/LOP than the C code at low-level design. 
Next, we are planning to verify other modules of Zephyr, which may be easier due to simpler data structures and algorithms. 
We are also working on extending {\slang} to support verification of confidentiality on concurrent reactive systems, which together with the already supported integrity, will allow to verify non-interference of {\slang} specifications. 
For the purpose of fully formal verification of OSs at source code level, we will replace the imperative language in {\slang} by a more expressive one and add a verification condition generator (VCG) to reduce the cost of the verification. To improve modularity, we are also working on horizontal verification, in order to compose the state of different modules of OS kernels whilst making few changes to the functional specification and formal proof.

%
% The acknowledgments section is defined using the "acks" environment (and NOT an unnumbered section). This ensures
% the proper identification of the section in the article metadata, and the consistent spelling of the heading.
%%\begin{acks}
%%%To Robert, for the bagels and explaining CMYK and color spaces.
%%\end{acks}

%
% The next two lines define the bibliography style to be used, and the bibliography file.

\bibliographystyle{ACM-Reference-Format}
\bibliography{paperref}

\appendix
\section{C Code of \emph{k\_mem\_pool\_free}}
\label{appx:mempoolfree_c}

\begin{lstlisting}
static void free_block(struct k_mem_pool *p, int level, size_t *lsizes, int bn)
{
  int i, key, lsz = lsizes[level];
  void *block = block_ptr(p, lsz, bn);

  key = irq_lock();

  set_free_bit(p, level, bn);

  if (level && partner_bits(p, level, bn) == 0xf) {
    for (i = 0; i < 4; i++) {
      int b = (bn & ~3) + i;

      clear_free_bit(p, level, b);
      if (b != bn && block_fits(p, block_ptr(p, lsz, b), lsz)) {
        sys_dlist_remove(block_ptr(p, lsz, b));
      }
    }

    irq_unlock(key);
    free_block(p, level-1, lsizes, bn / 4); /* tail recursion! */
    return;
  }

  if (block_fits(p, block, lsz)) {
    sys_dlist_append(&p->levels[level].free_list, block);
  }

  irq_unlock(key);
}

void k_mem_pool_free(struct k_mem_block *block)
{
  int i, key, need_sched = 0;
  struct k_mem_pool *p = get_pool(block->id.pool);
  size_t lsizes[p->n_levels];

  /* As in k_mem_pool_alloc(), we build a table of level sizes
  * to avoid having to store it in precious RAM bytes.
  * Overhead here is somewhat higher because free_block()
  * doesn't inherently need to traverse all the larger
  * sublevels.
  */
  lsizes[0] = _ALIGN4(p->max_sz);
  for (i = 1; i <= block->id.level; i++) {
    lsizes[i] = _ALIGN4(lsizes[i-1] / 4);
  }

  free_block(get_pool(block->id.pool), block->id.level, lsizes, block->id.block);

  /* Wake up anyone blocked on this pool and let them repeat
   * their allocation attempts
   */
  key = irq_lock();

  while (!sys_dlist_is_empty(&p->wait_q)) {
    struct k_thread *th = (void *)sys_dlist_peek_head(&p->wait_q);

    _unpend_thread(th);
    _abort_thread_timeout(th);
    _ready_thread(th);
    need_sched = 1;
  }

  if (need_sched && !_is_in_isr()) {
    _reschedule_threads(key);
  } else {
    irq_unlock(key);
  }
}
\end{lstlisting}

\section{Specification and Proof Sketch of \emph{k\_mem\_pool\_free}}
\label{appx:mempoolfree}

The formal specification of \emph{k\_mem\_pool\_free} (in \emph{black} color) and its rely-guarantee proof sketch (in \emph{blue} color) are shown as follows. 

\vspace{5mm}
\isabellestyle{sl} %tt,it,literal,sl
\begin{isabellec} \fontsize{8pt}{0cm} %\footnotesize %\scriptsize
\specrg{\isacommand{Mem{\isacharunderscore}pool{\isacharunderscore}free{\isacharunderscore}pre}\ t\ {\isasymequiv}\ {\isasymlbrace} {\isasymacute}inv\ {\isasymand}\ {\isasymacute}allocating{\isacharunderscore}node\ t\ {\isacharequal}\ None\ {\isasymand}\ {\isasymacute}freeing{\isacharunderscore}node\ t\ {\isacharequal}\ None{\isasymrbrace}}

\isacommand{EVENT}\ \isacommand{Mem{\isacharunderscore}pool{\isacharunderscore}free}\ [Block\ b]\ @\ {\isacharparenleft}$\mathcal{T}$\ t{\isacharparenright}\isanewline
\isacommand{WHEN}\isanewline
\quad pool\ b\ {\isasymin}\ {\isasymacute}mem{\isacharunderscore}pools\ \isanewline
\quad {\isasymand}\ level\ b\ {\isacharless}\ length\ {\isacharparenleft}levels\ {\isacharparenleft}{\isasymacute}mem{\isacharunderscore}pool{\isacharunderscore}info\ {\isacharparenleft}pool\ b{\isacharparenright}{\isacharparenright}{\isacharparenright}\isanewline
\quad {\isasymand}\ block\ b\ {\isacharless}\ length\ {\isacharparenleft}bits\ {\isacharparenleft}levels\ {\isacharparenleft}{\isasymacute}mem{\isacharunderscore}pool{\isacharunderscore}info\ {\isacharparenleft}pool\ b{\isacharparenright}{\isacharparenright}{\isacharbang}{\isacharparenleft}level\ b{\isacharparenright}{\isacharparenright}{\isacharparenright}\isanewline
\quad {\isasymand}\ data\ b\ {\isacharequal}\ block{\isacharunderscore}ptr\ {\isacharparenleft}{\isasymacute}mem{\isacharunderscore}pool{\isacharunderscore}info\ {\isacharparenleft}pool\ b{\isacharparenright}{\isacharparenright}\isanewline
\quad \quad \quad \quad \quad  {\isacharparenleft}{\isacharparenleft}ALIGN{\isadigit{4}}\ {\isacharparenleft}max{\isacharunderscore}sz\ {\isacharparenleft}{\isasymacute}mem{\isacharunderscore}pool{\isacharunderscore}info\ {\isacharparenleft}pool\ b{\isacharparenright}{\isacharparenright}{\isacharparenright}{\isacharparenright}\ div\ {\isacharparenleft}{\isadigit{4}}\ {\isacharcircum}\ {\isacharparenleft}level\ b{\isacharparenright}{\isacharparenright}{\isacharparenright}\ {\isacharparenleft}block\ b{\isacharparenright}\isanewline
\isacommand{THEN}

\quad \specrg{Mem{\isacharunderscore}pool{\isacharunderscore}free{\isacharunderscore}pre\ t\ {\isasyminter}\ {\isasymlbrace} g {\isasymrbrace}}
\quad \speccomment{(* g is the guard condition of the event *)}

\quad \speccomment{{\isacharparenleft}{\isacharasterisk}\ here\ we\ set\ the\ bit\ to\ FREEING{\isacharcomma}\ so\ that\ other\ thread\ cannot\ mem{\isacharunderscore}pool{\isacharunderscore}free\ the\ same\ block\ \isanewline
\ \ \ \ \ \ \ it\ also\ requires\ that\ it\ can\ only\ free\ ALLOCATED\ block\ {\isacharasterisk}{\isacharparenright}}

\quad t\ \isactrlenum \ \isacommand{AWAIT}\ {\isacharparenleft}bits\ {\isacharparenleft}{\isacharparenleft}levels\ {\isacharparenleft}{\isasymacute}mem{\isacharunderscore}pool{\isacharunderscore}info\ {\isacharparenleft}pool\ b{\isacharparenright}{\isacharparenright}{\isacharparenright}\ {\isacharbang}\ {\isacharparenleft}level\ b{\isacharparenright}{\isacharparenright}{\isacharparenright}\ {\isacharbang}\ {\isacharparenleft}block\ b

\quad \quad \quad \quad \quad \quad \quad \quad \quad {\isacharequal}\ ALLOCATED\ \isacommand{THEN}

\quad \quad \quad \quad {\isasymacute}mem{\isacharunderscore}pool{\isacharunderscore}info\ {\isacharcolon}{\isacharequal}\ set{\isacharunderscore}bit{\isacharunderscore}freeing\ {\isasymacute}mem{\isacharunderscore}pool{\isacharunderscore}info\ {\isacharparenleft}pool\ b{\isacharparenright}\ {\isacharparenleft}level\ b{\isacharparenright}\ {\isacharparenleft}block\ b{\isacharparenright}{\isacharsemicolon}{\isacharsemicolon}

\quad \quad \quad \quad {\isasymacute}freeing{\isacharunderscore}node\ {\isacharcolon}{\isacharequal}\ {\isasymacute}freeing{\isacharunderscore}node\ {\isacharparenleft}t\ {\isacharcolon}{\isacharequal}\ Some\ b{\isacharparenright}

\quad \quad \quad \isacommand{END}{\isacharparenright}{\isacharsemicolon}{\isacharsemicolon}

\quad \specrg{\isacommand{mp{\isacharunderscore}free{\isacharunderscore}precond{\isadigit{1}}}\ t\ b\ {\isasymequiv}\ {\isasymlbrace} {\isasymacute}inv\ {\isasymand}\ {\isasymacute}allocating{\isacharunderscore}node\ t\ {\isacharequal}\ None\ {\isasymand} \ g\ {\isasymand}\ {\isasymacute}freeing{\isacharunderscore}node\ t\ {\isacharequal}\ Some\ b{\isasymrbrace}}

\quad \speccomment{(* remove the mem block from the thread's allocated block set *)}

\quad t\ \isactrlenum \ {\isasymacute}mblocks := {\isasymacute}mblocks(t:={\isasymacute}mblocks\ t\ - \{{\isasymlbrace} pool = (pool\ b), level={\isasymacute}lvl\ t,block={\isasymacute}bn\ t, 

\quad \quad \quad \quad \quad \quad \quad \quad \quad 
data=block\_ptr\ ({\isasymacute}mem\_pool\_info\ (pool\ b))

\quad \quad \quad \quad \quad \quad \quad \quad \quad \quad \quad
(((ALIGN4\ (max\_sz\ ({\isasymacute}mem\_pool\_info\ (pool\ b))))\ div\ (4 {\isacharcircum} ({\isasymacute}lvl\ t))))\ ({\isasymacute}bn\ t) {\isasymrbrace}\});;

\quad \specrg{\isacommand{mp{\isacharunderscore}free{\isacharunderscore}precond{\isadigit{2}}}\ t\ b\ {\isasymequiv}\
\isacommand{mp{\isacharunderscore}free{\isacharunderscore}precond{\isadigit{1}}}\ t\ b}

\quad t\ \isactrlenum \ {\isasymacute}need{\isacharunderscore}resched\ {\isacharcolon}{\isacharequal}\ {\isasymacute}need{\isacharunderscore}resched{\isacharparenleft}t\ {\isacharcolon}{\isacharequal}\ False{\isacharparenright}{\isacharsemicolon}{\isacharsemicolon}

\quad \specrg{\isacommand{mp{\isacharunderscore}free{\isacharunderscore}precond{\isadigit{3}}}\ t\ b\ {\isasymequiv}\ {\isacharparenleft}mp{\isacharunderscore}free{\isacharunderscore}precond{\isadigit{2}}\ t\ b{\isacharparenright}\ {\isasyminter}\ {\isasymlbrace}{\isasymacute}need{\isacharunderscore}resched\ t\ {\isacharequal}\ False{\isasymrbrace}}

\quad t\ \isactrlenum \ {\isasymacute}lsizes\ {\isacharcolon}{\isacharequal}\ {\isasymacute}lsizes{\isacharparenleft}t\ {\isacharcolon}{\isacharequal}\ {\isacharbrackleft}ALIGN{\isadigit{4}}\ {\isacharparenleft}max{\isacharunderscore}sz\ {\isacharparenleft}{\isasymacute}mem{\isacharunderscore}pool{\isacharunderscore}info\ {\isacharparenleft}pool\ b{\isacharparenright}{\isacharparenright}{\isacharparenright}{\isacharbrackright}{\isacharparenright}{\isacharsemicolon}{\isacharsemicolon}

\quad \specrg{\isacommand{mp{\isacharunderscore}free{\isacharunderscore}precond{\isadigit{4}}}\ t\ b\ {\isasymequiv}\ \isanewline
\quad \quad mp{\isacharunderscore}free{\isacharunderscore}precond{\isadigit{3}}\ t\ b\ {\isasyminter}\ {\isasymlbrace}{\isasymacute}lsizes\ t\ {\isacharequal}\ {\isacharbrackleft}ALIGN{\isadigit{4}}\ {\isacharparenleft}max{\isacharunderscore}sz\ {\isacharparenleft}{\isasymacute}mem{\isacharunderscore}pool{\isacharunderscore}info\ {\isacharparenleft}pool\ b{\isacharparenright}{\isacharparenright}{\isacharparenright}{\isacharbrackright}{\isasymrbrace}}

\quad \isacommand{FOR}\ {\isacharparenleft}t\ \isactrlenum \ {\isasymacute}i\ {\isacharcolon}{\isacharequal}\ {\isasymacute}i{\isacharparenleft}t\ {\isacharcolon}{\isacharequal}\ {\isadigit{1}}{\isacharparenright}{\isacharparenright}{\isacharsemicolon}\  {\isasymacute}i\ t\ {\isasymle}\ level\ b{\isacharsemicolon}\ \ {\isacharparenleft}t\ \isactrlenum \ {\isasymacute}i\ {\isacharcolon}{\isacharequal}\ {\isasymacute}i{\isacharparenleft}t\ {\isacharcolon}{\isacharequal}\ {\isasymacute}i\ t\ {\isacharplus}\ {\isadigit{1}}{\isacharparenright}{\isacharparenright}\ \isacommand{DO}\isanewline
\quad \quad t\ \isactrlenum \ {\isasymacute}lsizes\ {\isacharcolon}{\isacharequal}\ {\isasymacute}lsizes{\isacharparenleft}t\ {\isacharcolon}{\isacharequal}\ {\isasymacute}lsizes\ t\ {\isacharat}\ {\isacharbrackleft}ALIGN{\isadigit{4}}\ {\isacharparenleft}{\isasymacute}lsizes\ t\ {\isacharbang}\ {\isacharparenleft}{\isasymacute}i\ t\ {\isacharminus}\ {\isadigit{1}}{\isacharparenright}\ div\ {\isadigit{4}}{\isacharparenright}{\isacharbrackright}{\isacharparenright}

\quad \isacommand{ROF}{\isacharsemicolon}{\isacharsemicolon}

\quad \specrg{\isacommand{mp{\isacharunderscore}free{\isacharunderscore}precond{\isadigit{5}}}\ t\ b\ {\isasymequiv}\ mp{\isacharunderscore}free{\isacharunderscore}precond{\isadigit{3}}\ t\ b \ {\isasyminter} \isanewline
\quad \quad {\isasymlbrace}{\isacharparenleft}{\isasymforall}ii{\isacharless}length\ {\isacharparenleft}{\isasymacute}lsizes\ t{\isacharparenright}{\isachardot}\ {\isasymacute}lsizes\ t\ {\isacharbang}\ ii\ {\isacharequal}\ {\isacharparenleft}ALIGN{\isadigit{4}}\ {\isacharparenleft}max{\isacharunderscore}sz\ {\isacharparenleft}{\isasymacute}mem{\isacharunderscore}pool{\isacharunderscore}info\ {\isacharparenleft}pool\ b{\isacharparenright}{\isacharparenright}{\isacharparenright}{\isacharparenright}\isanewline
\quad \quad \quad \quad \quad \quad \quad \quad \quad \quad \quad \quad \quad \quad \quad \quad \quad \quad \quad \quad 
div\ {\isacharparenleft}{\isadigit{4}}\ {\isacharcircum}\ ii{\isacharparenright}{\isacharparenright}\ {\isasymand}\ length\ {\isacharparenleft}{\isasymacute}lsizes\ t{\isacharparenright}\ {\isachargreater}\ level\ b{\isasymrbrace}}

\quad \speccomment{{\isacharparenleft}{\isacharasterisk}\ {\isacharequal} {\isacharequal} {\isacharequal} start{\isacharcolon}\ free{\isacharunderscore}block{\isacharparenleft}pool{\isacharcomma}\ level{\isacharcomma}\ lsizes{\isacharcomma}\ block{\isacharparenright}{\isacharsemicolon}\ {\isacharequal} {\isacharequal} {\isacharequal}{\isacharasterisk}{\isacharparenright}}

\quad t\ \isactrlenum \ {\isasymacute}free{\isacharunderscore}block{\isacharunderscore}r\ {\isacharcolon}{\isacharequal}\ {\isasymacute}free{\isacharunderscore}block{\isacharunderscore}r\ {\isacharparenleft}t\ {\isacharcolon}{\isacharequal}\ True{\isacharparenright}{\isacharsemicolon}{\isacharsemicolon}

\quad \specrg{\isacommand{mp{\isacharunderscore}free{\isacharunderscore}precond{\isadigit{6}}}\ t\ b\ {\isasymequiv}\ mp{\isacharunderscore}free{\isacharunderscore}precond{\isadigit{5}}\ t\ b\ {\isasyminter}\ {\isasymlbrace}{\isasymacute}free{\isacharunderscore}block{\isacharunderscore}r\ t\ {\isacharequal}\ True{\isasymrbrace}}

\quad t\ \isactrlenum \ {\isasymacute}bn\ {\isacharcolon}{\isacharequal}\ {\isasymacute}bn\ {\isacharparenleft}t\ {\isacharcolon}{\isacharequal}\ block\ b{\isacharparenright}{\isacharsemicolon}{\isacharsemicolon}

\quad \specrg{\isacommand{mp{\isacharunderscore}free{\isacharunderscore}precond{\isadigit{7}}}\ t\ b\ {\isasymequiv}\ mp{\isacharunderscore}free{\isacharunderscore}precond{\isadigit{6}}\ t\ b\ {\isasyminter}\ {\isasymlbrace}{\isasymacute}bn\ t\ {\isacharequal}\ block\ b{\isasymrbrace}}

\quad t\ \isactrlenum \ {\isasymacute}lvl\ {\isacharcolon}{\isacharequal}\ {\isasymacute}lvl\ {\isacharparenleft}t\ {\isacharcolon}{\isacharequal}\ level\ b{\isacharparenright}{\isacharsemicolon}{\isacharsemicolon}

\quad \specrg{
\isacommand{mp{\isacharunderscore}free{\isacharunderscore}loopinv}\ t\ b\ {\isasymalpha}
}\isanewline
\quad \isacommand{WHILE}\ {\isasymacute}free{\isacharunderscore}block{\isacharunderscore}r\ t\ \isacommand{DO}\isanewline
\quad \specrg{
\isacommand{mp{\isacharunderscore}free{\isacharunderscore}cnd1}\ t\ b\ {\isasymalpha}\ {\isasymequiv} mp{\isacharunderscore}free{\isacharunderscore}loopinv\ t\ b\ {\isasymalpha} {\isasyminter} {\isasymlbrace} {\isasymalpha} {\isachargreater} 0 {\isasymrbrace}
}\isanewline
\quad \quad t\ {\isactrlenum} \ {\isasymacute}lsz\ {\isacharcolon}{\isacharequal}\ {\isasymacute}lsz\ {\isacharparenleft}t\ {\isacharcolon}{\isacharequal}\ {\isasymacute}lsizes\ t\ {\isacharbang}\ {\isacharparenleft}{\isasymacute}lvl\ t{\isacharparenright}{\isacharparenright}{\isacharsemicolon}{\isacharsemicolon}\isanewline
\quad \specrg{
\isacommand{mp{\isacharunderscore}free{\isacharunderscore}cnd2}\ t\ b\ {\isasymalpha}\ {\isasymequiv} mp{\isacharunderscore}free{\isacharunderscore}cnd1 \ t\ b\ {\isasymalpha} {\isasyminter} {\isasymlbrace} {\isasymacute}lsz\ t\ {\isacharequal}\ {\isasymacute}lsizes\ t\ {\isacharbang}\ {\isacharparenleft}{\isasymacute}lvl\ t{\isacharparenright} {\isasymrbrace}
}\isanewline
\quad \quad t\ {\isactrlenum} \ {\isasymacute}blk\ {\isacharcolon}{\isacharequal}\ {\isasymacute}blk\ {\isacharparenleft}t\ {\isacharcolon}{\isacharequal}\ block{\isacharunderscore}ptr\ {\isacharparenleft}{\isasymacute}mem{\isacharunderscore}pool{\isacharunderscore}info\ {\isacharparenleft}pool\ b{\isacharparenright}{\isacharparenright}\ {\isacharparenleft}{\isasymacute}lsz\ t{\isacharparenright}\ {\isacharparenleft}{\isasymacute}bn\ t{\isacharparenright}{\isacharparenright}{\isacharsemicolon}{\isacharsemicolon}\isanewline
\quad \specrg{
\isacommand{mp{\isacharunderscore}free{\isacharunderscore}cnd3}\ t\ b\ {\isasymalpha}\ {\isasymequiv} mp{\isacharunderscore}free{\isacharunderscore}cnd2 \ t\ b\ {\isasymalpha} {\isasyminter} \isanewline
\quad \quad \quad \quad \quad \quad \quad \quad \quad \quad {\isasymlbrace} {\isasymacute}blk\ t\ {\isacharequal}\ block{\isacharunderscore}ptr\ {\isacharparenleft}{\isasymacute}mem{\isacharunderscore}pool{\isacharunderscore}info\ {\isacharparenleft}pool\ b{\isacharparenright}{\isacharparenright}\ {\isacharparenleft}{\isasymacute}lsz\ t{\isacharparenright}\ {\isacharparenleft}{\isasymacute}bn\ t{\isacharparenright} {\isasymrbrace}
}\isanewline
\quad \quad t\ {\isactrlenum} \ \isacommand{ATOM}\isanewline
\quad \quad \specrg{\{V1\} \speccomment{{\isacharparenleft}V1 {\isasymin} mp{\isacharunderscore}free{\isacharunderscore}cnd3 \ t\ b\ {\isasymalpha} {\isasyminter} {\isasymlbrace}{\isasymacute}cur\ {\isacharequal}\ Some\ t{\isasymrbrace}{\isacharparenright}}}
\isanewline
\quad \quad \quad {\isasymacute}mem{\isacharunderscore}pool{\isacharunderscore}info\ {\isacharcolon}{\isacharequal}\ set{\isacharunderscore}bit{\isacharunderscore}free\ {\isasymacute}mem{\isacharunderscore}pool{\isacharunderscore}info\ {\isacharparenleft}pool\ b{\isacharparenright}\ {\isacharparenleft}{\isasymacute}lvl\ t{\isacharparenright}\ {\isacharparenleft}{\isasymacute}bn\ t{\isacharparenright}{\isacharsemicolon}{\isacharsemicolon}\isanewline
\quad \quad \specrg{\{V2\} \speccomment{{\isacharparenleft}V2 = V1{\isasymlparr}mem{\isacharunderscore}pool{\isacharunderscore}info\ {\isacharcolon}{\isacharequal}\isanewline
\quad \quad \quad \quad \quad \quad \quad \quad \quad set{\isacharunderscore}bit{\isacharunderscore}free\ {\isacharparenleft}mem{\isacharunderscore}pool{\isacharunderscore}info\ V1{\isacharparenright}\ {\isacharparenleft}pool\ b{\isacharparenright}\ {\isacharparenleft}lvl\ V1\ t{\isacharparenright}\ {\isacharparenleft}bn\ V1\ t{\isacharparenright}{\isasymrparr}{\isacharparenright}}}
\isanewline
\quad \quad \quad {\isasymacute}freeing{\isacharunderscore}node\ {\isacharcolon}{\isacharequal}\ {\isasymacute}freeing{\isacharunderscore}node\ {\isacharparenleft}t\ {\isacharcolon}{\isacharequal}\ None{\isacharparenright}{\isacharsemicolon}{\isacharsemicolon}\isanewline
\quad \quad \specrg{\{V3\} \speccomment{{\isacharparenleft}V3 = V2{\isasymlparr}freeing{\isacharunderscore}node\ {\isacharcolon}{\isacharequal}\ {\isacharparenleft}freeing{\isacharunderscore}node\ V2{\isacharparenright}{\isacharparenleft}t\ {\isacharcolon}{\isacharequal}\ None{\isacharparenright}{\isasymrparr}{\isacharparenright}}}
\isanewline
\quad \quad \quad \isacommand{IF}\ {\isasymacute}lvl\ t\ {\isachargreater}\ {\isadigit{0}}\ {\isasymand}\ partner{\isacharunderscore}bits\ {\isacharparenleft}{\isasymacute}mem{\isacharunderscore}pool{\isacharunderscore}info\ {\isacharparenleft}pool\ b{\isacharparenright}{\isacharparenright}\ {\isacharparenleft}{\isasymacute}lvl\ t{\isacharparenright}\ {\isacharparenleft}{\isasymacute}bn\ t{\isacharparenright}\ \isacommand{THEN}

\quad \quad \quad
\speccomment{(V3 {\isasymin} {\isasymlbrace}NULL\ {\isacharless}\ {\isasymacute}lvl\ t\ {\isasymand}\ partner{\isacharunderscore}bits\ {\isacharparenleft}{\isasymacute}mem{\isacharunderscore}pool{\isacharunderscore}info\ {\isacharparenleft}pool\ b{\isacharparenright}{\isacharparenright}\ {\isacharparenleft}{\isasymacute}lvl\ t{\isacharparenright}\ {\isacharparenleft}{\isasymacute}bn\ t{\isacharparenright}{\isasymrbrace})}

\quad \quad \specrg{
\isacommand{mergeblock{\isacharunderscore}loopinv}\ V3 \ t\ b\ {\isasymalpha}\ {\isasymequiv} \isanewline
\quad \quad {\isacharbraceleft}V{\isachardot}\ let\ minf{\isadigit{0}}\ {\isacharequal}\ {\isacharparenleft}mem{\isacharunderscore}pool{\isacharunderscore}info\ V3{\isacharparenright}{\isacharparenleft}pool\ b{\isacharparenright}{\isacharsemicolon}
\ lvl{\isadigit{0}}\ {\isacharequal}\ {\isacharparenleft}levels\ minf{\isadigit{0}}{\isacharparenright}\ {\isacharbang}\ {\isacharparenleft}lvl\ V3\ t{\isacharparenright}{\isacharsemicolon}\isanewline
\quad \quad \quad \quad \quad minf{\isadigit{1}}\ {\isacharequal}\ {\isacharparenleft}mem{\isacharunderscore}pool{\isacharunderscore}info\ V{\isacharparenright}{\isacharparenleft}pool\ b{\isacharparenright}{\isacharsemicolon}
\ lvl{\isadigit{1}}\ {\isacharequal}\ {\isacharparenleft}levels\ minf{\isadigit{1}}{\isacharparenright}\ {\isacharbang}\ {\isacharparenleft}lvl\ V3\ t{\isacharparenright}\ in\ \isanewline
\quad \quad \quad \quad {\isacharparenleft}bits\ lvl{\isadigit{1}}\ {\isacharequal}\ list{\isacharunderscore}updates{\isacharunderscore}n\ {\isacharparenleft}bits\ lvl{\isadigit{0}}{\isacharparenright}\ {\isacharparenleft}{\isacharparenleft}bn\ V3\ t\ div\ {\isadigit{4}}{\isacharparenright}\ {\isacharasterisk}\ {\isadigit{4}}{\isacharparenright}\ {\isacharparenleft}i\ V\ t{\isacharparenright}\ NOEXIST{\isacharparenright}\isanewline
\quad \quad \quad \quad {\isasymand}\ {\isacharparenleft}free{\isacharunderscore}list\ lvl{\isadigit{1}}\ {\isacharequal}\ removes\ {\isacharparenleft}map\ {\isacharparenleft}{\isasymlambda}ii{\isachardot}\ block{\isacharunderscore}ptr\ minf{\isadigit{0}}\ {\isacharparenleft}lsz\ V3\ t{\isacharparenright}\isanewline
\quad \quad \quad \quad \quad \quad \quad \quad \quad \quad \quad \quad \quad {\isacharparenleft}{\isacharparenleft}bn\ V3\ t\ div\ {\isadigit{4}}{\isacharparenright}\ {\isacharasterisk}\ {\isadigit{4}}\ {\isacharplus}\ ii{\isacharparenright}{\isacharparenright}\ {\isacharbrackleft}{\isadigit{0}}{\isachardot}{\isachardot}{\isacharless}{\isacharparenleft}i\ V\ t{\isacharparenright}{\isacharbrackright}{\isacharparenright}\ {\isacharparenleft}free{\isacharunderscore}list\ lvl{\isadigit{0}}{\isacharparenright}{\isacharparenright}

\quad \quad \quad \quad {\isasymand}\ {\isacharparenleft}wait{\isacharunderscore}q\ minf{\isadigit{0}}\ {\isacharequal}\ wait{\isacharunderscore}q\ minf{\isadigit{1}}{\isacharparenright}\ {\isasymand}\ {\isacharparenleft}{\isasymforall}t{\isacharprime}{\isachardot}\ t{\isacharprime}\ {\isasymnoteq}\ t\ {\isasymlongrightarrow}\ lvars{\isacharunderscore}nochange\ t{\isacharprime}\ V\ V3{\isacharparenright}

\quad \quad \quad \quad {\isasymand}\ {\isacharparenleft}{\isasymforall}p{\isachardot}\ p\ {\isasymnoteq}\ pool\ b\ {\isasymlongrightarrow}\ mem{\isacharunderscore}pool{\isacharunderscore}info\ V\ p\ {\isacharequal}\ mem{\isacharunderscore}pool{\isacharunderscore}info\ V3\ p{\isacharparenright}\isanewline
\quad \quad \quad \quad {\isasymand}\ {\isacharparenleft}{\isasymforall}j{\isachardot}\ j\ {\isasymnoteq}\ lvl\ V3\ t\ {\isasymlongrightarrow}\ {\isacharparenleft}levels\ minf{\isadigit{0}}{\isacharparenright}{\isacharbang}j\ {\isacharequal}\ {\isacharparenleft}levels\ minf{\isadigit{1}}{\isacharparenright}{\isacharbang}j{\isacharparenright}

\quad \quad \quad \quad {\isasymand}\ {\isacharparenleft}V{\isacharcomma}V3{\isacharparenright}{\isasymin}gvars{\isacharunderscore}conf{\isacharunderscore}stable\ {\isasymand} \ i\ V\ t\ {\isasymle}\ \isadigit{4} \ {\isasymand} \ {\isasymand} \ {\isasymalpha} = 4 - i\ V\ t\ ...... {\isacharbraceright}
}

\quad \quad \quad \quad \isacommand{FOR}\ {\isasymacute}i\ {\isacharcolon}{\isacharequal}\ {\isasymacute}i{\isacharparenleft}t\ {\isacharcolon}{\isacharequal}\ {\isadigit{0}}{\isacharparenright}{\isacharsemicolon}\ {\isasymacute}i\ t\ {\isacharless}\ {\isadigit{4}}{\isacharsemicolon}\ {\isasymacute}i\ {\isacharcolon}{\isacharequal}\ {\isasymacute}i{\isacharparenleft}t\ {\isacharcolon}{\isacharequal}\ {\isasymacute}i\ t\ {\isacharplus}\ {\isadigit{1}}{\isacharparenright}\ \isacommand{DO}

\quad \quad \quad \quad \specrg{
\isacommand{mergeblock{\isacharunderscore}loopinv}\ V3 \ t\ b\ {\isasymalpha}\ {\isasyminter}\ {\isasymlbrace} {\isasymalpha} {\isachargreater} 0 {\isasymrbrace}
}

\quad \quad \quad \quad \quad 
\specrg{\{V4\} \speccomment{{\isacharparenleft}V4 {\isasymin} mergeblock{\isacharunderscore}loopinv\ V3 \ t\ b\ {\isasymalpha}\ {\isasyminter}\ {\isasymlbrace} {\isasymalpha} {\isachargreater} 0 {\isasymrbrace} {\isacharparenright}}}

\quad \quad \quad \quad \quad {\isasymacute}bb\ {\isacharcolon}{\isacharequal}\ {\isasymacute}bb\ {\isacharparenleft}t\ {\isacharcolon}{\isacharequal}\ {\isacharparenleft}{\isasymacute}bn\ t\ div\ {\isadigit{4}}{\isacharparenright}\ {\isacharasterisk}\ {\isadigit{4}}\ {\isacharplus}\ {\isasymacute}i\ t{\isacharparenright}{\isacharsemicolon}{\isacharsemicolon}

\quad \quad \quad \quad \quad 
\specrg{\{V5\} \speccomment{{\isacharparenleft}V5 {\isasymequiv}\ V4{\isasymlparr}bb\ {\isacharcolon}{\isacharequal}\ {\isacharparenleft}bb\ V{\isacharparenright}\ {\isacharparenleft}t{\isacharcolon}{\isacharequal}{\isacharparenleft}bn\ V4\ t\ div\ {\isadigit{4}}{\isacharparenright}\ {\isacharasterisk}\ {\isadigit{4}}\ {\isacharplus}\ i\ V4\ t{\isacharparenright}{\isasymrparr} {\isacharparenright}}}

\quad \quad \quad \quad \quad {\isasymacute}mem{\isacharunderscore}pool{\isacharunderscore}info\ {\isacharcolon}{\isacharequal}\ set{\isacharunderscore}bit{\isacharunderscore}noexist\ {\isasymacute}mem{\isacharunderscore}pool{\isacharunderscore}info\ {\isacharparenleft}pool\ b{\isacharparenright}\ {\isacharparenleft}{\isasymacute}lvl\ t{\isacharparenright}\ {\isacharparenleft}{\isasymacute}bb\ t{\isacharparenright}{\isacharsemicolon}{\isacharsemicolon}

\quad \quad \quad \quad \quad 
\specrg{\{V6\} \speccomment{{\isacharparenleft}V6 {\isasymequiv}\ V5{\isasymlparr} mem{\isacharunderscore}pool{\isacharunderscore}info\ {\isacharcolon}{\isacharequal}\isanewline
\quad \quad \quad \quad \quad \quad \quad \quad \quad \quad \quad 
set{\isacharunderscore}bit{\isacharunderscore}noexist\ (mem{\isacharunderscore}pool{\isacharunderscore}info\ V5)\ (pool\ b)\ (lvl\ V5\ t)\ (bb\ V5\ t) {\isasymrparr} {\isacharparenright}}}

\quad \quad \quad \quad \quad {\isasymacute}block{\isacharunderscore}pt\ {\isacharcolon}{\isacharequal}\ {\isasymacute}block{\isacharunderscore}pt\ {\isacharparenleft}t\ {\isacharcolon}{\isacharequal}\ block{\isacharunderscore}ptr\ {\isacharparenleft}{\isasymacute}mem{\isacharunderscore}pool{\isacharunderscore}info\ {\isacharparenleft}pool\ b{\isacharparenright}{\isacharparenright}\ {\isacharparenleft}{\isasymacute}lsz\ t{\isacharparenright}\ {\isacharparenleft}{\isasymacute}bb\ t{\isacharparenright}{\isacharparenright}{\isacharsemicolon}{\isacharsemicolon}

\quad \quad \quad \quad \quad 
\specrg{\{V7\} \speccomment{{\isacharparenleft}V7 {\isasymequiv}\ V6{\isasymlparr}block{\isacharunderscore}pt\ {\isacharcolon}{\isacharequal}\ {\isacharparenleft}block{\isacharunderscore}pt\ V6{\isacharparenright}\isanewline
\quad \quad \quad \quad \quad \quad \quad \quad \quad \quad  {\isacharparenleft}t{\isacharcolon}{\isacharequal}block{\isacharunderscore}ptr\ {\isacharparenleft}mem{\isacharunderscore}pool{\isacharunderscore}info\ V6\ {\isacharparenleft}pool\ b{\isacharparenright}{\isacharparenright}\ {\isacharparenleft}lsz\ V6\ t{\isacharparenright}\ {\isacharparenleft}bb\ V6\ t{\isacharparenright}{\isacharparenright}{\isasymrparr} {\isacharparenright}}}

\quad \quad \quad \quad \quad \isacommand{IF}\ {\isasymacute}bn\ t\ {\isasymnoteq}\ {\isasymacute}bb\ t\ {\isasymand}\ block{\isacharunderscore}fits\ {\isacharparenleft}{\isasymacute}mem{\isacharunderscore}pool{\isacharunderscore}info\ {\isacharparenleft}pool\ b{\isacharparenright}{\isacharparenright}\ {\isacharparenleft}{\isasymacute}block{\isacharunderscore}pt\ t{\isacharparenright}\ {\isacharparenleft}{\isasymacute}lsz\ t{\isacharparenright}\ \isacommand{THEN}\isanewline
\quad \quad \quad \quad \quad \quad {\isasymacute}mem{\isacharunderscore}pool{\isacharunderscore}info\ {\isacharcolon}{\isacharequal}\ {\isasymacute}mem{\isacharunderscore}pool{\isacharunderscore}info\ {\isacharparenleft}{\isacharparenleft}pool\ b{\isacharparenright}\ {\isacharcolon}{\isacharequal}\ \isanewline
\quad \quad \quad \quad \quad \quad \quad \quad remove{\isacharunderscore}free{\isacharunderscore}list\ {\isacharparenleft}{\isasymacute}mem{\isacharunderscore}pool{\isacharunderscore}info\ {\isacharparenleft}pool\ b{\isacharparenright}{\isacharparenright}\ {\isacharparenleft}{\isasymacute}lvl\ t{\isacharparenright}\ {\isacharparenleft}{\isasymacute}block{\isacharunderscore}pt\ t{\isacharparenright}{\isacharparenright}\isanewline
\quad \quad \quad \quad \quad \isacommand{FI}\isanewline
\quad \quad \quad \quad \isacommand{ROF}{\isacharsemicolon}{\isacharsemicolon}

\quad \quad \quad \quad \specrg{
\isacommand{mergeblock{\isacharunderscore}loopinv}\ V3 \ t\ b\ {\isasymalpha}\ {\isasyminter}\ {\isasymlbrace} {\isasymalpha} = 0 {\isasymrbrace}
}

\quad \quad \quad \quad {\isasymacute}lvl\ {\isacharcolon}{\isacharequal}\ {\isasymacute}lvl\ {\isacharparenleft}t\ {\isacharcolon}{\isacharequal}\ {\isasymacute}lvl\ t\ {\isacharminus}\ {\isadigit{1}}{\isacharparenright}{\isacharsemicolon}{\isacharsemicolon}\isanewline
\quad \quad \quad \quad {\isasymacute}bn\ {\isacharcolon}{\isacharequal}\ {\isasymacute}bn\ {\isacharparenleft}t\ {\isacharcolon}{\isacharequal}\ {\isasymacute}bn\ t\ div\ {\isadigit{4}}{\isacharparenright}{\isacharsemicolon}{\isacharsemicolon}\isanewline
\quad \quad \quad \quad {\isasymacute}mem{\isacharunderscore}pool{\isacharunderscore}info\ {\isacharcolon}{\isacharequal}\ set{\isacharunderscore}bit{\isacharunderscore}freeing\ {\isasymacute}mem{\isacharunderscore}pool{\isacharunderscore}info\ {\isacharparenleft}pool\ b{\isacharparenright}\ {\isacharparenleft}{\isasymacute}lvl\ t{\isacharparenright}\ {\isacharparenleft}{\isasymacute}bn\ t{\isacharparenright}{\isacharsemicolon}{\isacharsemicolon}\isanewline
\quad \quad \quad \quad {\isasymacute}freeing{\isacharunderscore}node\ {\isacharcolon}{\isacharequal}\ {\isasymacute}freeing{\isacharunderscore}node\ {\isacharparenleft}t\ {\isacharcolon}{\isacharequal}\ Some\ {\isasymlparr}pool\ {\isacharequal}\ {\isacharparenleft}pool\ b{\isacharparenright}{\isacharcomma}\ level\ {\isacharequal}\ {\isacharparenleft}{\isasymacute}lvl\ t{\isacharparenright}{\isacharcomma}\ \isanewline
\quad \ \ \ \ \ \ \ \ \ \ \ \ \ \ \ \ \ \ \ \ block\ {\isacharequal}\ {\isacharparenleft}{\isasymacute}bn\ t{\isacharparenright}{\isacharcomma}\ 
data\ {\isacharequal}\ block{\isacharunderscore}ptr\ {\isacharparenleft}{\isasymacute}mem{\isacharunderscore}pool{\isacharunderscore}info\ {\isacharparenleft}pool\ b{\isacharparenright}{\isacharparenright}\ \isanewline
\quad \ \ \ \ \ \ \ \ \ \ \ \ \ \ \ \ \ \ \ \ \ \  {\isacharparenleft}{\isacharparenleft}{\isacharparenleft}ALIGN{\isadigit{4}}\ {\isacharparenleft}max{\isacharunderscore}sz\ {\isacharparenleft}{\isasymacute}mem{\isacharunderscore}pool{\isacharunderscore}info\ {\isacharparenleft}pool\ b{\isacharparenright}{\isacharparenright}{\isacharparenright}{\isacharparenright}\ div\ {\isacharparenleft}{\isadigit{4}}\ {\isacharcircum}\ {\isacharparenleft}{\isasymacute}lvl\ t{\isacharparenright}{\isacharparenright}{\isacharparenright}{\isacharparenright}\ 
{\isacharparenleft}{\isasymacute}bn\ t{\isacharparenright}\ {\isasymrparr}{\isacharparenright}\isanewline
\quad \quad \quad \isacommand{ELSE}

\quad \quad \quad \quad \specrg{\{V3\}\ {\isasyminter}\ {\isacharminus}\ {\isasymlbrace}NULL\ {\isacharless}\ {\isasymacute}lvl\ t\ {\isasymand}\ partner{\isacharunderscore}bits\ {\isacharparenleft}{\isasymacute}mem{\isacharunderscore}pool{\isacharunderscore}info\ {\isacharparenleft}pool\ b{\isacharparenright}{\isacharparenright}\ {\isacharparenleft}{\isasymacute}lvl\ t{\isacharparenright}\ {\isacharparenleft}{\isasymacute}bn\ t{\isacharparenright}{\isasymrbrace}}

\quad \quad \quad \quad \isacommand{IF}\ block{\isacharunderscore}fits\ {\isacharparenleft}{\isasymacute}mem{\isacharunderscore}pool{\isacharunderscore}info\ {\isacharparenleft}pool\ b{\isacharparenright}{\isacharparenright}\ {\isacharparenleft}{\isasymacute}blk\ t{\isacharparenright}\ {\isacharparenleft}{\isasymacute}lsz\ t{\isacharparenright}\ \isacommand{THEN}\isanewline
\quad \quad \quad \quad \quad {\isasymacute}mem{\isacharunderscore}pool{\isacharunderscore}info\ {\isacharcolon}{\isacharequal}\ {\isasymacute}mem{\isacharunderscore}pool{\isacharunderscore}info\ {\isacharparenleft}{\isacharparenleft}pool\ b{\isacharparenright}\ {\isacharcolon}{\isacharequal}\ \isanewline
\quad \quad \quad \quad \quad \quad append{\isacharunderscore}free{\isacharunderscore}list\ {\isacharparenleft}{\isasymacute}mem{\isacharunderscore}pool{\isacharunderscore}info\ {\isacharparenleft}pool\ b{\isacharparenright}{\isacharparenright}\ {\isacharparenleft}{\isasymacute}lvl\ t{\isacharparenright}\ {\isacharparenleft}{\isasymacute}blk\ t{\isacharparenright}\ {\isacharparenright}\isanewline
\quad \quad \quad \quad \isacommand{FI}{\isacharsemicolon}{\isacharsemicolon}\isanewline
\quad \quad \quad \quad {\isasymacute}free{\isacharunderscore}block{\isacharunderscore}r\ {\isacharcolon}{\isacharequal}\ {\isasymacute}free{\isacharunderscore}block{\isacharunderscore}r\ {\isacharparenleft}t\ {\isacharcolon}{\isacharequal}\ False{\isacharparenright}\isanewline
\quad \quad \quad \isacommand{FI}\isanewline
\quad \quad \isacommand{END} \ \speccomment{(* END\ of\ ATOM *)} \isanewline
\quad \isacommand{OD} \ \speccomment{(* END\ of\ WHILE \ free\_block\_r \ DO  *)} 

\quad \specrg{\isacommand{mp{\isacharunderscore}free{\isacharunderscore}precond{\isadigit{9}}}\ t\ b\ {\isasymequiv}\ Mem{\isacharunderscore}pool{\isacharunderscore}free{\isacharunderscore}pre\ t\ {\isasyminter}\ {\isasymlbrace} g {\isasymrbrace}}

\quad \speccomment{{\isacharparenleft}{\isacharasterisk}\ {\isacharequal} {\isacharequal} {\isacharequal} end of {\isacharcolon}\ free{\isacharunderscore}block{\isacharparenleft}pool{\isacharcomma}\ level{\isacharcomma}\ lsizes{\isacharcomma}\ block{\isacharparenright}{\isacharsemicolon}\ {\isacharequal} {\isacharequal} {\isacharequal}{\isacharasterisk}{\isacharparenright}}

\quad t\ \isactrlenum \ \isacommand{ATOMIC}

\quad \specrg{\{Va\} {\isacharparenleft}Va {\isasymin} mp{\isacharunderscore}free{\isacharunderscore}precond{\isadigit{9}} \ t\ b\ {\isasyminter} {\isasymlbrace}{\isasymacute}cur\ {\isacharequal}\ Some\ t{\isasymrbrace}{\isacharparenright}}

\quad  \specrg{
stm{\isadigit{9}}{\isacharunderscore}loopinv\ Va\ t\ b\ {\isasymalpha}\ {\isasymequiv}\isanewline
\quad \quad {\isacharbraceleft}V{\isachardot}\ inv\ V\ {\isasymand}\ cur\ V\ {\isacharequal}\ cur\ Va\ {\isasymand}\ tick\ V\ {\isacharequal}\ tick\ Va\ {\isasymand}\ {\isacharparenleft}V{\isacharcomma}Va{\isacharparenright}{\isasymin}gvars{\isacharunderscore}conf{\isacharunderscore}stable\ \isanewline
\quad \quad \quad {\isasymand}\ freeing{\isacharunderscore}node\ V\ t\ {\isacharequal}\ freeing{\isacharunderscore}node\ Va\ t\ {\isasymand}\ allocating{\isacharunderscore}node\ V\ t\ {\isacharequal}\ allocating{\isacharunderscore}node\ Va\ t\isanewline
\quad \quad \quad {\isasymand}\ {\isacharparenleft}{\isasymforall}p{\isachardot}\ levels\ {\isacharparenleft}mem{\isacharunderscore}pool{\isacharunderscore}info\ V\ p{\isacharparenright}\ {\isacharequal}\ levels\ {\isacharparenleft}mem{\isacharunderscore}pool{\isacharunderscore}info\ Va\ p{\isacharparenright}{\isacharparenright}\isanewline
\quad \quad \quad {\isasymand}\ {\isacharparenleft}{\isasymforall}p{\isachardot}\ p\ {\isasymnoteq}\ pool\ b\ {\isasymlongrightarrow}\ mem{\isacharunderscore}pool{\isacharunderscore}info\ V\ p\ {\isacharequal}\ mem{\isacharunderscore}pool{\isacharunderscore}info\ Va\ p{\isacharparenright}\isanewline
\quad \quad \quad {\isasymand}\ {\isacharparenleft}{\isasymforall}t{\isacharprime}{\isachardot}\ t{\isacharprime}\ {\isasymnoteq}\ t\ {\isasymlongrightarrow}\ lvars{\isacharunderscore}nochange\ t{\isacharprime}\ V\ Va{\isacharparenright}\isanewline
\quad \quad \quad {\isasymand} {\isasymalpha}\ {\isacharequal} length\ {\isacharparenleft}wait{\isacharunderscore}q\ {\isacharparenleft}{\isasymacute}mem{\isacharunderscore}pool{\isacharunderscore}info\ {\isacharparenleft}pool\ b{\isacharparenright}{\isacharparenright}{\isacharparenright}
{\isacharbraceright}
}

\quad \quad \isacommand{WHILE}\ wait{\isacharunderscore}q\ {\isacharparenleft}{\isasymacute}mem{\isacharunderscore}pool{\isacharunderscore}info\ {\isacharparenleft}pool\ b{\isacharparenright}{\isacharparenright}\ {\isasymnoteq}\ {\isacharbrackleft}{\isacharbrackright}\ \isacommand{DO}\

\quad \quad \quad \specrg{
stm{\isadigit{9}}{\isacharunderscore}loopinv\ Va\ t\ b\ {\isasymalpha}\ {\isasyminter}\ {\isasymlbrace} {\isasymalpha} {\isachargreater} 0 {\isasymrbrace}
}

\quad \quad \quad {\isasymacute}th\ {\isacharcolon}{\isacharequal}\ {\isasymacute}th\ {\isacharparenleft}t\ {\isacharcolon}{\isacharequal}\ hd\ {\isacharparenleft}wait{\isacharunderscore}q\ {\isacharparenleft}{\isasymacute}mem{\isacharunderscore}pool{\isacharunderscore}info\ {\isacharparenleft}pool\ b{\isacharparenright}{\isacharparenright}{\isacharparenright}{\isacharparenright}{\isacharsemicolon}{\isacharsemicolon}

\quad \quad \quad \speccomment{{\isacharparenleft}{\isacharasterisk}\ {\isacharunderscore}unpend{\isacharunderscore}thread{\isacharparenleft}th{\isacharparenright}{\isacharsemicolon}\ {\isacharasterisk}{\isacharparenright}}

\quad \quad \quad {\isasymacute}mem{\isacharunderscore}pool{\isacharunderscore}info\ {\isacharcolon}{\isacharequal}\ {\isasymacute}mem{\isacharunderscore}pool{\isacharunderscore}info\ {\isacharparenleft}pool\ b\ {\isacharcolon}{\isacharequal}\ {\isasymacute}mem{\isacharunderscore}pool{\isacharunderscore}info\ {\isacharparenleft}pool\ b{\isacharparenright}

\quad \quad \quad \quad \quad {\isasymlparr}wait{\isacharunderscore}q\ {\isacharcolon}{\isacharequal}\ tl\ {\isacharparenleft}wait{\isacharunderscore}q\ {\isacharparenleft}{\isasymacute}mem{\isacharunderscore}pool{\isacharunderscore}info\ {\isacharparenleft}pool\ b{\isacharparenright}{\isacharparenright}{\isacharparenright}{\isasymrparr}{\isacharparenright}{\isacharsemicolon}{\isacharsemicolon}

\quad \quad \quad \speccomment{{\isacharparenleft}{\isacharasterisk}\ {\isacharunderscore}ready{\isacharunderscore}thread{\isacharparenleft}th{\isacharparenright}{\isacharsemicolon}\ {\isacharasterisk}{\isacharparenright}}

\quad \quad \quad {\isasymacute}thd{\isacharunderscore}state\ {\isacharcolon}{\isacharequal}\ {\isasymacute}thd{\isacharunderscore}state\ {\isacharparenleft}{\isasymacute}th\ t\ {\isacharcolon}{\isacharequal}\ READY{\isacharparenright}{\isacharsemicolon}{\isacharsemicolon}

\quad \quad \quad {\isasymacute}need{\isacharunderscore}resched\ {\isacharcolon}{\isacharequal}\ {\isasymacute}need{\isacharunderscore}resched{\isacharparenleft}t\ {\isacharcolon}{\isacharequal}\ True{\isacharparenright}

\quad \quad  \isacommand{OD}{\isacharsemicolon}{\isacharsemicolon}

\quad \quad \specrg{
stm{\isadigit{9}}{\isacharunderscore}loopinv\ Va\ t\ b\ {\isasymalpha}\ {\isasyminter}\ {\isasymlbrace} {\isasymalpha} = 0 {\isasymrbrace}
}

\quad \quad \isacommand{IF}\ {\isasymacute}need{\isacharunderscore}resched\ t\ \isacommand{THEN}\isanewline
\quad \quad \quad reschedule \quad \speccomment{(* \_reschedule\_threads(key) *)}\isanewline
\quad \quad \isacommand{FI}\isanewline
\quad \isacommand{END} \ \speccomment{(* END\ of\ ATOM *)}\isanewline
\isacommand{END}{\isachardoublequoteclose}

\specrg{\isacommand{Mem{\isacharunderscore}pool{\isacharunderscore}free{\isacharunderscore}post}\ t\ {\isasymequiv}\ {\isasymlbrace} {\isasymacute}inv\ {\isasymand}\ {\isasymacute}allocating{\isacharunderscore}node\ t\ {\isacharequal}\ None\ {\isasymand}\ {\isasymacute}freeing{\isacharunderscore}node\ t\ {\isacharequal}\ None{\isasymrbrace}}

\end{isabellec}

%
% If your work has an appendix, this is the place to put it.
\end{document}